\documentclass[sn-mathphys,Numbered]{sn-jnl}

\usepackage{graphicx}
\usepackage{multirow}
\usepackage{amsmath,amssymb,amsfonts}
\usepackage{amsthm}
\usepackage{mathrsfs}
\usepackage{mathtools}
\usepackage[title]{appendix}
\usepackage{xcolor}
\usepackage{textcomp}
\usepackage{manyfoot}
\usepackage{booktabs}
\usepackage[ruled,vlined,linesnumbered]{algorithm2e}
\usepackage{algorithmicx}
\usepackage{algpseudocode}
\usepackage{listings}
\usepackage{caption}
\usepackage{subcaption}

\newtheorem{theorem}{Theorem} 
 
\newtheorem{conjecture}[theorem]{Conjecture}

\newtheorem{lemma}{Lemma}
\newtheorem{definition}{Definition}
\newtheorem{observation}{Observation}
\newtheorem{corollary}{Corollary}
\begin{document}

\title[Pareto Sums of Pareto Sets]{Pareto Sums of Pareto Sets: Lower Bounds and Algorithms}

\author[1]{\fnm{Daniel} \sur{Funke}}\email{daniel.funke@kit.edu}

\author[1]{\fnm{Demian} \sur{Hespe}}\email{demian.hespe@outlook.com}

\author[1]{\fnm{Peter} \sur{Sanders}}\email{sanders@kit.edu}

\author[2]{\fnm{Sabine} \sur{Storandt}}\email{sabine.storandt@uni-konstanz.de}

\author*[2]{\fnm{Carina} \sur{Truschel}}\email{carina.truschel@uni-konstanz.de}

\affil[1]{\orgname{KIT}, \orgaddress{\street{Am Fasanengarten}, \city{Karlsruhe}, \postcode{76128}, \country{Germany}}}

\affil[2]{\orgname{University of Konstanz}, \orgaddress{\street{Universit\"atsstraße}, \city{Konstanz}, \postcode{78464}, \country{Germany}}}

\abstract{In  bi-criteria optimization problems, the goal is  typically to compute the set of Pareto-optimal solutions. Many algorithms for these types of problems rely on efficient merging or combining of partial solutions and filtering of dominated solutions in the resulting sets. In this article, we  consider the  task of computing the Pareto sum of two given Pareto sets $A, B$ of size $n$. The Pareto sum $C$   contains all non-dominated points of the Minkowski sum $M = \{a+b|a \in A, b\in B\}$. Since the Minkowski sum has a size of $n^2$, but the Pareto sum $C$ can be much smaller, the goal is to compute $C$ without having to compute and store all of $M$. We present several new algorithms for efficient Pareto sum computation, including an output-sensitive successive algorithm with a running time of $\mathcal{O}(n \log n + nk)$ and a space consumption of $\mathcal{O}(n+k)$ for $k=|C|$. If the elements of $C$ are streamed, the space consumption reduces to $\mathcal{O}(n)$.
For output sizes $k \geq 2n$, we prove a conditional lower bound for Pareto sum computation, which excludes running times in $\mathcal{O}(n^{2-\delta})$ for $\delta > 0$ unless the (min,+)-convolution hardness conjecture fails. The successive algorithm matches this lower bound for $k \in \Theta(n)$. However, for $k \in \Theta(n^2)$, the successive algorithm exhibits a cubic running time.
But we also present an algorithm with an output-sensitive space consumption and a running time of $\mathcal{O}(n^2 \log n)$, which matches the lower bound up to a logarithmic factor even for large $k$. Furthermore, we describe suitable engineering techniques to improve the practical running times of our algorithms. Finally, we provide an extensive comparative experimental study on generated and real-world data. As a showcase application, we consider preprocessing-based  bi-criteria route planning in road networks.  Pareto sum computation is the bottleneck task in the preprocessing phase and in the query phase. We show that  using  our algorithms with an output-sensitive space consumption allows to tackle larger instances and   reduces the preprocessing and query time   compared to algorithms that fully store $M$.}

\maketitle

\section{Introduction}\label{sec1}
Solving multi-objective combinatorial optimization problems demands to find the set of non-dominated solutions, also referred to as skyline, Pareto frontier or Pareto set. To solve problem instances of substantial size, solution approaches often rely on efficient  combination and filtering of partial solutions. 
In particular, non-dominance filtering of unions or Minkowski sums of 
intermediate Pareto sets occur as a frequent subtasks in optimization algorithms. Examples include decomposition approaches for multi-objective integer programming \cite{schulze2019multi}, dynamic programming methods for multi-objective knapsack \cite{ehrgott2000survey}, bi-directional search algorithms for multi-criteria shortest path problems \cite{artigues2013state}, or Pareto local search for multi-objective set cover \cite{lust2014variable}.

\begin{figure}
    \centering
    \includegraphics[height=6cm]{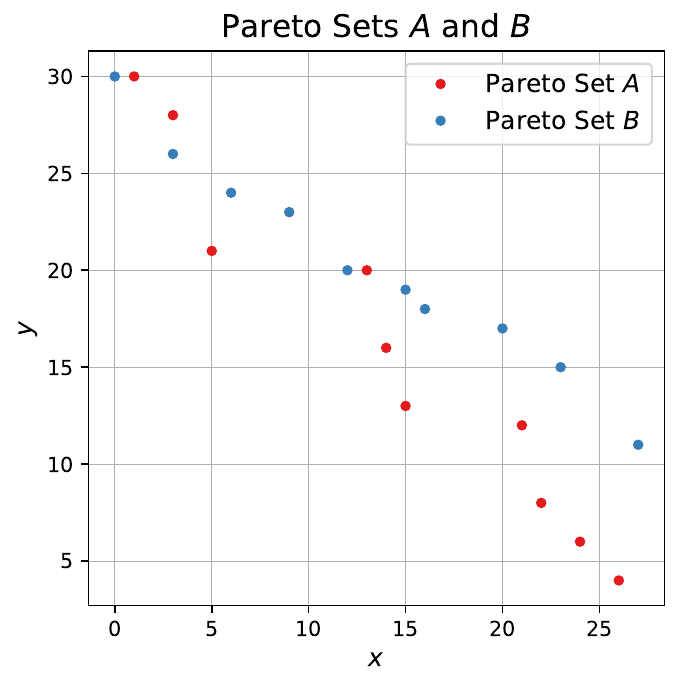}\hfill
     \includegraphics[height=6cm]{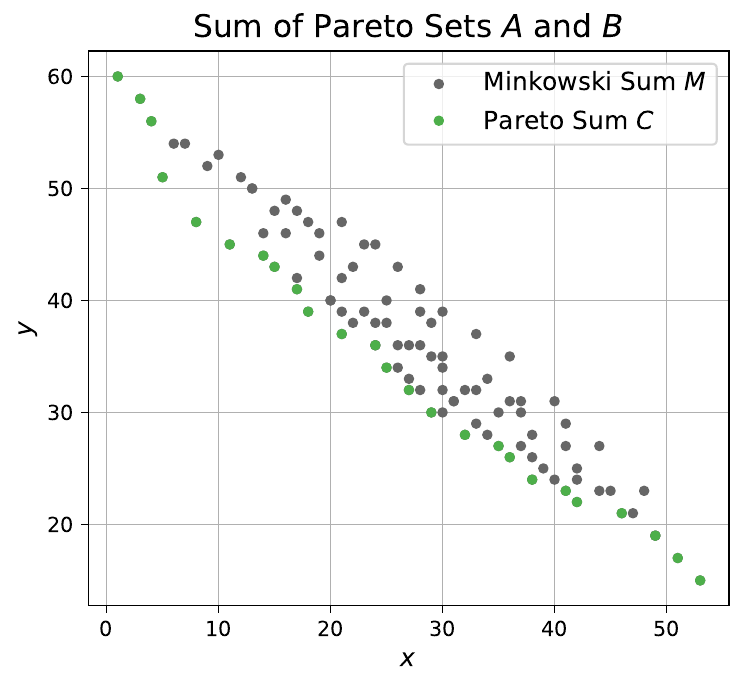}
     \footnotesize
    \begin{tabular}{lcccccccccc}
                         & $A_1$    & $A_2$                           & $A_3$    & $A_4$    & $A_5$                           & $A_6$    & $A_7$                           & $A_8$                           & $A_9$ & $A_{10}$ \\
$B_1$ & {\color[HTML]{049600}\ 1, 60} & {\color[HTML]{049600}\ 3, 58}  & {\color[HTML]{049600}\ 5, 51}  & 13, 50 & 14, 46 & {\color[HTML]{049600} 15, 43} & 21, 42 & 22, 38                        & {\color[HTML]{049600} 24, 36} & 26, 34                        \\
$B_2$ & {\color[HTML]{049600}\ 4, 56} &\ 6, 54  & {\color[HTML]{049600}\ 8, 47}  & 16, 46 & 17, 42 & {\color[HTML]{049600} 18, 39} & 24, 38 & {\color[HTML]{049600} 25, 34} & {\color[HTML]{049600} 27, 32} & {\color[HTML]{049600} 29, 30} \\
$B_3$ &\ 7, 54                        &\ 9, 52  & {\color[HTML]{049600} 11, 45} & 19, 44 & 20, 40 & {\color[HTML]{049600} 21, 37} & 27, 36 & 28, 32                        & 30, 30                        & {\color[HTML]{049600} 32, 28} \\
$B_4$ & 10, 53                        & 12, 51 & {\color[HTML]{049600} 14, 44} & 22, 43 & 23, 39 & {\color[HTML]{049600} 24, 36} & 30, 35 & 31, 31                        & 33, 29                        & {\color[HTML]{049600} 35, 27} \\
$B_5$ & 13, 50                       & 15, 48 & {\color[HTML]{049600} 17, 41} & 25, 40 & 26, 36 & 27, 33                        & 33, 32 & 34, 28                        & {\color[HTML]{049600} 36, 26} & {\color[HTML]{049600} 38, 24} \\
$B_6$ & 16, 49                       & 18, 47 & 20, 40                        & 28, 39 & 29, 35 & 30, 32                        & 36, 31 & 37, 27                        & 39, 25                        & {\color[HTML]{049600} 41, 23} \\
$B_7$ & 17, 48                       & 19, 46 & 21, 39                        & 29, 38 & 30, 34 & 31, 31                        & 37, 30 & 38, 26                        & 40, 24                        & {\color[HTML]{049600} 42, 22} \\
$B_8$ & 21, 47                       & 23, 45 & 25, 38                        & 33, 37 & 34, 33 & 35, 30                        & 41, 29 & 42, 35                        & 44, 23                        & {\color[HTML]{049600} 46, 21} \\
$B_9$ & 24, 45                       & 26, 43 & 28, 36                        & 36, 35 & 37, 31 & 38, 28                        & 44, 27 & 45, 23                        & 47, 21                        & {\color[HTML]{049600} 49, 19} \\
$B_{10}$ & 28, 41                       & 30, 39 & 32, 32                        & 40, 31 & 41, 27 & 42, 24                        & 48, 23 & {\color[HTML]{049600} 49, 19} & {\color[HTML]{049600} 51, 17} & {\color[HTML]{049600} 53, 15}
\end{tabular}

    \caption{Example instance with input Pareto sets $A, B$ of size 10. The Minkowski sum has 100 elements. The Pareto sum $C$ consists of 27 elements (marked green in the plot as well as in the matrix representation).}
    \label{fig:teaser}
\end{figure}
In this article, we focus on the efficient computation of the filtered Minkowski sum of   two-dimensional Pareto sets $A, B$. The Minkowski sum $M$ is defined as the set of elements derived from pairwise addition of elements in $A$ and $B$. However, the Minkowski sum often contains many dominated elements, see Figure \ref{fig:teaser} for an example. Indeed, it was proven in \cite{klamroth2022efficient} that for $A,B$ of size $n$, the set of non-dominated elements in $M$ -- which we refer to as Pareto sum of $A, B$ -- might have a size in $o(n)$. Thus, algorithms that first compute all elements of $M$ and subsequently apply non-dominance filtering might be unnecessarily wasteful as they come with a running time and space consumption in $\Omega(n^2)$. Our goal is to design practical algorithms for Pareto sum computation with output-sensitive space consumption, and to evaluate their performance on realistic inputs. 

\begin{figure}
    \centering
    \includegraphics[width=\textwidth]{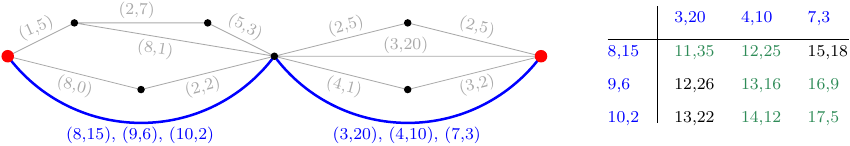}
    \caption{Left: Example road network (gray edges) with bi-criteria edge costs. The blue edges indicate shortcut edges. Each shortcut has a set of cost vectors that encodes the aggregated costs of all Pareto-optimal paths between its endpoints. Right: Inserting a new shortcut between the two red points requires to compute the Pareto sum of the two cost vector sets of the blue shortcuts. In this example, the Minkowski sum contains 9 elements out of which 6 form the Pareto sum (marked green in the table). Note that same computation is necessary in a bi-directional bi-criteria path search from the two red points at the moment at which the node in the middle is visited from both directions. }
    \label{fig:bicritroute}
\end{figure}
As one particular use case of our methods, we will consider the bi-criteria route planning problem in road networks. There exists a plethora of algorithms to compute the set of  Pareto-optimal paths between a given source and a target node in network where edges have two costs, see e.g. \cite{erb2014parallel,ahmadi2021bi}. The currently fastest methods rely on preprocessing.  In particular, variants of contraction hierarchies (CH) have been proven to be very useful in this context \cite{storandt2012route,zhang2023efficient}.  In a CH, the input graph is augmented with so called shortcut edges that represent sets of Pareto-optimal paths between their end points. The shortcuts store the costs of these paths in the form of  Pareto sets. On query time, shortcuts are used to decrease the search space size of a Pareto-Dijkstra run, resulting in significantly faster query times and reduced space consumption. In the preprocessing phase, the  shortcuts are inserted incrementally. The base operation is to concatenate two shortcut or original edges $e=\{u,v\}$ and $e'=\{v,w\}$  to form a new shortcut $\{u,w\}$. The Pareto set of the new shortcut is the set of non-dominated elements in the Minkowski sum of the Pareto sets corresponding to $e$ and $e'$, see Figure \ref{fig:bicritroute}. Thus, the preprocessing time crucially depends on an efficient Pareto sum computation. In \cite{funke2015personalized}, it was discussed that  computing the Minkowski sum and  filtering all dominated elements in a naive fashion is too time-consuming. Therefore, filtering strategies were proposed that prune dominated elements. However, these strategies are not guaranteed to retrieve the Pareto sum but usually produce a superset thereof. Keeping supersets slows down later stages of the preprocessing as well as query answering.

In this article, we propose novel algorithms that allow for fast and exact Pareto sum computation.

\section{Related Work}
Non-dominance filtering in point sets is a well-studied task in computational geometry  \cite{chen2012maxima}, also referred to as skyline or maxima computation. 
There  exist  output-sensitive algorithms for the two-dimensional case as e.g. the one proposed by Kirkpatrick and Seidel \cite{kirkpatrick1985output} with a running time of $\mathcal{O}(N \log k)$ where $N$ denotes the size of the point set and $k$ the size of the skyline.  The basic idea is to first partition the input points into $k$ sets of size $\approx N/k$ with non-overlapping ranges with respect to their $x$-coordinates. Then, the sets are processed individually in sorted order. As $k$ is typically not known beforehand, a more intricate version of the algorithm allows to achieve the same asymptotic running time by starting with a coarse partition  and  refining it on demand as soon as a certain number of non-dominated points are identified. With a worst-case running time of $\mathcal{O}(N \log N)$ and close-to-linear running time for small $k$, this algorithm seems to be well-suited for Pareto sum computation. However, in our application we have $N=|M|$ where $M$ is the Minkowski sum of the input sets $A,B$; and any approach that relies on materializing  $M$ is bound to a running time and space consumption in $\Omega(n^2)$. Using the interpretation of input elements as two-dimensional points, Pareto sum computation can also be reduced to computing the Minkowski sum of  the  orthogonal hulls of  $A$ and $B$ (where both sets are augmented with a dummy point based on the maximum coordinate values in the respective set). The Minkowski sum of two convex polygons can be computed in linear time \cite{mark2008computational}. For non-convex inputs, the polygons are first decomposed into convex subpolygons  and then  the linear time algorithm is applied to all pairs of subpolygons. Finally, the union of all partial results is computed. The running time  depends on the applied decomposition technique and the number and complexity of the resulting subpolygons \cite{agarwal2002polygon}. Considering orthogonal convex hulls of size $n$, their convex decomposition cannot contain fewer than $n$ subpolygons, and thus the approach needs to compute the  union of   $\Theta(n^2)$   partial solutions.

Recently, new algorithms for Pareto sum computation with the potential to achieve  subquadratic running time and space consumption were proposed in \cite{klamroth2022efficient}. The so called NonDomDC algorithm exploits the structure of the matrix that represents the Minkowski sum (see Figure \ref{fig:teaser}). In particular, it makes use of the fact that columns in the matrix are Pareto sets themselves. Assuming the Pareto sum $P_i$ of elements occurring in the first $i$ columns is known, $P_{i+1}$ can be computed by merging $P_i$ and column $i+1$ and pruning dominated elements in $\mathcal{O}(|P_i| + n)$ time. Thus, with $P := \max_{i=1}^n |P_i|$ denoting the maximum size of an intermediate solution, the total running time  is in $\mathcal{O}(Pn)$ and the space consumption is in $\mathcal{O}(n+P)$. However, this does not constitute an output-sensitive algorithm as the size of the intermediate Pareto sum can be significantly larger than the final result size. So even for small $k$, the algorithm might have  cubic running time and  quadratic space consumption.  However, their experimental study demonstrates good performance in practice. Similar methods, as the box-based method proposed in  \cite{kerberenes2022computing}, were shown to be outperformed.

There is also work on the heuristic dominance pruning of the Minkowski sum or general point sets, where the goal is to reduce the set size quickly but the resulting set might still be a superset of the underlying Pareto set.
Shu et al. \cite{shu2023two} presented a two-phase algorithm for efficient dominance pruning in large sets, with the goal to make evolutionary multi-objective optimization more scalable. Funke et al. \cite{funke2015personalized} presented several Minkowski sum pruning methods to make personalized route planning with two- to ten-dimensional cost vectors viable.
The higher-dimensional exact version of the problem was studied in \cite{preparata2012computational,gomes2018boosting, klamroth2022efficient}. The respective methods apply the divide \& conquer paradigm to the elements in the Minkowski sum as well as to the dimension. Thus, they also rely on fast methods for the two-dimensional case.

\section{Contribution}
\begin{table}[b]
    \centering
        \caption{Running time and space consumption of different algorithms for Pareto sum computation. The input size is denoted by $n$ and the output size by $k$.}
    \label{tab:overview}
    \begin{tabular}{|ll|ll|l|}
    \hline
     algorithm  &  & running  time & space consumption &\\
      \hline    
      NonDomDC                 & (ND)  & $\mathcal{O}(n^3)$      & $\mathcal{O}(n^2)$ & \cite{klamroth2022efficient}\\
    Kirkpatrick-Seidel & (KS)    & $\mathcal{O}(n^2 \log k)$ & $\Theta(n^2)$ & \cite{kirkpatrick1985output}\\
        \hline 
    Brute Force &(BF)      & $\mathcal{O}(n^4)$  & $\mathcal{O}(n+k)$  & \ref{sec:BF}\\
    Binary Search &(BS) & $\mathcal{O}(n^3 \log n)$  & $\mathcal{O}(n+k)$ & \ref{sec:BS}\\
    Sort \& Compare &(SC) & $\mathcal{O}(n^2 \log n)$  & $\mathcal{O}(n+k)$ & \ref{sec:SC}\\ 
    Successive Binary Search &(SBS) & $\mathcal{O}(nk \log n)$  & $\mathcal{O}(n+k)$ & \ref{sec:CBS}\\ 
    Successive Sweep  Search &(SSS) & $\mathcal{O}(n \log n + nk)$  & $\mathcal{O}(n+k)$ & \ref{sec:CSS}\\ 
    \hline
    \end{tabular}
\end{table}
In this article, we consider the problem of  efficient Pareto sum computation from a theoretical and practical perspective.
Our main results are as follows:
\begin{itemize}
    \item We first investigate the fine-grained complexity of Pareto sum computation. We show that for $A,B$ of size $n$, the existence of  an algorithm that computes their Pareto sum in $\mathcal{O}(n^{2-\delta})$ for $\delta > 0$ contradicts the (min,+)-convolution hardness conjecture. Notably, this result applies to instances where the output size $k$ is linear in $n$.
    \item For the special case where $A$ and $B$ are sorted convex point sets, we devise an algorithm that computes their Pareto sum in linear time. For general inputs, the algorithm identifies a part of the Pareto sum and thus can be used as an efficient  preprocessing step for other Pareto sum computation algorithms.
    \item We present and thoroughly analyze several algorithms for exact Pareto sum  computation with a  focus on achieving an output-sensitive space consumption.  Table \ref{tab:overview} provides an overview of the characteristics of our algorithms as well as existing baseline approaches. For output sizes $k \in o(n)$, our newly proposed Successive Sweep Search (SSS) algorithm achieves a subquadratic running time. This shows that the proven fine-grained complexity result does not extend to the case where the output size is sublinear.  However, for $k \in \Theta(n^2)$, SSS has a  cubic running time. Our Sort \& Compare algorithm, on the other hand,  achieves a close to quadratic running time on all instances,  independent of $k$.
    \item To accelerate the proposed algorithms in practice, we devise suitable data structures for their implementation and present multiple engineering techniques.
    \item In an extensive experimental study, we compare the scalability of all presented algorithms. We consider  generated data  sets as well as  real inputs that stem from bi-criteria route planning instances. We observe that our novel algorithms, and in particular the engineered variants, are faster than existing approaches by up to two orders of magnitude. Moreover, in line with our theoretical analysis, we show that the output size  impacts  which algorithm works best.
\end{itemize}

\section{Formal Definitions}
In this section, we formally define the notion of a Pareto sum and provide notation used throughout the article.

\begin{definition}[Domination]
Given two points $p,p' \in \mathbb{R}^2$, we say that $p$ dominates $p' $, or $p \prec p'$, if $p \neq p'$ and $p.x \leq p'.x$ as well as $p.y \leq p'.y$.
\end{definition}
As input for Pareto sum computation, we always expect two Pareto sets.

\begin{definition}[Pareto set]
A set  $S \subset \mathbb{R}^2$ is a Pareto set if no point in $S$ dominates another point in $S$, that is $\nexists s,s' \in S$ with $s \prec s'$.
\end{definition}
We   assume that Pareto sets are sorted in lexicographic order. We  use $S_i$ to refer to the element with rank $i$ in set $S$. 

\begin{definition}[Minkowski sum]
Given two Pareto sets $A, B \subset \mathbb{R}^2$, their Minkowski sum   $M = A \oplus B$ is defined as $M:= \{a+b|~a \in A, b \in B\}$.
\end{definition}
In a slight abuse of notation, we will use $M$ to refer to the set of elements in the Minkowski sum as well as the  matrix where $M_{ij}=A_i+B_j$.

\begin{definition}[Pareto sum]
Let $A, B \subset \mathbb{R}^2$ be two Pareto sets of size $n$ and let  $M = A \oplus B$ denote their Minkowski sum. Then the Pareto sum $C$ of  $A,B$ is defined as the set of all non-dominated points in $M$.
\end{definition}
Throughout the paper, we will use $k$ to denote the size of $C$. 
 
\section{Conditional Lower Bound}
In this section, we investigate the fine-grained complexity of Pareto sum computation. We prove that  an algorithm for Pareto sum computation with a running time of $\mathcal{O}(n^{2-\delta})$ for $\delta > 0$ would disprove the prevailing hardness conjecture that the (min,+)-convolution problem admits no truly subquadratic algorithm \cite{cygan2019problems}. 

In the (min,+)-convolution problem, the input are two arrays  $a[0..n-1], b[0..n-1]$  and the goal is to compute the output array $c[0..n-1]$ with $c[x] := \min_{i+j = x} a[i]+b[j]$. In the following theorem, we  show that (min,+)-convolution can be reduced to Pareto sum in linear time. Thus, any algorithm with a subquadratic running time for Pareto sum can be used to also solve (min,+)-convolution in subquadratic time.

\begin{theorem}\label{thm:reduce}
If Pareto sum can be solved in time $T(n)$ then (min,+)-convolution can be solved in time $\mathcal{O}(T(n)).$
\end{theorem}
\begin{proof}
Let $a,b$ be the input of a (min,+)-convolution instance. We construct Pareto sets $A, B$ from $a,b$ in linear time using the following transformation
\begin{align*}
    A &:= \{(i, a[i] + (n-i)U) : \forall i = 0,\dots, n-1\}\\
    B &:= \{(j, b[j] + (n-j)U) : \forall j = 0, \dots, n-1\}
\end{align*}
for $U := \max a + \max b + 1$. Let now $C$ denote the Pareto sum of $A, B$.  We infer the output of (min,+)-convolution by setting $c[x] := y-(2n-x)U$ for $(x, y)\in C$. Note that $|C| < 2n$ as the $x$-coordinates of the points in the Minkowski sum are integers between $0$ and $2n-2$.

To prove the correctness of this reduction, we need to show that for each $x = 0,\dots,n-1$, the Pareto sum $C$ contains the element $(x, y_x)$ where $y_x-(2n-x) = c[x]$ and where $(x, y_x)$ is not dominated by any other element in $C$.

Let $y_x^*:=\min\{a[i]+b[j] : i+j=x\}+(2n-x)U$. As $c[x] = \min_{i+j = x} a[i]+b[j]$ and for all $i, j$ with $i+j = x$ the additive term is the same, namely $(n-i)U+(n-j)U = (2n-x)U$, it follows that $y_x^* = c[x] + (2n-x)U$ is the smallest possible $y$-coordinate for any element of the form $(x, y_x)$ in the Minkowski sum of $A, B$.

It remains to prove that $(x,y_x^*)$ cannot be dominated by another element $(x', y')$ with $x' < x$. This element would have the form
$(x', a[i']+(n-i')U + b[j]+(n-j')U)$ with $i'+j'=x'$. We can then lower bound  $y'$ as follows:
\begin{align*}
y'&= a[i']+b[j'] + (2n-(i'+ j'))U  = a[i']+b[j'] + (2n-x')U\\
   &= a[i']+b[j'] + U + (2n-(x'+1))U  \geq	a[i']+b[j'] + U + (2n-x)U\\
   &\geq	U + (2n-x)U >	a[i]+b[j] + (2n-x)U\text{ for any }i+j=x
\end{align*}
The last inequality follows from $U$ being larger than the sum of any element pair from $a$ and $b$. Thus, $(x', y')$ cannot dominate $(x,y_x^*)$. We conclude that for each  $x$-coordinate $0,1,\dots,n-1$ the element $(x, c[x]+(2n-x)U)$ is contained in the Pareto sum $C$ of $A,B$.
\end{proof}
Based on the conjecture that there is no algorithm that solves (min,+)-convolution in time $\mathcal{O}(n^{2-\delta})$ for $\delta >0$ \cite{cygan2019problems} and the reduction described in Theorem \ref{thm:reduce}, we derive the following novel hardness conjecture.

\begin{conjecture}
There is no algorithm for Pareto sum for input Pareto sets $A, B$ of size $n$ each with a running time of $\mathcal{O}(n^{2-\delta})$ for $\delta >0$ for $k \geq 2n$.
\end{conjecture}
While for Pareto sums of size $k \in \Theta(n^2)$, a quadratic running time for their computation is clearly necessary, our reduction shows that the complexity of the problem does not hinge on a large output size. Instead, a truly subquadratic algorithm is unlikely to exist even if the output size in linear in the input size. Based on this hardness result, we next focus on designing algorithms  for Pareto sum computation with close-to-quadratic running time. Furthermore, we note that the hardness result does not apply to instances with $k \in o(n)$. Instances with sublinear output size do exist as proven in \cite{klamroth2022efficient}. We will show below that in this case truly subquadratic running times  are achievable.

\section{Algorithms for Pareto Sum Computation}
In this section, we propose a variety of algorithms to tackle the problem of efficient Pareto sum computation for given Pareto sets $A, B$ of size $n$. We first discuss a linear time algorithm for the special case where the points in $A$ and $B$ are in convex position. Afterwards, we discuss base algorithms and successive algorithms for the general case and analyze their running times  and space consumption. Furthermore, we discuss engineering methods to foster their efficiency in practice.
\subsection{Pareto Sums of Convex Point Sets}\label{sec:convex}
As a warm-up, we present an algorithm that for given Pareto sets $A, B$ computes part of their Pareto sum in linear time. Thus, this algorithm can be used as an efficient preprocessing method before applying other techniques.  
Furthermore, we show that for convex input sets, the algorithm already returns the whole Pareto sum.

Our algorithm relies on interpreting the given point sets as polygons.
For two  convex polygons $P, Q \in \mathbb{R}^2$, their Minkowski sum $P\oplus Q$ is a convex polygon with at most $|P|+|Q|$ vertices and these vertices can be computed in linear time \cite{mark2008computational}. 
Let now $A, B$ be sorted Pareto sets  augmented   with  dummy points  $(x,y)$ where $x:= \max_{s \in S}s.x$ and $y:=\max_{s \in S}s.y$ for $S=A$ and $S=B$, respectively.  We use $CH(A)$ and $CH(B)$ to refer to the convex hulls of these two sets. The following observation captures the connection between these convex hulls and the Pareto sum.
\begin{observation}
The vertices of the Minkowski sum $CH(A)\oplus CH(B)$ are a subset of the Pareto sum of $A$ and $B$ (excluding the dummy point sum).
\end{observation}
\begin{figure}
    \centering
    \includegraphics[width=0.495\textwidth]{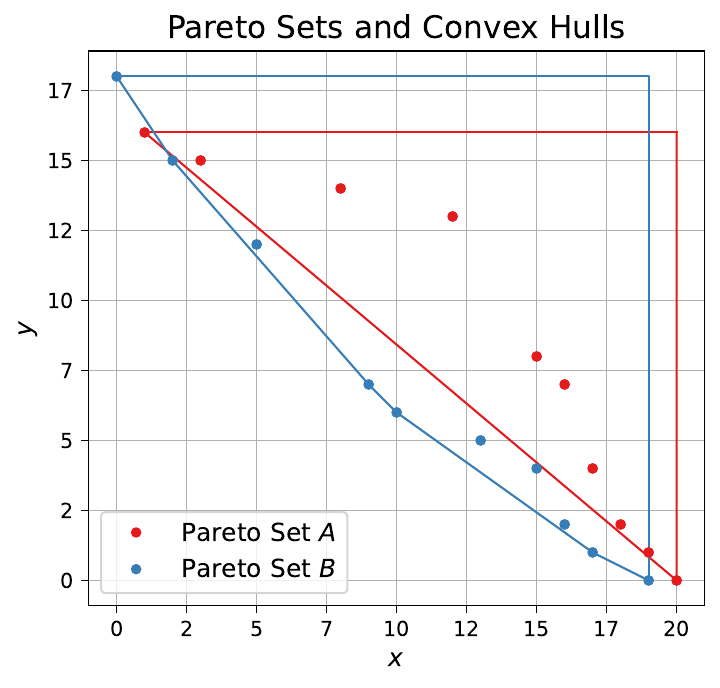}\hfill
    \includegraphics[width=0.495\textwidth]{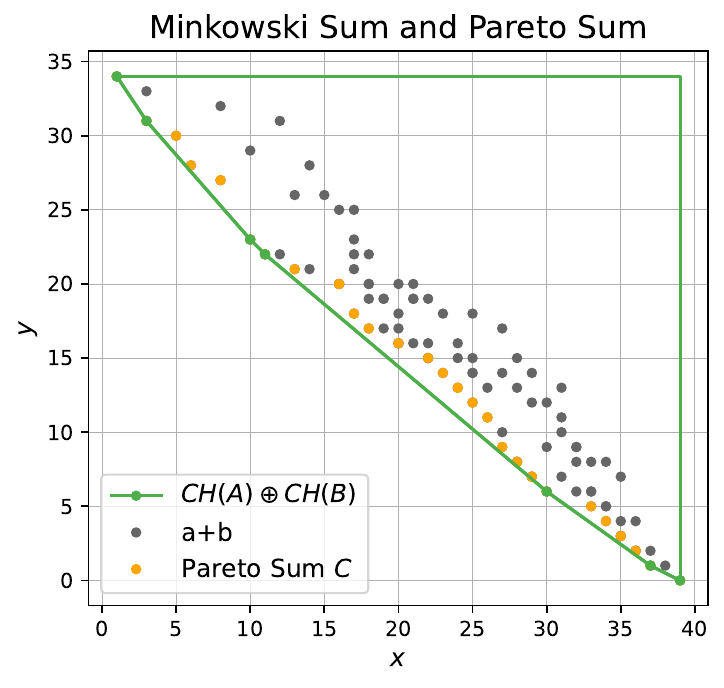}
    \caption{Left: Pareto sets $A, B$ together with their convex hulls.  Right:  The Minkowski sum $CH(A) \oplus CH(B)$ of the convex hulls encloses all pairwise vector additions $a+b$ with $a \in CH(A)$ and $b \in CH(B)$. The respective vertices (green points) are a subset of the Pareto Sum $C$.}
    \label{fig:Minkowski}
\end{figure}

For a sorted Pareto set, its convex hull can be computed in linear time using Andrew's algorithm \cite{andrew1979another}. Then, using the linear time Minkowski sum algorithm on the two convex hulls and extracting the respective polygon vertices, we obtain a subset of the Pareto sum $C$, see Figure \ref{fig:Minkowski}. If both $A,B$ are convex, then this procedure already returns all of $C$. We thus get the following corollary.
\begin{corollary}
The Pareto sum of two sorted convex Pareto sets of size $n$ each can be computed in $\mathcal{O}(n)$.
\end{corollary}
For non-convex $A, B$, we might only get part of the Pareto sum.  However, as   this step only takes linear time (assuming the Pareto sets are presorted), it can always be used as an initial step before applying other algorithms. We will discuss below in more detail how the knowledge of $C' \subset C$ can be exploited to decrease the practical running time of several of the algorithms we propose.

\subsection{Base Algorithms for Pareto Sum Computation}
Next, we discuss  three simple base algorithms for Pareto sum computation along with  engineering concepts for their acceleration and space consumption reduction. The algorithms all proceed by checking for each element $p \in M$ whether there exists $p' \in M$ that dominates $p$. If there is no such $p'$, the point $p$ is added to the Pareto sum $C$. The only difference between the algorithms is the implementation of the dominance check.
\subsubsection{Brute Force (BF)}\label{sec:BF}
The easiest way to check for a point $p \in M$ whether it is non-dominated is by pairwise comparison to all other elements in $M$.   This dominance check takes  $\mathcal{O}(|M|)$ time per point, accumulating to a total  time of $\mathcal{O}(|M|^2) = \mathcal{O}(n^4)$. As the elements $M_{ij}$ can be computed on demand, the space consumption is linear in the input size $n$ and the output size $k$.
\begin{corollary}
 The BF algorithm runs in $\mathcal{O}(n^4)$ time using $\mathcal{O}(n+k)$ space.
\end{corollary}

\subsubsection{Binary Search (BS)}\label{sec:BS}
To decrease the time needed for the dominance check, we take the structure of $M$ into account. Based on the assumption that $A$ and $B$ are sorted and that $M_{ij}$ is defined as $A_i + B_j$, we have the property that each column (and each row) of the matrix $M$ forms a sorted Pareto set on its own. Thus, if we want to check whether column $M_j$ contains an element dominating $p$, we simply have to find the entry $M_{ij}$ with the largest index $i$ such that $M_{ij}.x \leq p.x$ as well as the entry $M_{i'j}$ with the smallest index $i'$ such that $M_{i'j}.y \leq p.y$. Then all elements in $M_j$ with a row index in $[i',i]$ dominate  $p$ (or are equal to $p$). Accordingly, the dominance check for $M_j$ boils down to evaluating whether $i' \leq i$ holds. These two indices  can each be identified via a binary search over the respective coordinates in the column. Hence the dominance check time per column is in $\mathcal{O}(\log n)$, resulting in a total time of $\mathcal{O}(n \log n)$ per element in $M$. Entries of $M$ that need to be accessed can be computed on demand.
\begin{corollary}
 BS runs in $\mathcal{O}(n^3\log n)$ time using $\mathcal{O}(n+k)$ space.
\end{corollary}

To reduce the practical running time of the BS algorithm, we propose the following engineering  techniques.
\smallskip\\
\textit{Pruning.} Whenever we identify a non-dominated point $p$ and add it to $C$, we can also compute all entries in $M$ dominated by $p$ in time $\mathcal{O}(n \log n)$, again with the help of two binary searches per column. For those points, dominance does not need to be checked again. However, if we simply store a flag for each entry in $M$ whether it needs to be further considered or not, the space consumption increases to $n^2$. Instead we can store for each column the set of intervals of dominated points in an interval tree. The number of intervals per column is upper bounded by $k$. Then, for a point $p=M_{ij}$ we can query the interval tree in time $\mathcal{O}(\log k)$ to see whether the point lies in a dominated region. Intervals can also be added or merged within the same time. However, the space consumption would still increase to $\mathcal{O}(nk)$. To keep the space consumption linear, one might only want to store a constant number of intervals per column (e.g. only the largest one) and then merge or replace intervals if possible or needed. If pruning is applied, we can also disregard fully dominated columns in the binary searches of the remaining elements. 
\smallskip\\
\textit{Priority Binary Search (PBS).} As soon as dominance checks might be avoided for some of the elements based on the pruning techniques described above, the order in which the points are considered impacts the running time. Identifying points that dominate many other points early on can significantly reduce  the total number of checks.  For that purpose, we will use the preprocessing step described in Section \ref{sec:convex} to get an initial set of points $C' \subset C$. We can directly exclude any points dominated by the points in $C'$. Furthermore, we conjecture that  points in the same rows or columns as the points of $C'$ occupy in $M$ are likely to be also part of $C$.  Thus, we give priority to these points in our search. Figure \ref{fig:priority} shows some visual support for this hypothesis.
\begin{figure}
\centering
\begin{center}
   \begin{subfigure}{0.43\textwidth}
        \includegraphics[width=\textwidth]{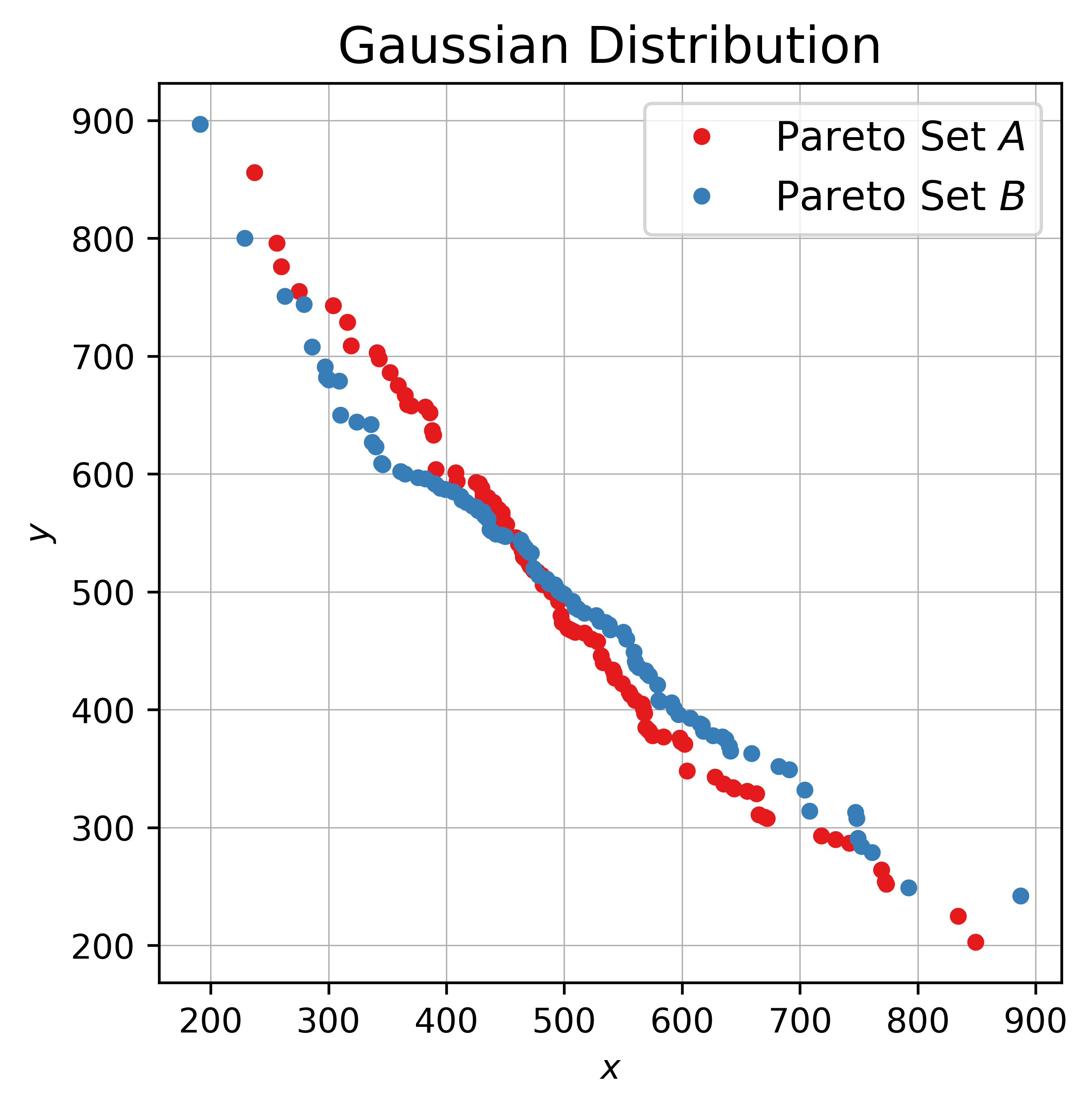}
    \end{subfigure} \hfill
    \begin{subfigure}{0.43\textwidth}
        \includegraphics[width=\textwidth]{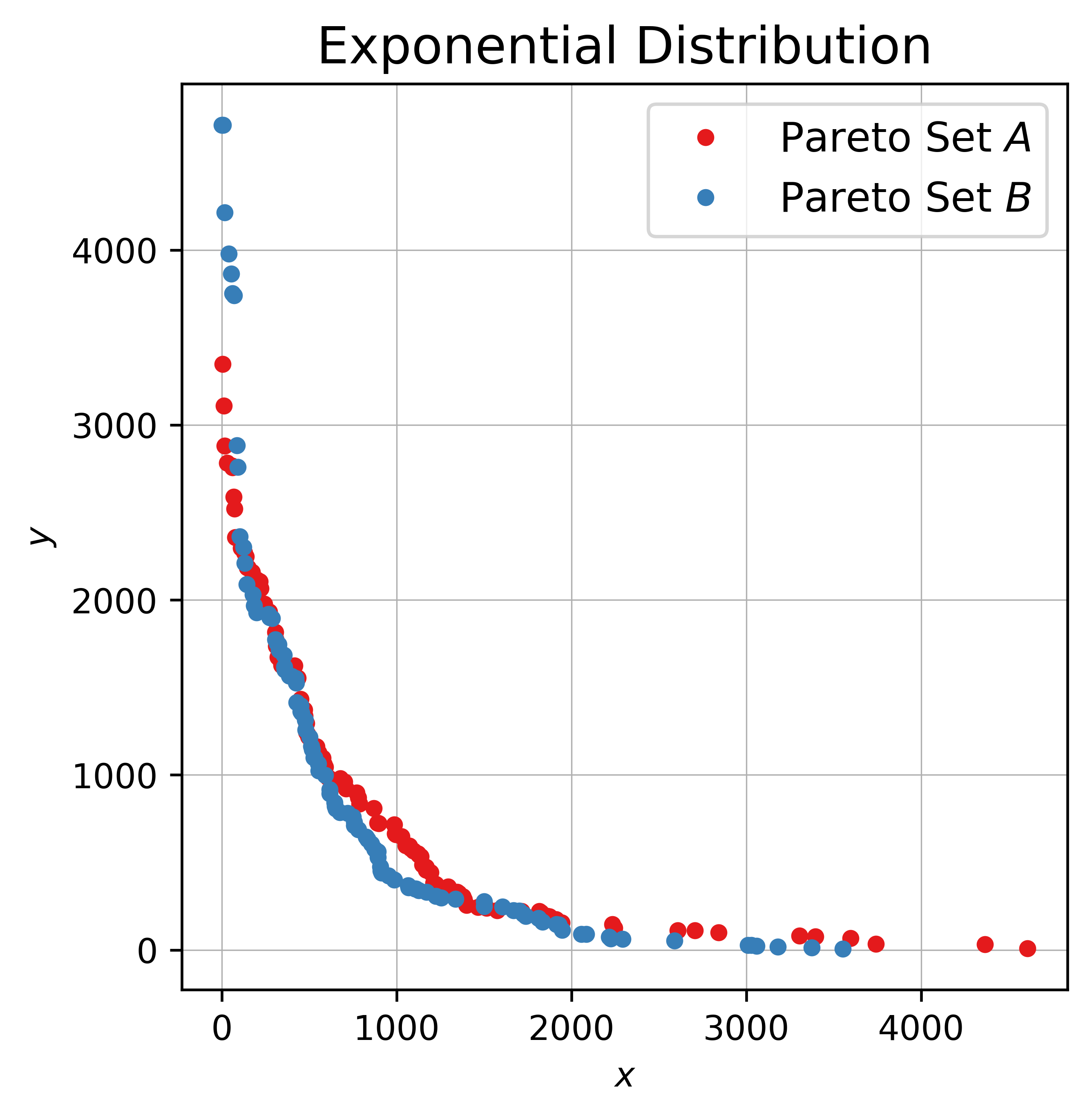}
    \end{subfigure} 
\end{center}

\begin{center}
    \begin{subfigure}{0.45\textwidth}
        \includegraphics[width=\textwidth]{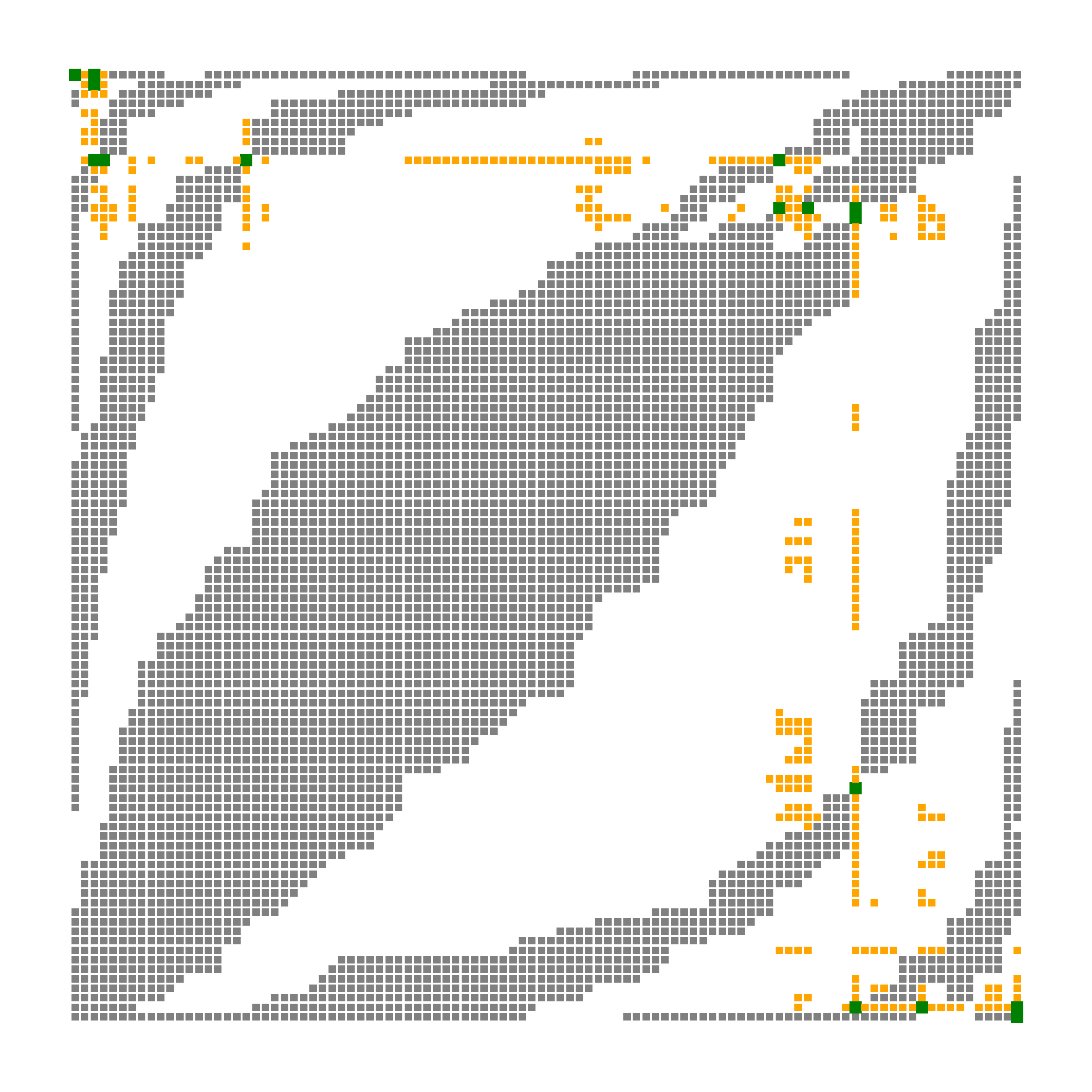}
    \end{subfigure}\hfill
    \begin{subfigure}{0.45\textwidth}
        \includegraphics[width=\textwidth]{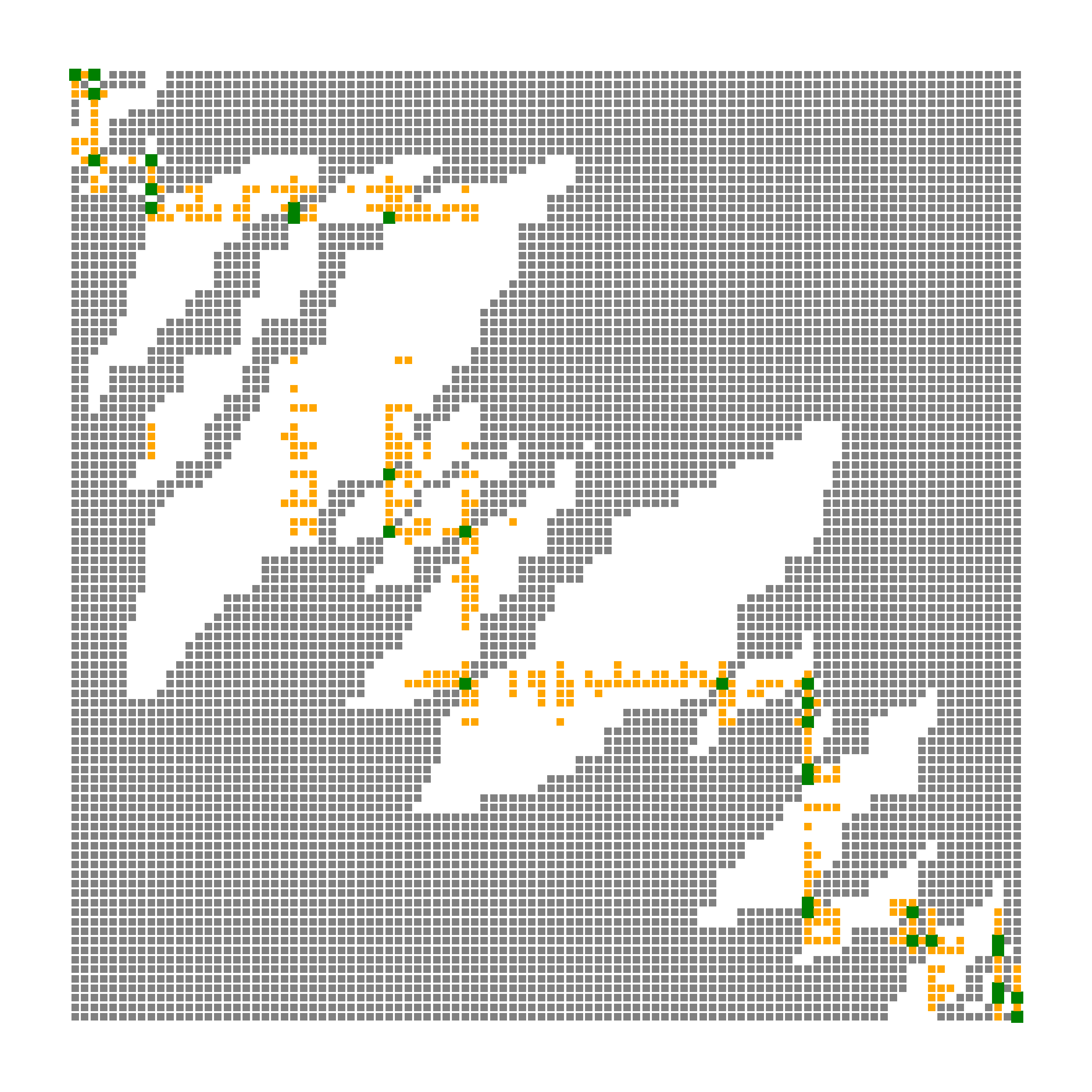}
    \end{subfigure}
\end{center}  

\begin{center}
    \begin{subfigure}{0.495\textwidth}
        \includegraphics[width=\textwidth]{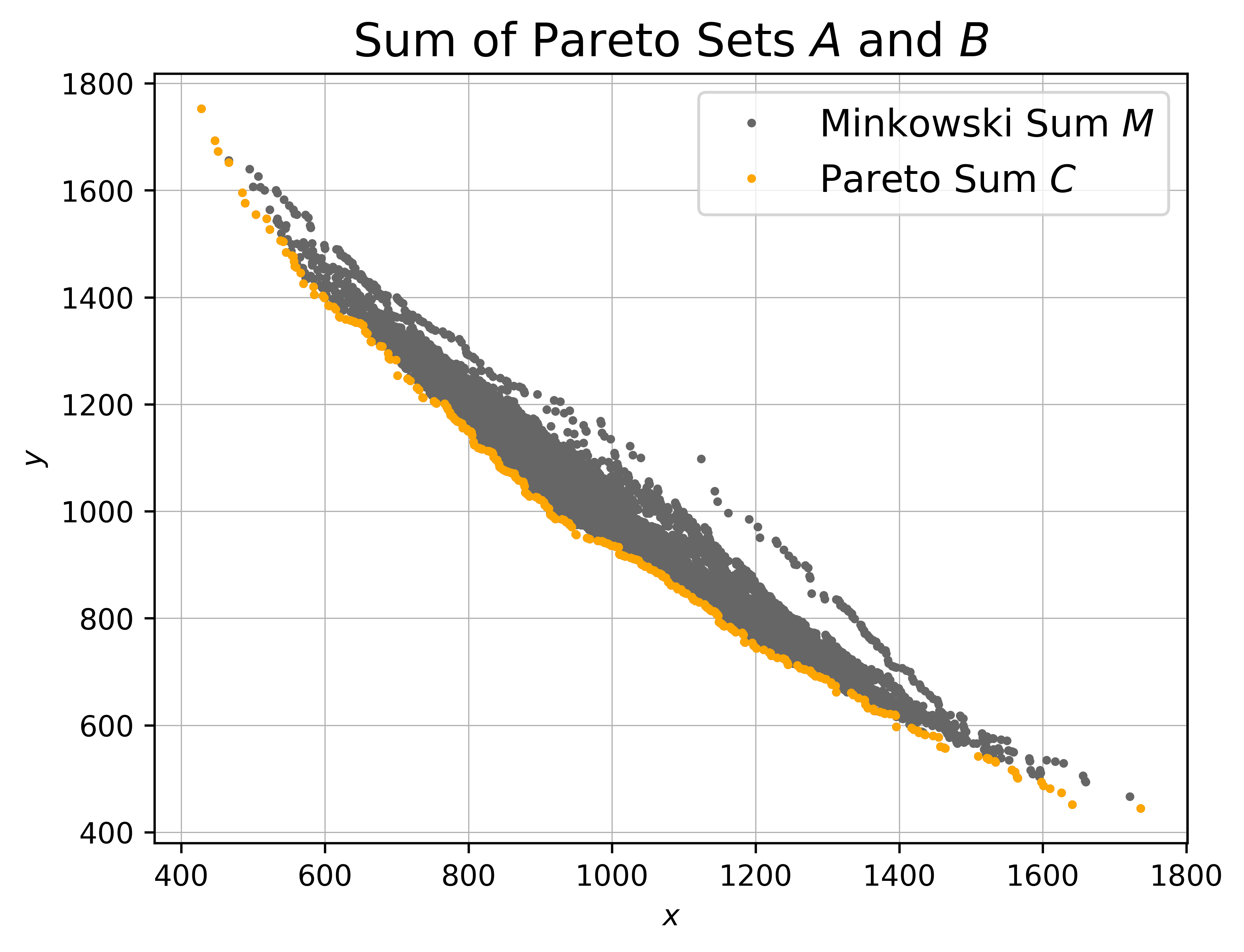}
    \end{subfigure}\hfill
    \begin{subfigure}{0.495\textwidth}
        \includegraphics[width=\textwidth]{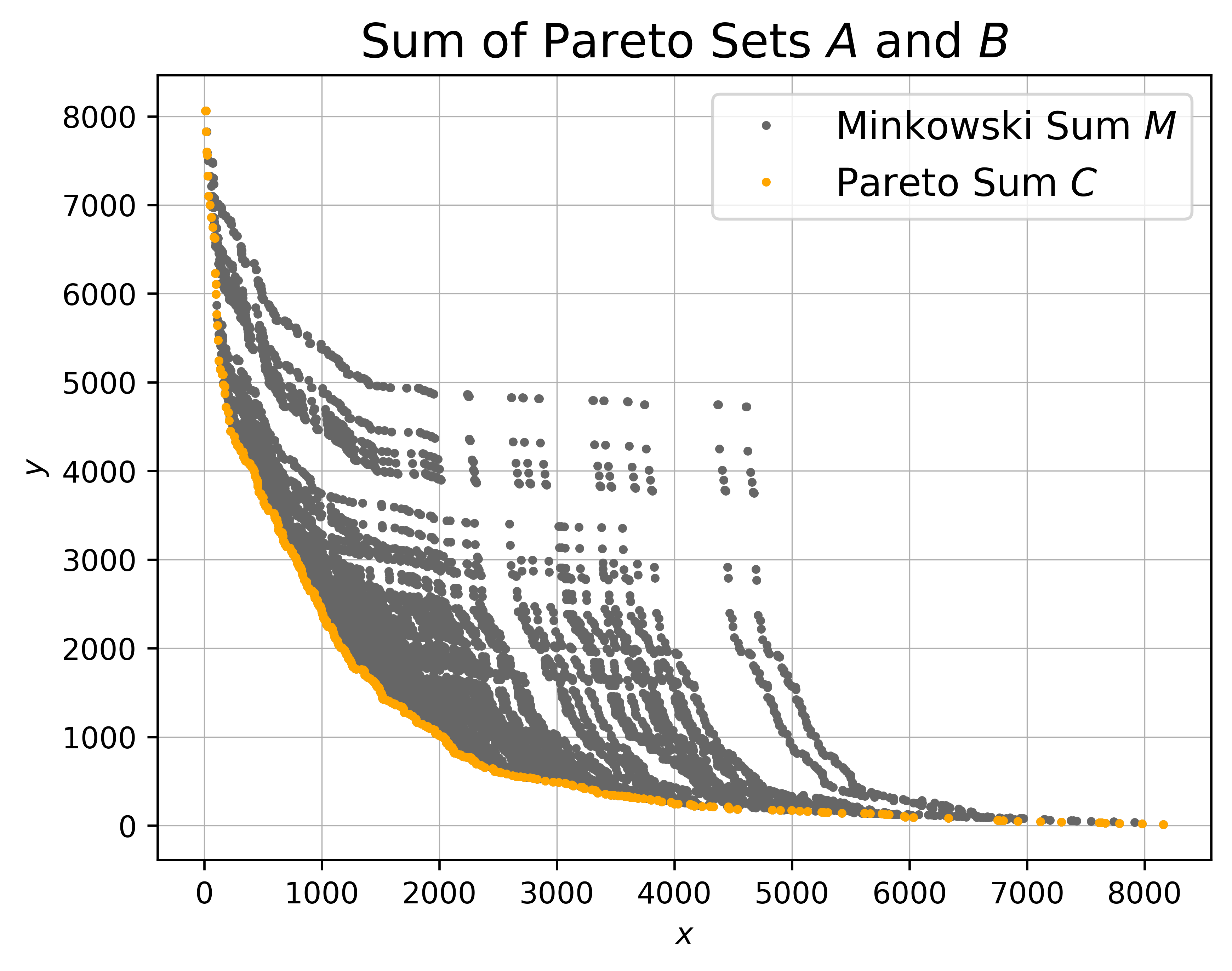}
    \end{subfigure}
\end{center}
    \caption{Input Pareto sets following different kinds of distributions (top) and schematic depiction of the corresponding Minkowski matrix $M$ (middle). Green dots indicate entries in $M$ that are points on the Minkowski sum of the convex hulls of $A$ and $B$, black dots indicate entries that are dominated by the green ones, and orange dots encode the remaining elements of the Pareto sum which then together dominate the white dots. In the bottom row, the Minkowski sum and the Pareto sum are illustrated based on point coordinates.}
    \label{fig:priority}
\end{figure}
\subsubsection{Sort \& Compare (SC)} \label{sec:SC}
Given a sorted set of points, extracting the set of non-dominated points can be accomplished in constant time per point. The smallest element is always added to $C$. For each other element in sorted order, we check whether it is dominated by (or equal to) the currently last element in $C$. If that is not the case, the element is added to $C$.

Computing and sorting $M$ as a whole takes time $\mathcal{O}(n^2 \log n)$ and requires quadratic space. But we can  exploit the structure of $M$ to improve the space consumption to linear as follows: We use a  min-heap data structure and initialize it with the first row of $M$. Each element in the heap remembers its position in $M$. When  we extract the min element $M_{ij}$ from the heap and $C$ is empty so far, we add the element to $C$. Otherwise, we compare $M_{ij}$ to the element added to $C$ last. If $M_{ij}$ is not dominated, we also add it  to $C$. In any case, we add its column successor $M_{i+1j}$ to the heap (as long as $i < n$). As each column is a sorted Pareto set in itself, we know that $M_{i+1j}$ has to have  larger $x$-value than $M_{ij}$. Thus, we extract the elements from the heap exactly according to their global lexicographic order. As the heap never contains more than $n$ elements, its space consumption is in $\mathcal{O}(n)$ and the heap operations take $\mathcal{O}(\log n)$ per round. In conclusion, the heap-based variant has the same asymptotic running time as the one where we  fully compute and sort $M$, but a significantly reduced space consumption.
\begin{corollary}
SC runs in $\mathcal{O}(n^2\log n)$ time using $\mathcal{O}(n+k)$ space.
\end{corollary}
To reduce the running time in practice, we check for each element before its insertion into the heap whether it is dominated by the last element added to $C$. If that is the case, we simply omit the insertion and proceed with the next element in the column.

\subsubsection{NonDomDC and Pareto Trees}\label{sec:ND_ParetoTrees}
Due to the structure of the Minkowski sum matrix $M$ each column forms a Pareto set itself. In general, the NonDomDC \cite{klamroth2022efficient} algorithm computes the Pareto sum $C$ of Pareto sets $A$ and $B$ by merging columns of the matrix and pruning dominated points. Let $P_i$ be the Pareto sum of points in the first $i$ columns of $M$. The column $i+1$ is merged with $P_i$ by pruning dominated elements to obtain the Pareto sum $P_{i+1}$. Since $P_i$ and the column $i+1$ are sorted according to the $x$-coordinate, we use two pointers to keep track of the smallest point not yet merged. Similarly to the Sort \& Compare algorithm, we process points in $P_i$ and the column $i+1$ in the order of their $x$-coordinates and add the current point to $P_{i+1}$ if it is not dominated by the last point added to $P_{i+1}$. Computing the intermediate Pareto sum $P_{i+1}$ by merging and pruning takes $\mathcal{O}(|P_i|+n)$ time. This merging of columns and intermediate Pareto sets is performed for all $n$ columns in $M$ and is referred to as sequential NonDomDC (ND). With the maximum size of an intermediate Pareto set denoted as $P := \max_{i=1}^n |P_i|$, the total running time of the NonDomDC approach is in $\mathcal{O}(Pn)$ and the space consumption is in $\mathcal{O}(n+P)$. However, the NonDomDC algorithm is not an output-sensitive algorithm since $P$ might be significantly larger than the actual output size $k=|C|$. Another variant of the NonDomDC algorithm, called doubling ND, merges pairs of columns in a MergeSort like fashion until the final Pareto sum $C$ is obtained.

Pareto Trees, or ND trees \cite{jaszkiewicz2018nd} have been proven useful for the dynamic non-dominance problem. A Pareto tree maintains a Pareto set when new points are added to the tree. Non-dominated points are stored in the leaves of a Pareto tree where one leaf contains up to $p$ points. Non-leaf nodes have up to $c$ children. Each node $v$ in the Pareto tree represents a set $S_v$ of points. If $v$ is a leaf, $S_v$ is the set of points stored in $v$. For non-leaf nodes $v$, the set $S_v$ represents all points stored in the subtree rooted at $v$. Note that in the latter case, the set $S_v$ is not stored explicitly. We use $c=3$ and $p=20$ as recommended in \cite{jaszkiewicz2018nd} for the best practical performance of two-dimensional Pareto trees. Additionally, each node $v$ stores a lower and upper bound such that each point in $S_v$ is within the bounding rectangle defined by the bounds. 

The refinement into space-partitioned ND trees (SPND trees) \cite{lang2022space} ensures that the tree is balanced and the bounding rectangles of children do not overlap. The space is partitioned along the $x$- and $y$-axes such that points stored in a leaf node are geometrically close to each other and further apart from points stored in other leaves. Determining if a candidate point should be added to the Pareto set represented by the tree is done by calling $\mathrm{NonDomPrune}$. It returns True if the point is non-dominated in the Pareto tree and removes any points from the tree that are dominated by the candidate. Otherwise, it returns False, indicating that the candidate point is dominated by some point in the Pareto tree. The operation $\mathrm{Insert}$ then inserts the candidate point into the Pareto tree at the corresponding leaf node following the space-partitioning along the $x$- and $y$-axes.

Using the Pareto trees we derive a simple algorithm to successively compute the Pareto sum $C$ of Pareto sets $A$ and $B$. The first and the last entry in the Minkowski sum $M$ have the smallest $x$-value and $y$-value, respectively, among all points in $M$ and thus are part of the Pareto sum $C$. We build a Pareto tree with the initial set $\{M_{11}, M_{nn} \}$. Further, we call $\mathrm{NonDomPrune}$ for each entry $M_{ij}$ and add the point $M_{ij}$ to the Pareto tree if is is non-dominated. Finally, the Pareto sum is retrieved from the set of points stored in the leaves of the Pareto tree.

Another approach leverages the linear time preprocessing algorithm described in Section \ref{sec:convex}. We first derive an initial set $C' \subset C$ with which we build the Pareto tree. The algorithm continues by traversing $M$ and adding points to the tree when appropriate. Since the Pareto tree entails more information apriori, using the initial set $C'$ rejects dominated points earlier and is therefore more efficient.

As for the two variants of NonDomDC, we use Pareto trees to represent intermediate Pareto sums $P_i$. When merging a column $i+1$ with $P_i$ we simply call $\mathrm{NonDomPrune}$ for each point in column $i+1$ and insert it into the Pareto tree if necessary. The resulting set $P_{i+1}$ is contained in the Pareto tree after the merging is completed. Hence, we only need to retrieve the Pareto set once from the Pareto tree after all columns are merged. Using the Pareto trees further improves the efficiency of the sequential and doubling NonDomDC algorithms. 

\subsection{Successive Algorithms  for Pareto Sum Computation}
If the Pareto sum $C$ contains (almost) all elements of the  Minkowski sum $M$,  a quadratic running time is needed already to report $C$. 
And as proven above in our fine-grained complexity analysis, even for linear output sizes and algorithm with a running time of $\mathcal{O}(n^{2-\delta})$ for $\delta > 0$ is unlikely to exist.
In these cases, the running time of  SC  is close-to-optimal. However, in case $C$ is small, subquadratic running times might be possible. We will present two output-sensitive  algorithms in this section that have a running time asymptotically faster than SC for $k \in o(n)$ or $k \in o(n \log n)$, respectively. Both algorithms detect the elements in $C$ successively. This is a well-established paradigm for output-sensitive skyline computation, see e.g. \cite{liu2014faster,yu2011improved}. However, known algorithms rely on the explicit availability of the point set  to construct an efficient search data structure. Based on the following lemma, we will design successive algorithms that do not need access to $M$ as a whole.
\begin{lemma}\label{lem:successive}
Let $A, B$ be two Pareto sets and $c,c' \in C$ two elements of their Pareto sum with $c.x < c'.x$. Then the lexicographically smallest element $m$ in $M$ that dominates  $(c'.x-\varepsilon,c.y-\varepsilon)$ for $\varepsilon > 0$ is also part of $C$ (if such an element exists). 
\end{lemma}
\begin{proof}
We first argue that for any $m \in M$, the  smallest point $p \in M$  dominating $m$ (or being equal to $m$) is  part of the Pareto sum $C$.  Assume otherwise for contradiction. Then there is a point $p' \in C$ that dominates $p$ and thus also $m$. But in this case $p'$ is smaller than $p$ which contradicts the choice of $p$ as smallest element to dominate $m$.

Now, if there is any point  $m \in M$ that dominates the dummy point $(c'.x-\varepsilon,c.y-\varepsilon)$ the above argumentation applies. 
\end{proof}
Clearly, $M_{11}$ and $M_{nn}$ are always part of the Pareto sum, as those are the points with smallest global $x$-value and $y$-value, respectively. All other elements in $C$ must have an $x$-value in the open interval $(M_{11}.x, M_{nn}.x)$. Thus, if we have an oracle  that returns the smallest point $m \in M$ (with respect to lexicographic ordering) in a given range $[x_{\min}, x_{\max}) \times [y_{\min}, y_{\max})$, we can compute $C$ based on Lemma \ref{lem:successive} as follows. We initialize $C = M_{11}, M_{nn}$ and $x_{\min} = M_{11}.x, x_{\max} = M_{nn}.x, y_{\min} = M_{nn}.y, y_{\max}= M_{11}.y$. Then we query the oracle to get the smallest point $m$  in the respective range. If such a point does not exist, we abort. Otherwise  we add the point $m$ to $C$ and set $x_{\min} = m.x, y_{\max} = m.y$ before repeating the process. Figure \ref{fig:succ} illustrates the core concept.
\begin{figure}
\centering
     \includegraphics[width=0.55\textwidth]{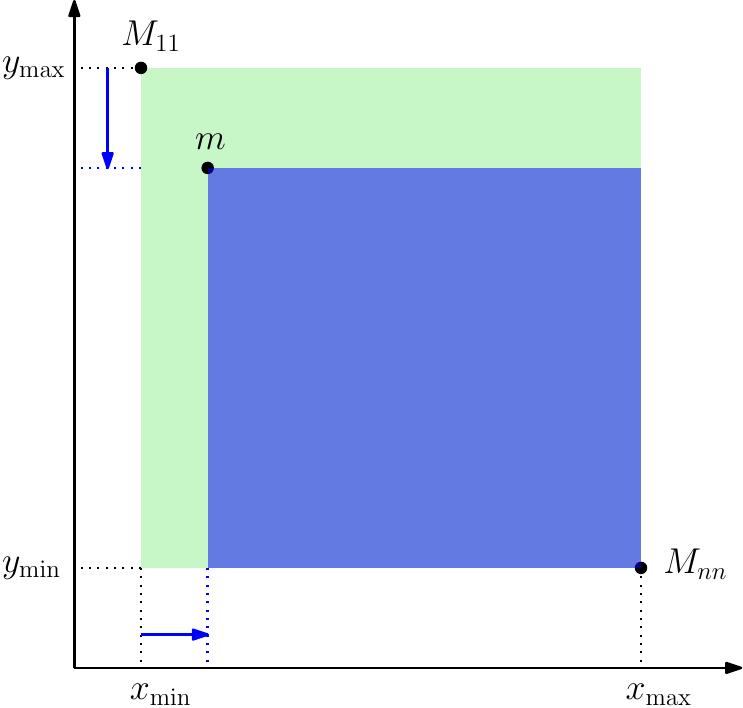}
      \caption{Initial search range (green rectangle) spanned by $M_{11}$ and $M_{nn}$. The range-minimum element $m$ then leads to a reduction of $y_{\max}$ and an increase of $x_{\min}$ (blue arrows) which tightens the search range for the next element of $C$.}
    \label{fig:succ}
\end{figure}

Thus, the algorithm discovers the points in $C$ one by one in increasing order of their $x$-values (except for $M_{nn}$ which is known from the start) using $k$ calls to the range-minimum oracle.  A naive oracle implementation would be to check all points in $M$ for containment in the range  and to keep track of the minimum  among them.  Then each call to the oracle costs $\mathcal{O}(n^2)$ and the overall running time of the successive algorithm would be $\mathcal{O}(n^2  k)$. Next, we describe how to implement the oracle more efficiently.

\subsubsection{Successive   Binary Search (SBS)} \label{sec:CBS}
In the BS approach described in Section \ref{sec:BS}, we use   two binary searches per column of $M$ to check for a query point $m \in M$ whether an element dominating $m$ exists in that column. We can use the same concept to implement a range-minimum oracle: For each column, we identify via binary searches the first position $f_x$ with an $x$-coordinate larger or equal to $x_{\min}$, the last position $l_x$ with an $x$-coordinate smaller than $x_{\max}$, the first position $f_y$ with a $y$-coordinate smaller than $y_{\max}$, and the last position $l_y$ with a $y$-coordinate larger or equal to $y_{\min}$. If $[f_x,l_x] \cap [f_y,l_y] \neq \emptyset$, we return $\max(f_x,f_y)$. The entry at that position is  the smallest point dominating $m$ in the column. Keeping track of the smallest returned point over all columns provides the desired result in $\mathcal{O}(n \log n)$ per oracle call.  

\begin{corollary}
Successive BS runs in $\mathcal{O}(nk\log n)$ time using $\mathcal{O}(n+k)$ space.
\end{corollary}

\subsubsection{Successive Sweep Search (SSS)} \label{sec:CSS}
To improve the oracle time of SBS, we observe that the binary searches in the columns are somewhat redundant. If $M$ was fully available, we could apply fractional cascading \cite{chazelle2005fractional} to the column vectors. This would reduce the running time to compute the positions $f_x, l_x, f_y, l_y$ in all columns from $\mathcal{O}(n \log n)$ to $\mathcal{O}(n)$. Thus, the $k$  oracle  calls   cost $\mathcal{O}(nk)$. Unfortunately, computing $M$ and the data structure for fractional cascading requires space and time in $\Theta(n^2)$. 
Fortunately, we can also achieve linear oracle time without the need to access $M$ as a whole. Based on the structure of $M$, we know that for an entry $M_{ij}$ all elements $M_{st}$ with $s \geq i$ and $t \geq j$ have a larger $x$-coordinate than $M_{ij}$ but a smaller $y$-coordinate. Vice versa, all elements in $M_{st}$ with $s \leq  i$ and $t \leq j$ have a smaller $x$-coordinate than $M_{ij}$ but a larger $y$-coordinate. This implies, for example, that the position of $f_x$ in some column cannot be larger than the position of $f_x$ in the neighboring column to its left. Similar relationships hold for the positions of $l_x, f_y$ and $l_y$ in neighboring columns. Accordingly, we can find the respective column values by a single left-to-right sweep, where the search path forms a monotone staircase structure and is thus bounded in length by $2n$. 

Even better, we can have a single unified sweep to find the range-minimum $m$ in linear time: We start at $M_{n1}$, that is, the last entry of the first column. Whenever we enter a new column $j$, we apply  upwards linear search in that column until we reach an entry $M_{ij}$ where either $M_{ij}.x > x_{\min} $ and $M_{i-1j}.x \leq x_{\min}$ or   where $M_{ij}.y < y_{\max} $ and $M_{i-1j}.y \geq y_{\max}$. Thus, we get $i = \max(f_x, f_y)$; except if the entry we start from already has a too small $x$-value or a too large $y$-value or both which means that the column contains no point in the query range. In the former case, we check whether $M_{ij}$ is contained in the range.  If the check is passed, $M_{ij}$ is a valid candidate for the range-minimum $m$. We keep track of the smallest viable candidate over the course of the algorithm. We then go from element $M_{ij}$ to its right neighbor $M_{ij+1}$ and  proceed with the new column as described above. After processing the last column, we return the current $m$ as the range-minimum element. An example for the procedure of the range-minimum oracle using the sweep search is given in Figure \ref{fig:oracle}. 

\begin{algorithm}
\DontPrintSemicolon
\SetKwFunction{SSS}{SuccessiveSweepSearch}
\SetKwFunction{sweep}{Sweep}
\caption{Successive Sweep Search}
\label{algo:sss}

\SetKwProg{Fn}{Function}{}{}
\SetKwFunction{Sweep}{Sweep}
\SetKwFor{While}{while}{do}{end while}
\SetKwFor{If}{if}{then}{end if}

define $M[r, c]$ as $A[r] + B[c]$\;
\BlankLine
\Fn{\SSS{$A, B$}}{
    $C \leftarrow \emptyset$\;
    $p \leftarrow M[1,1]$\;
    \While{$p \neq \textnormal{null}$}{
        $C \leftarrow C \cup \{p\}$\;
        $R \leftarrow [p.x, \infty) \times [-\infty, p.y)$\;
        $p \leftarrow$ \sweep{$A, B, R$}\;
    }
    \Return{$C$}
}
\BlankLine
\Fn{\sweep{$A, B, R$}}{
    $p \leftarrow \textnormal{null}$\;
    $i \leftarrow n$\;
    $j \leftarrow 1$\;
    \While{$j<n$}{
     \While{$M[i, j] \in R$ \textbf{and} $i \geq 2$}{
        $i--$\;
    }
    $i++$\;
    \If{$i \leq n$ \textbf{and} $M[i, j] \in R$ \textbf{and} $(p == \textnormal{null}$ \textbf{or} $M[i, j] < p)$}{
        $p \leftarrow M[i, j]$\;
    }
    $j++$\;
    }
    \Return{$p$}\;
}
\end{algorithm}

\begin{figure}
    \centering
    \includegraphics[width=0.7\textwidth]{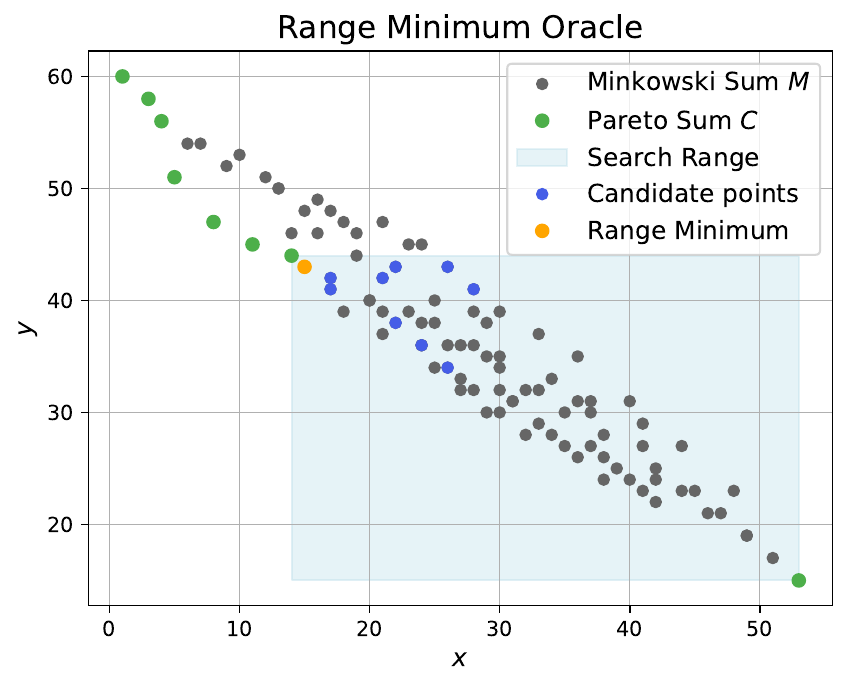} \\
    \includegraphics[width=1\textwidth]{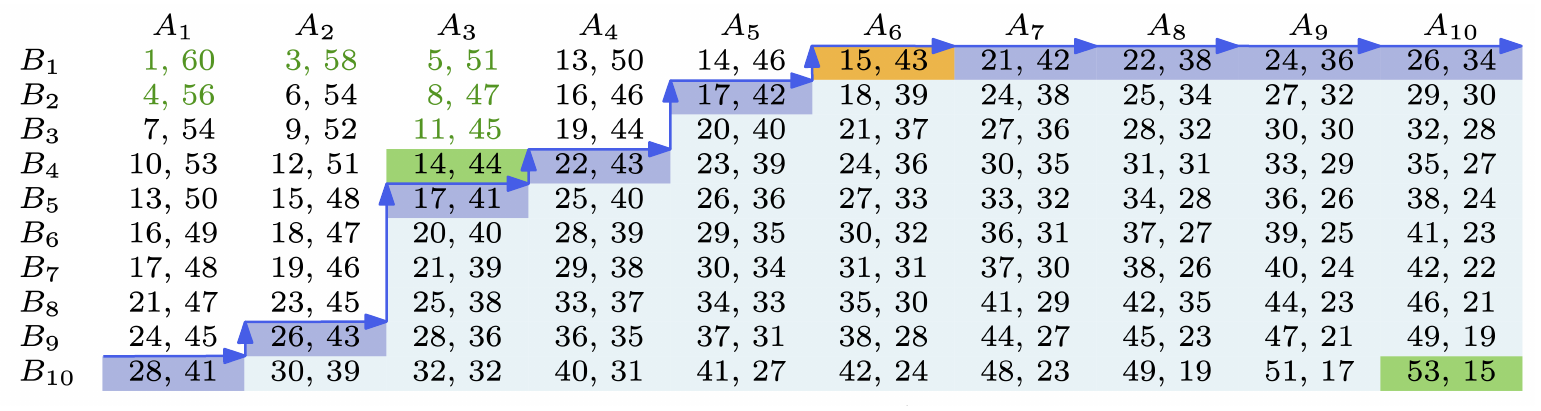}
    \caption{Exemplary traversal of the matrix during one call to the range minimum oracle for the Successive Sweep Search algorithm. The search range (blue shaded area) is defined by the range minimum $m'=(14,44)$ returned by the previous oracle call and the last point $M_{nn}=(53,15)$ in the Minkowski sum $M$. Candidates for the range minimum are located along a monotone staircase path in the matrix (blue path). One single left-to-right sweep suffices to determine the range minimum (orange) $m=(15,43)$ for the given range.}
    \label{fig:oracle}
\end{figure}

\begin{lemma}\label{lem:ss}
The sweep algorithm computes the  range-minimum in $\mathcal{O}(n)$.
\end{lemma}
\begin{proof}
To prove correctness, we need to argue that for a column $j$ entered at row $i$ and exited at row $i' \leq i$, the range-minimum $m$  can not be an entry $M_{i^*j}$ with $i^* < i'$ or $i^*>i$. Clearly, checking elements in column $j$ with a row index smaller than $i'$ cannot give us any viable candidates, as either their respective $x$-value is too small or their $y$-value is too large by definition.  Hence we only need to consider $i^* > i$.
If $M_{ij}$ is in the query range, then the entries in column $j$ with $i^*>i$ cannot constitute the range-minimum as they all have an $x$-value larger than that of $M_{ij}$. If $M_{ij}$ is not in the query range, we have the  following cases:
\begin{itemize}
    \item $M_{ij}.x > x_{\max}$. But then $M_{i^*j}.x > x_{\max}$ holds as well.
    \item $M_{ij}.x < x_{\min}$ or $M_{ij}.y > y_{\max}$. This case only occurs if $i=n$. Thus there is no $i^* > i$.
    \item $M_{ij}.y < y_{\min}$. But then $M_{i^*j}.y < y_{\min}$ holds as well.
\end{itemize}
Accordingly, if column $j$ contains the range minimum it needs to be an element with a row index in $[i',i]$. If $M_{i'j}$ is in the range, it is clearly the best candidate in column $j$ for the range-minimum $m$. If $M_{i'j}$ is not in the query range, then the same applies to all entries in the same column with larger row index as argued above. Thus, it is sufficient to check $M_{i'j}$ for each column $j$.
The running time is determined by the number of elements in $M$ that are accessed. As the interval of elements checked for each column only overlaps with the intervals of all columns to its left in a single row index, at most $2n$ elements in $M$ are considered in total.
\end{proof}
Based on this sweep search (SS) oracle, we now get a successive algorithm with better running time than SBS. The procedure for SSS is described in Algorithm \ref{algo:sss}.
\begin{corollary}
Successive SS runs in $\mathcal{O}(n\log n + nk)$ using $\mathcal{O}(n+k)$ space.
\end{corollary}
In fact, if $k \in o(\log n)$, the running time is dominated by the initial sorting step of the elements in $A, B$. For $k \in o(n)$, we achieve a subquadratic running time. Note that we might modify the algorithm to directly stream single output points instead of storing and returning the Pareto sum. To this end, line 6 where the current range-minimum $p$ is added to the Pareto sum $C$ is altered such that the algorithm outputs $p$ at this point. With this modification we obtain a linear space consumption for the SSS algorithm.

For  acceleration of  sweep search in practice, we observe that if we enter a column at row $i$ and confirm for some value $i'<i$ that $M_{i'j}$ is still feasible with respect to $x_{\min}$ and $y_{\max}$, we do not have to check intermediate rows to get the correct range-minimum by virtue of Lemma \ref{lem:ss}. Similarly, if we have not found any $M_{ij}$ in column $j$ with $M_{ij}.x \geq x_{\min}$ and $M_{ij}.y < y_{\max}$ and the same inequalities apply to some $M_{ij'}$ with $j'>j$, we do not need to check intermediate columns for range-minimum candidates. Thus, in both cases we can introduce a skip threshold $\Delta > 1$ and check for $i'=i-\Delta$ or $j'=j+ \Delta$, respectively, whether the necessary conditions apply. If that is the case we skip intermediate rows or columns and then try to skip ahead again. If skipping is no longer possible, we simply fall back to linear search. Accordingly, in the worst case, we check at most one superfluous element for each row and column. This does not increase the asymptotic running time of the oracle but might reduce its running time in practice if skipping is successful.  

\subsection{Remarks}
For all new algorithms presented in this section (BF, BS, SC, SBS, SSS), we showed that their space consumption is in $\mathcal{O}(n+k)$. The term $k$, of course, is due to having to store the output set $C$. However, if we consider a streaming output model, the space consumption is in $\mathcal{O}(n)$ for all algorithms. BF and BS do not need access to the temporary set $C$ at all, SC only requires access to the currently  last element of $C$, and the successive algorithms again do not consider $C$ but only require the knowledge about the current search range that is induced by $C$. All additional steps and data structures used in these algorithm can be implemented in space linear in $n$. Note that for NonDomDC, the streaming model does not impact its space consumption, as the algorithm needs to combine and filter partial solutions.

Furthermore, we remark that one can easily combine SSS and SC to get a running time of $\mathcal{O}(\min\{nk + n\log n, n^2 \log n\})$ and thus get the best of both, namely a subquadratic running time for sublinear output sizes and a close-to-quadratic running time for larger outputs. Here, we first  conduct up to $n$ range searches in the SSS algorithm. If the full Pareto sum is found in this step, we are done. Otherwise, we have established $k \in \Omega(n)$ in time $\mathcal{O}(n^2)$. We can then simply run SC to get a total running time of $\mathcal{O}(n^2 \log n)$.

\section{Experimental Evaluation}
We  implemented the seven algorithms for Pareto sum computation listed in Table \ref{tab:overview} in C++: The two existing approaches,  namely the Kirkpatrick-Seidel algorithm (KS) \cite{kirkpatrick1985output} and  NonDomDC (ND) \cite{klamroth2022efficient}, the three base algorithms (BF, BS, SC), and the two successive algorithms (SBS, SSS). Furthermore, we  implemented the algorithm based on Pareto trees from Section \ref{sec:ND_ParetoTrees} and engineered versions of the algorithms above (including PBS and SSS with $\Delta$-skipping). For KS, we implemented the simpler (and thus faster) variant, where the output  size $k$ (used for partitioning) is given as an input. We simply compute $k$ with one of our other algorithms and then feed the result into  KS.  For ND, we implemented the sequential and doubling variants described in Section \ref{sec:ND_ParetoTrees} which use the Pareto trees to maintain the intermediate Pareto sums. As benchmark data  we use randomly generated inputs as well as real inputs. Both types of data sets are described in more detail below.  
The experiments were conducted on a single core of a 4.5 GHz AMD Ryzen 9 7950X 16-Core Processor with 188 GB of RAM.

\subsection{Random Input Generators}
We first describe four different ways to obtain randomly generated Pareto sets. The type of generator used and the resulting distribution of the points heavily impacts the Pareto sum size and the running time of the algorithms, as will be shown below.
\subsubsection{Naive Generator}
For the naive Pareto set generator we firstly sample random points with $x$- and $y$-coordinates in the range $[0,n]$. Then we apply a Pareto-filtering based on the Sort \& Compare algorithm by sorting the entire random set according to the $x$-values. We process the points in increasing order of their $x$-values and compare the current point $p$ to the last one added to the Pareto set $A$. The point $p$ is added to the Pareto set if it is not dominated by the last one in $A$.

\begin{figure}
\centering
    \includegraphics[width=0.6\textwidth]{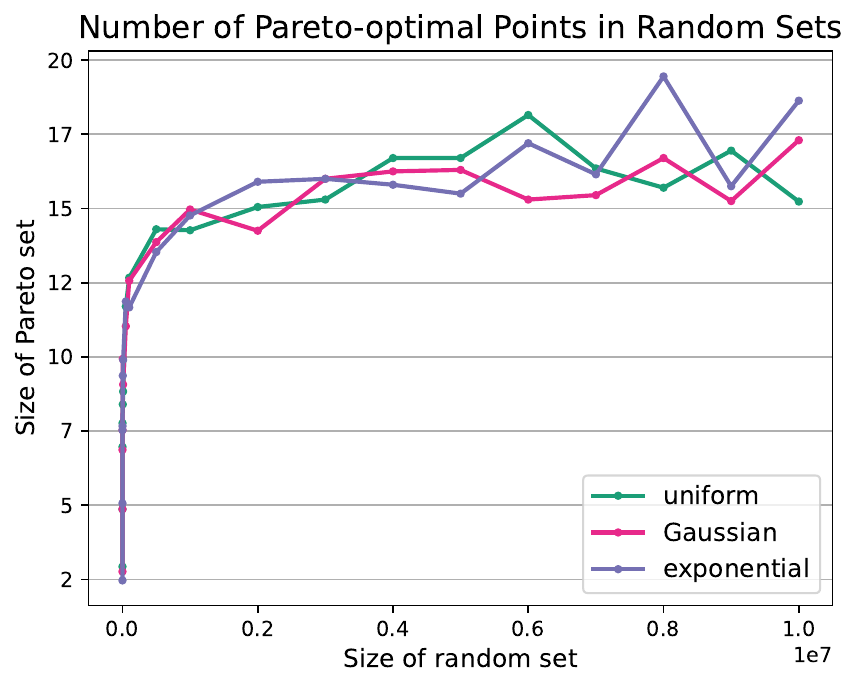}
    \caption{The number of Pareto-optimal points in randomly generated two-dimensional sets of up to 10 million points following the uniform, Gaussian, and exponential distribution. }
    \label{fig:paretoInRandom}
\end{figure}

In Figure \ref{fig:paretoInRandom} we sample random sets of up to 10 million points. The sampling follows the uniform, Gaussian or exponential distribution. After filtering, only a very small subset of Pareto-optimal points remains. The resulting Pareto sets contain less than 20 points for all three distributions. However, for the evaluation of the computation of Pareto sums these small input sets do not suffice. Thus, we propose further generators that yield larger, more practical input Pareto sets. 

\subsubsection{Incremental Generator}
The idea behind this generator is to build the Pareto set from the empty set by adding points that do not violate the Pareto-optimality of the current set. We start by sampling an $x$- and $y$-value in the given range and adding the first point $(x,y)$ to the Pareto set $A$. Until reaching the required size, we incrementally add points $p$ to the Pareto set $A$ for which the set $A \cup \{p\}$ is also a Pareto set. 

\begin{figure}
    \includegraphics[width=0.33\textwidth]{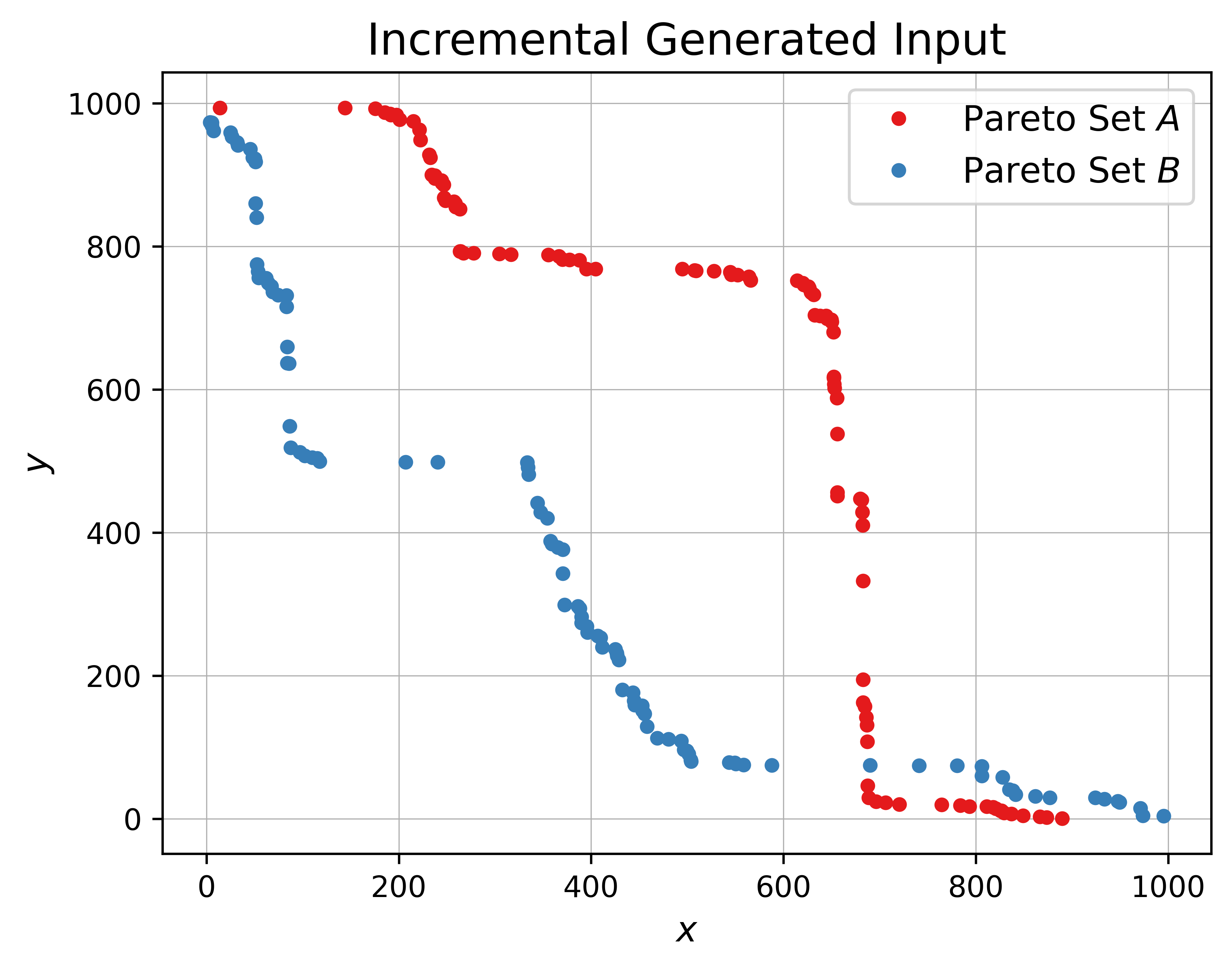}\hfill
    \includegraphics[width=0.33\textwidth]{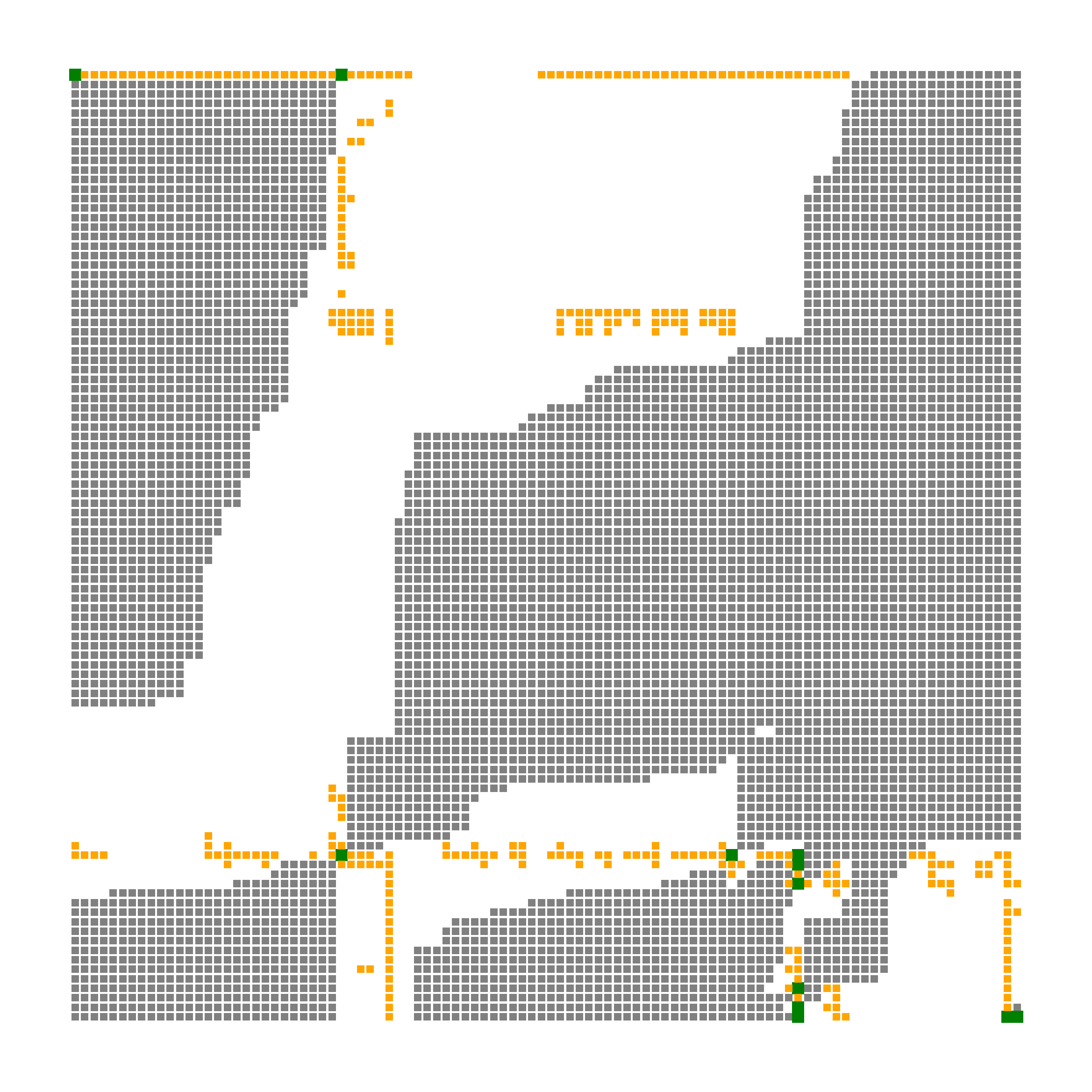}\hfill
    \includegraphics[width=0.33\textwidth]{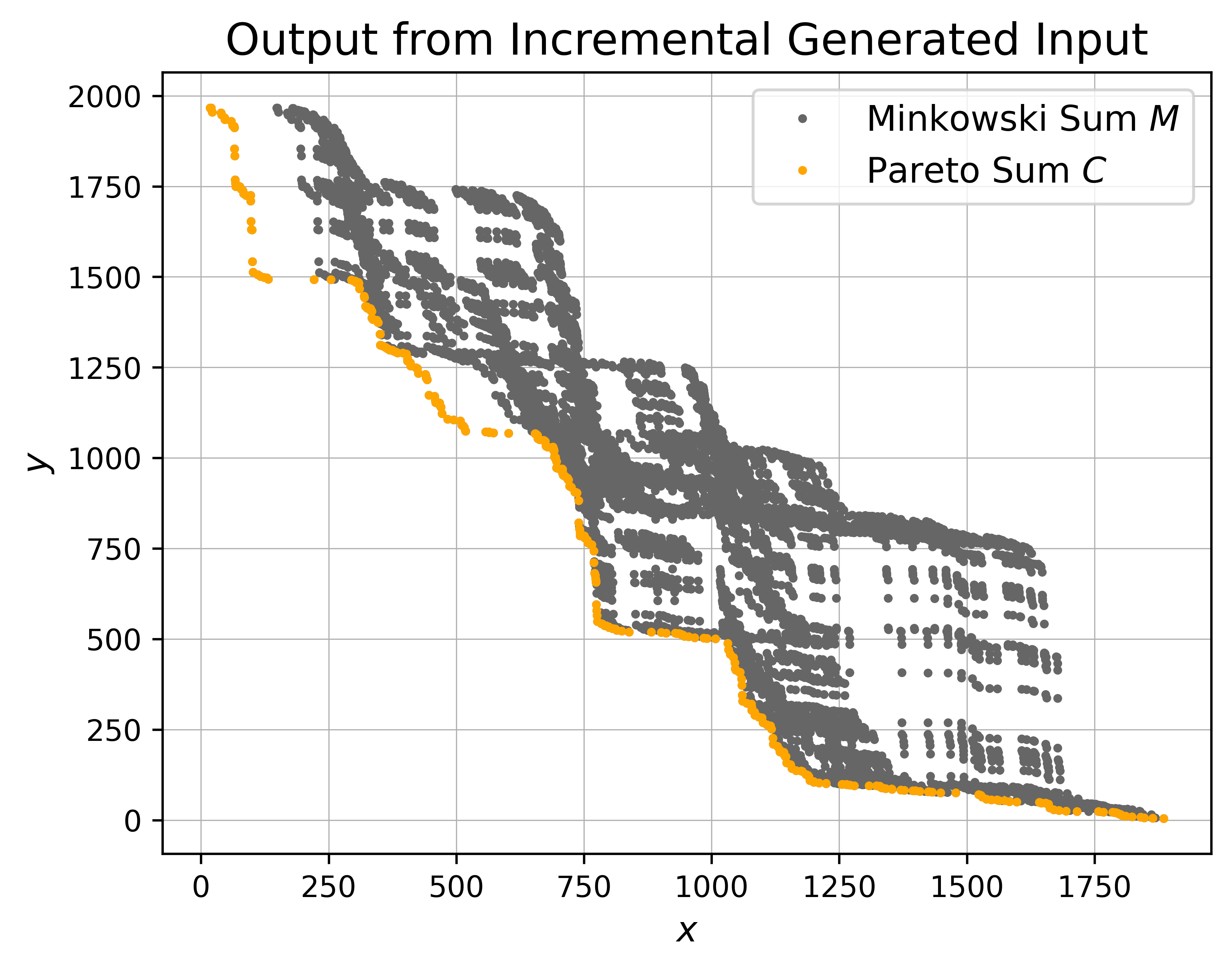}\\ 
    \caption{Example input Pareto sets $A$ and $B$ of size 100 generated by the incremental generator (left) and the schematic Minkowski matrix $M$ (middle) as described in Figure \ref{fig:priority}. The right image shows the Minkowski sum and Pareto sum using the $x$- and $y$-coordinates.}
    \label{fig:incrementalGen}
\end{figure}

The two input Pareto sets in Figure \ref{fig:incrementalGen} contain 100 points and are incrementally generated using the uniform distribution. This input set generation is time-consuming since the existing Pareto set already dominates a vast portion of the area in the range. Hence, it requires quite some time to find the next non-dominated point to add to the Pareto set. Obtaining large Pareto sets with this generator is impractical. 

\subsubsection{Sorted Generator}
The sorted generator leverages the fact that when sorting the points in a Pareto set by their $x$-values in increasing order, the $y$-values are already sorted decreasingly. We generate input Pareto sets by sampling two sequences of unique values in a given range. The first sequence is sorted increasingly and the second decreasingly. The values are assigned to the points such that the first point has the smallest $x$- and the largest $y$-value and the last point has the largest $x$-value and the smallest $y$-value. Since we do not sample any dominated points, the sorted Pareto set generator is more efficient than the previous generators. Hence, we can easily generate large Pareto sets for the evaluation of the proposed algorithms.

We consider uniform, Gaussian and exponentially distributed samples for the $x$- and $y$-values. In addition, we provide a variant of Pareto sets $A$ and $B$ where the respective $y$-coordinates are drawn from vastly different intervals, namely with upper bounds $\sqrt{n}$ and $n^2$ with $n$ being the number of points to be sampled. We call this a shifted distribution. Figure \ref{fig:sortedGen} shows example instances for each distribution.

\begin{figure}
    \includegraphics[width=0.33\textwidth]{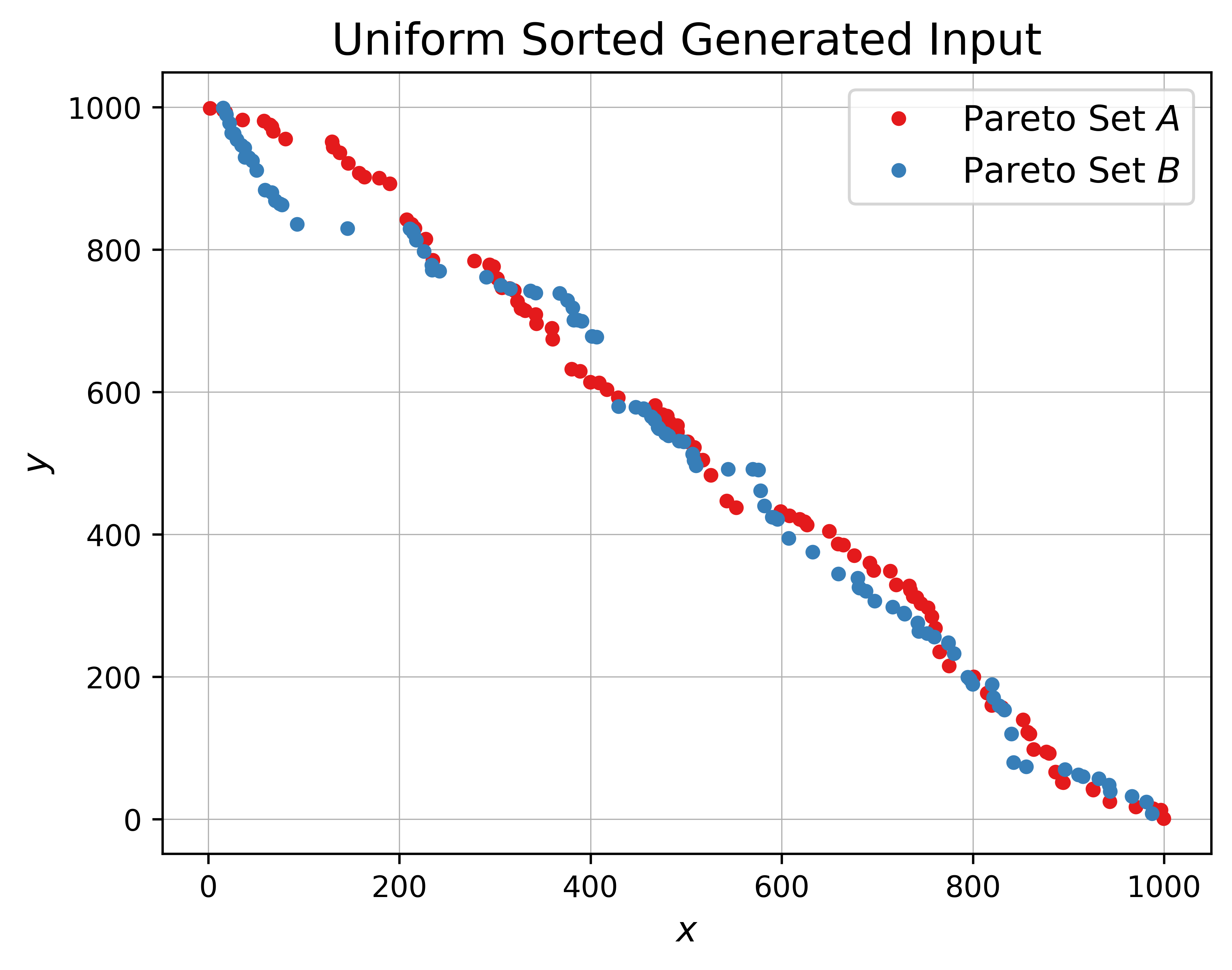}\hfill
    \includegraphics[width=0.33\textwidth]{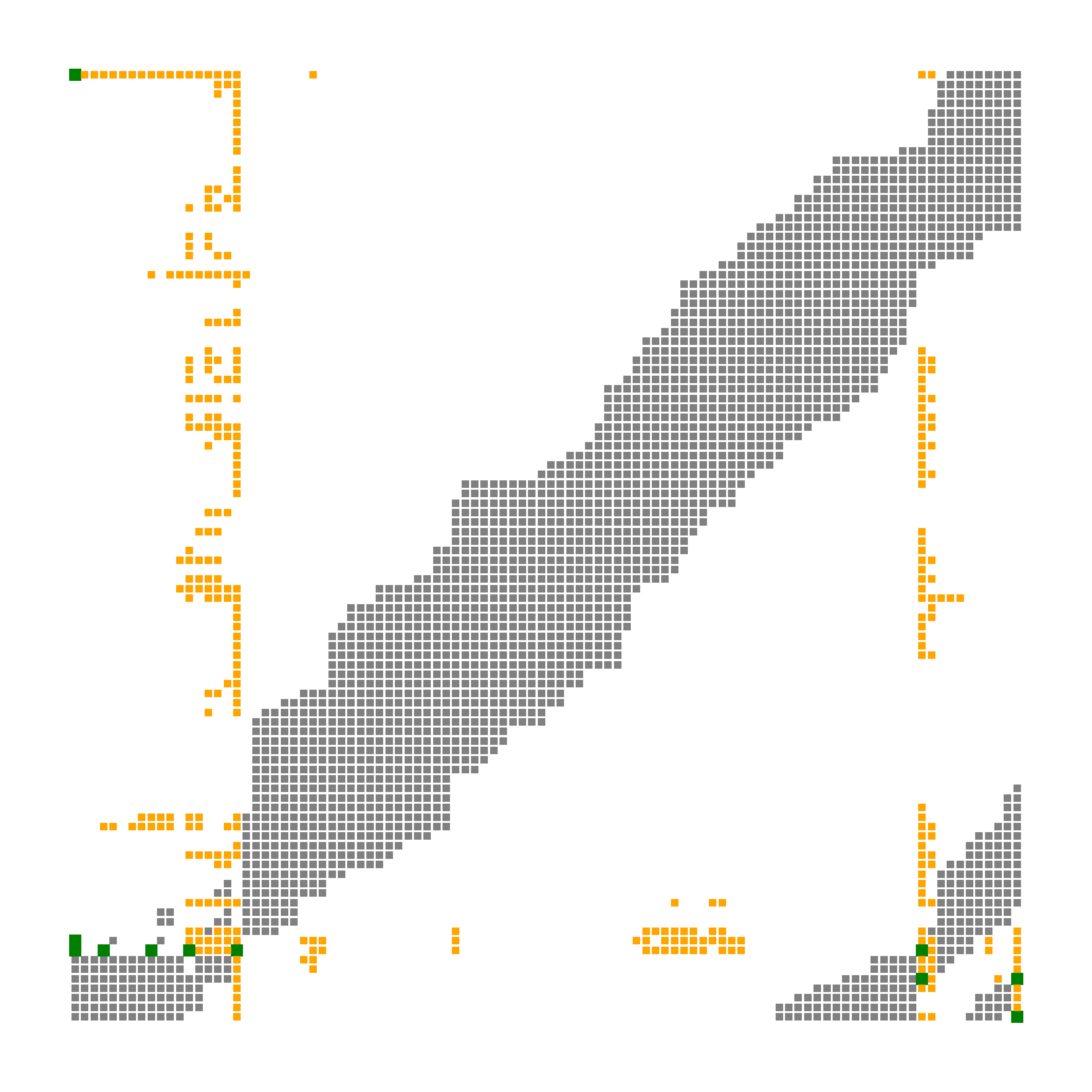}\hfill
    \includegraphics[width=0.33\textwidth]{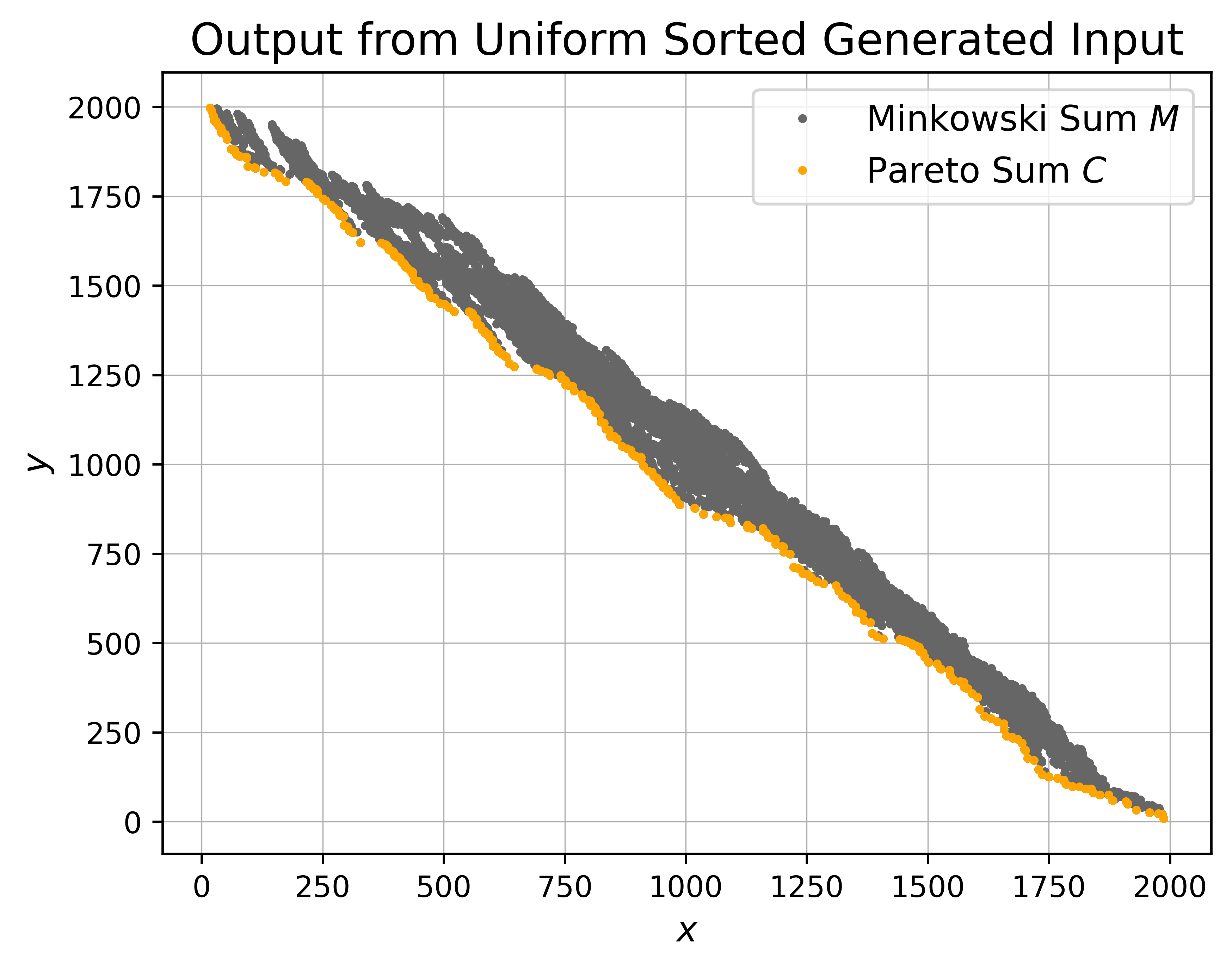}\\ 
    \includegraphics[width=0.33\textwidth]{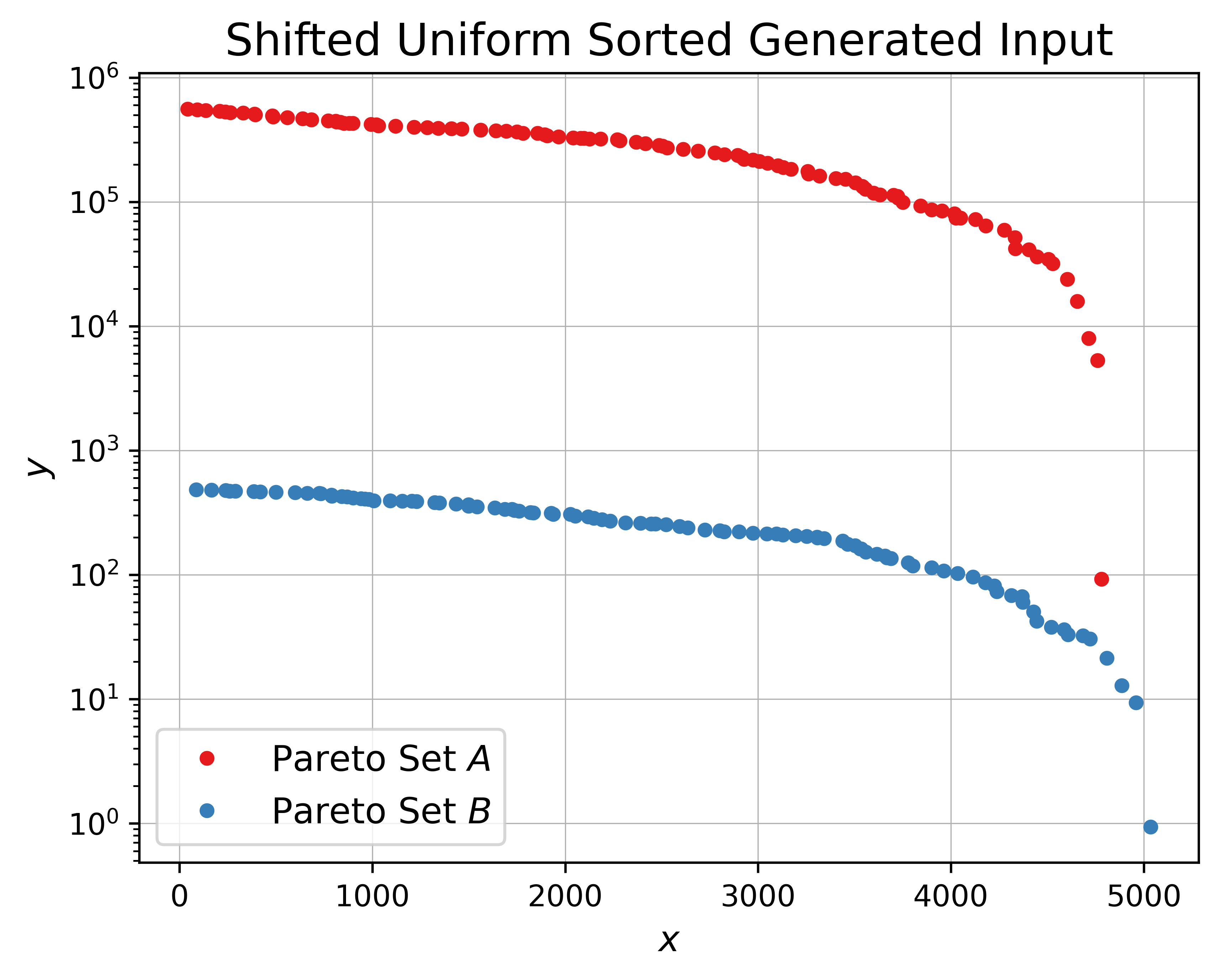}\hfill
    \includegraphics[width=0.33\textwidth]{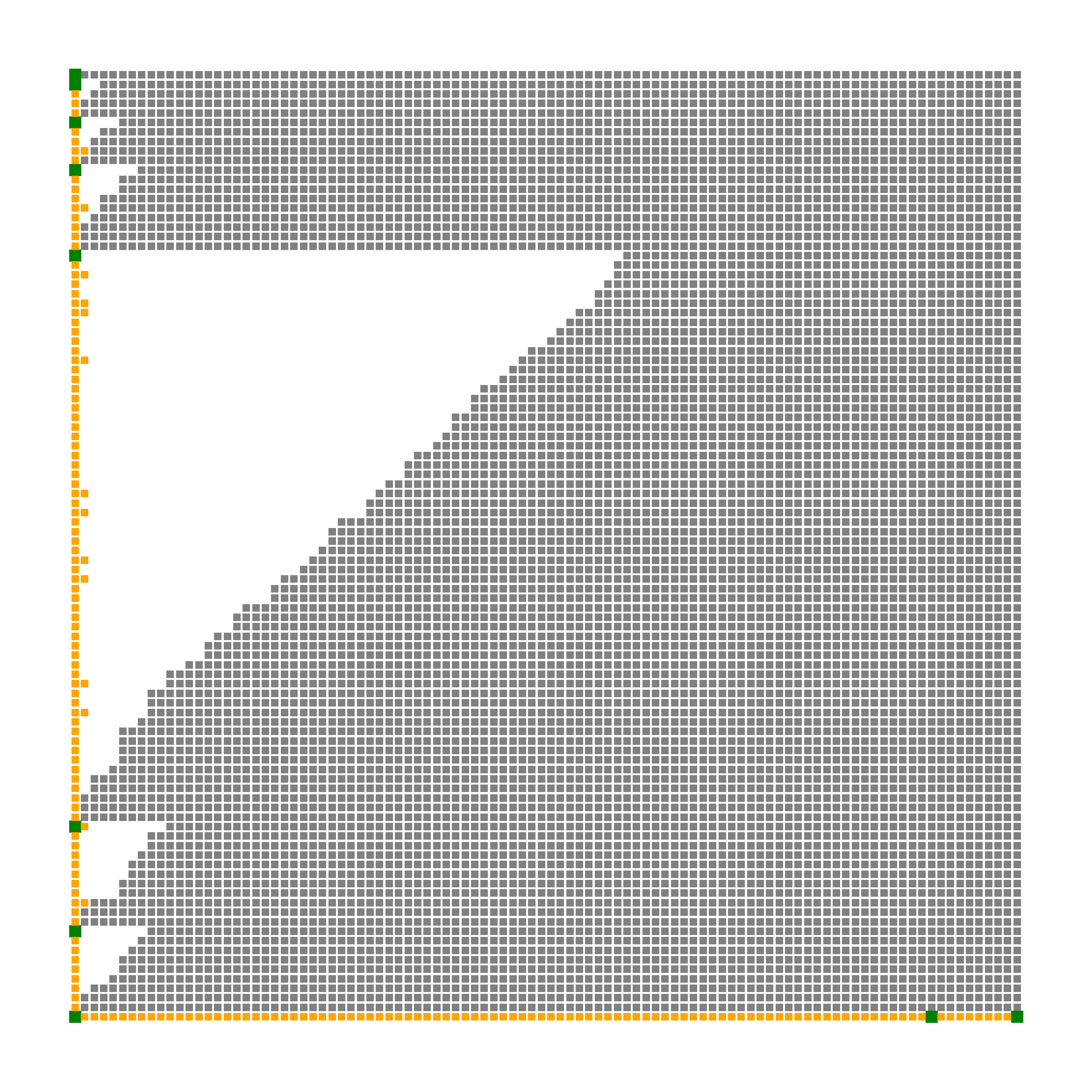}\hfill
    \includegraphics[width=0.33\textwidth]{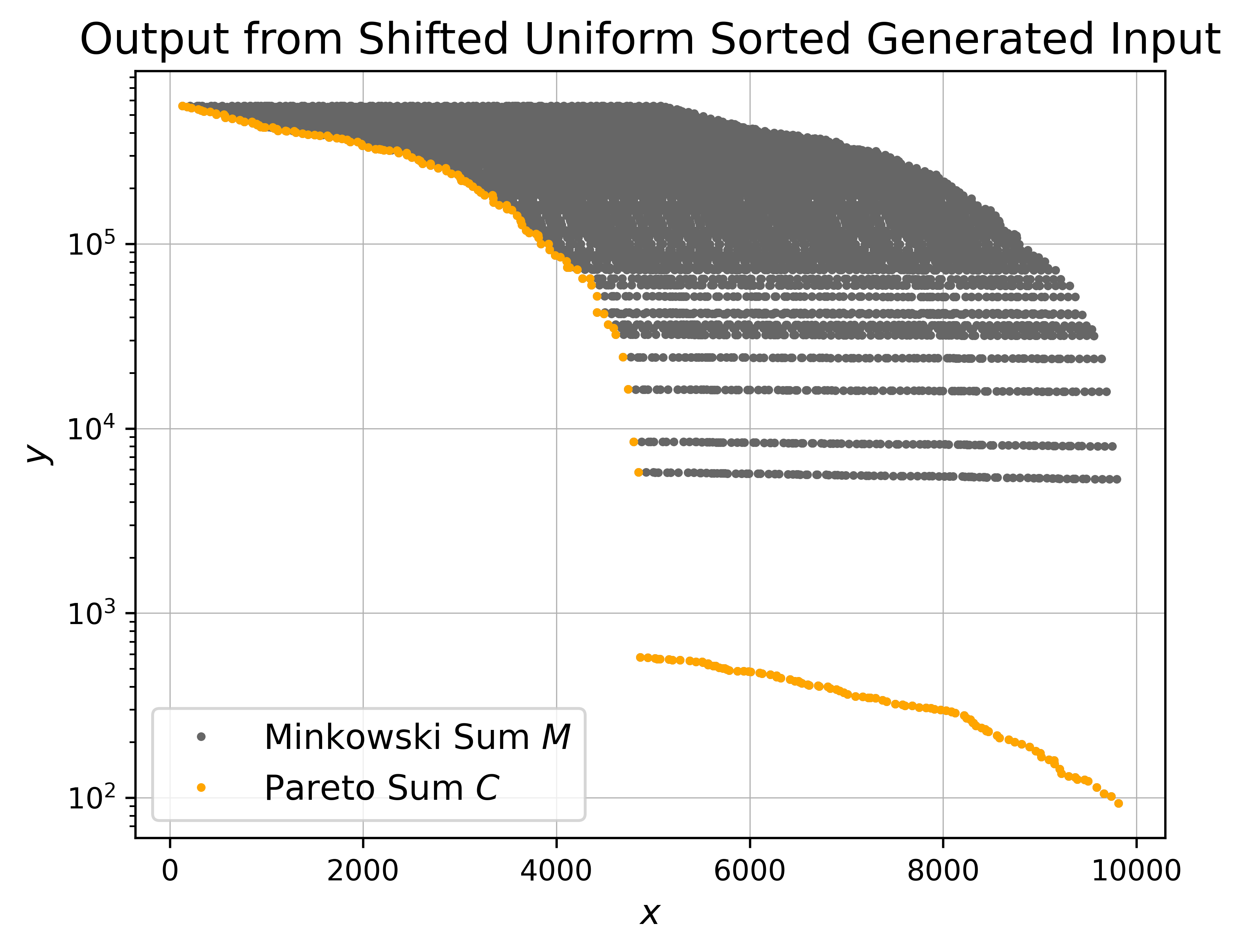}\\
    \includegraphics[width=0.33\textwidth]{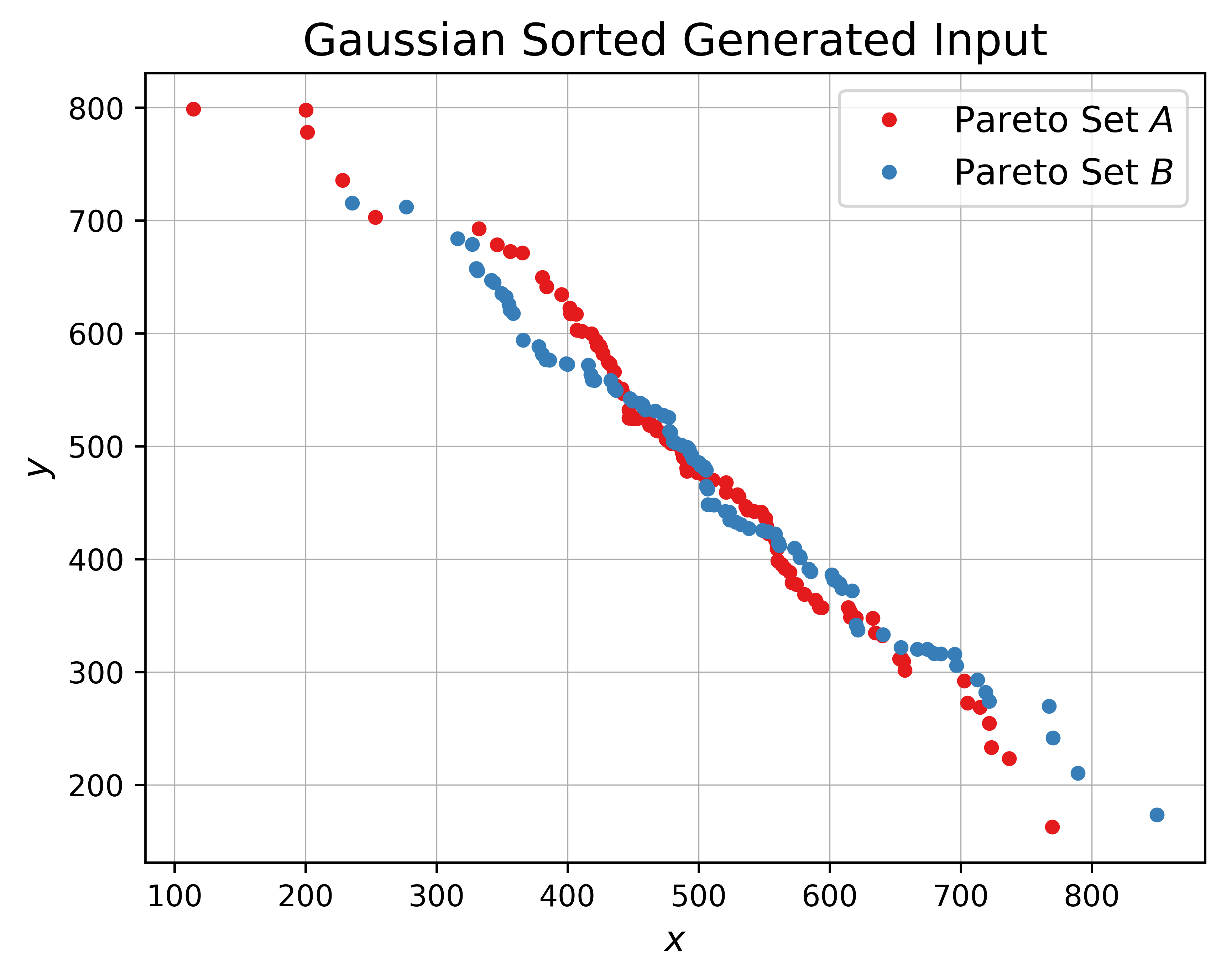}\hfill
    \includegraphics[width=0.33\textwidth]{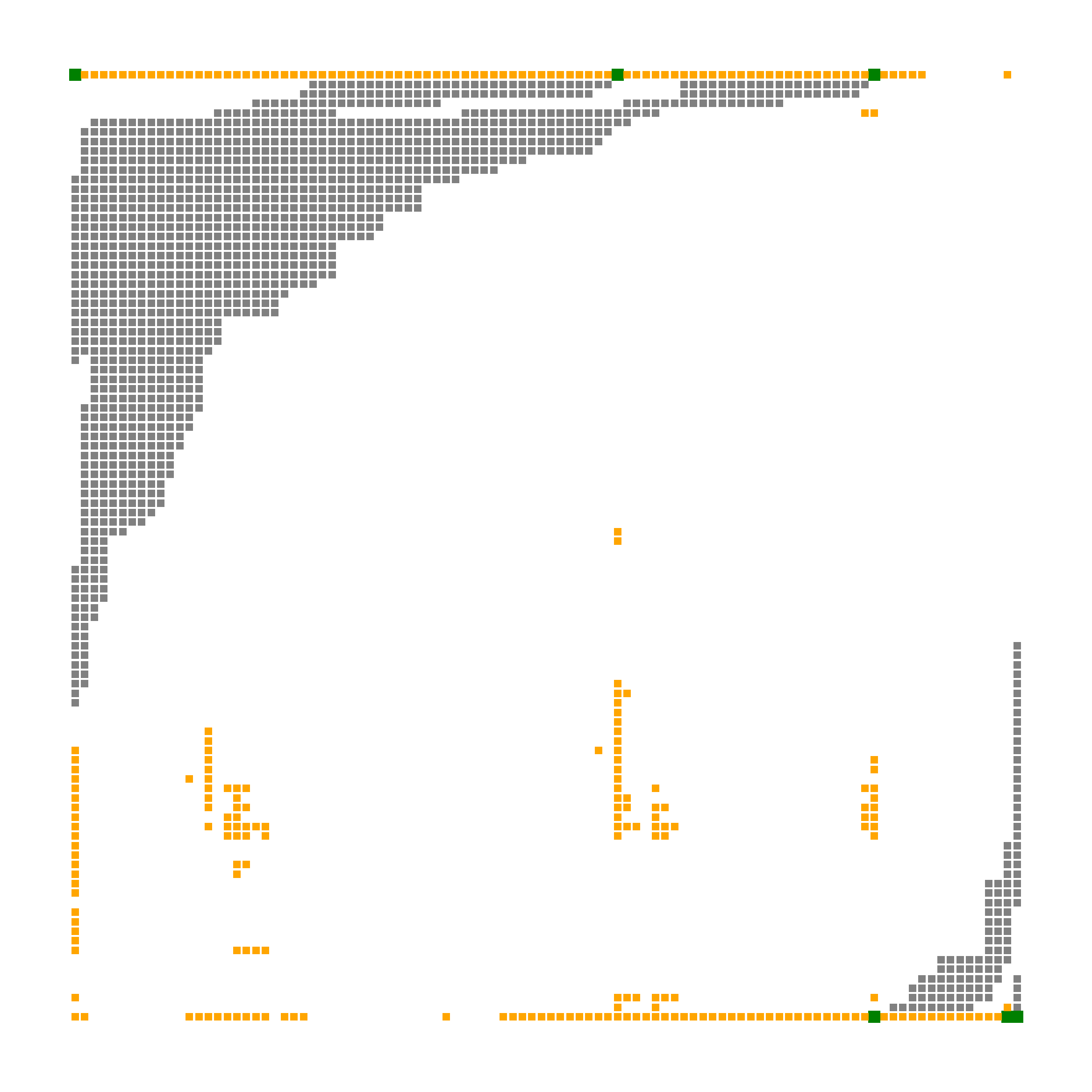}\hfill
    \includegraphics[width=0.33\textwidth]{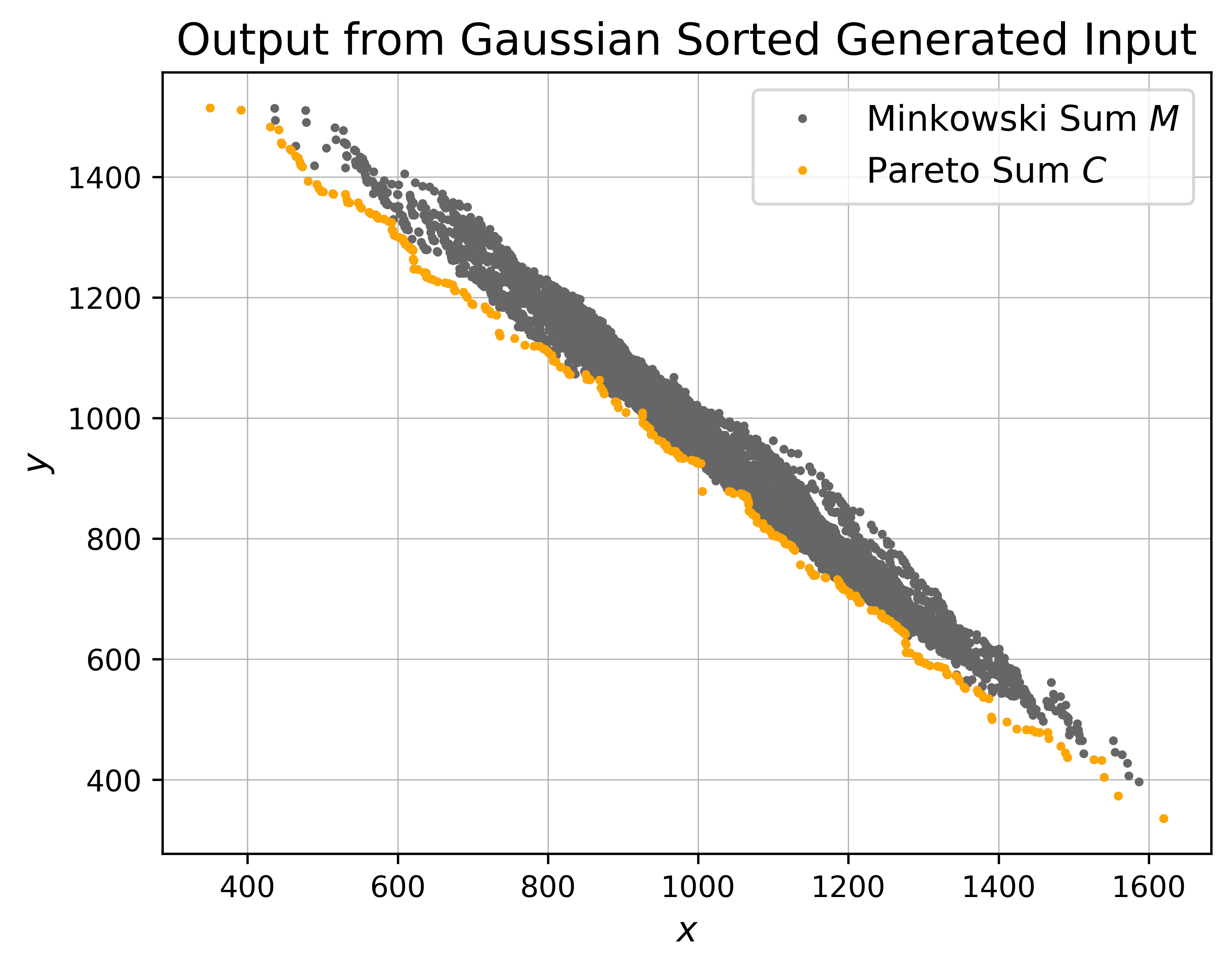}\\ 
    \includegraphics[width=0.33\textwidth]{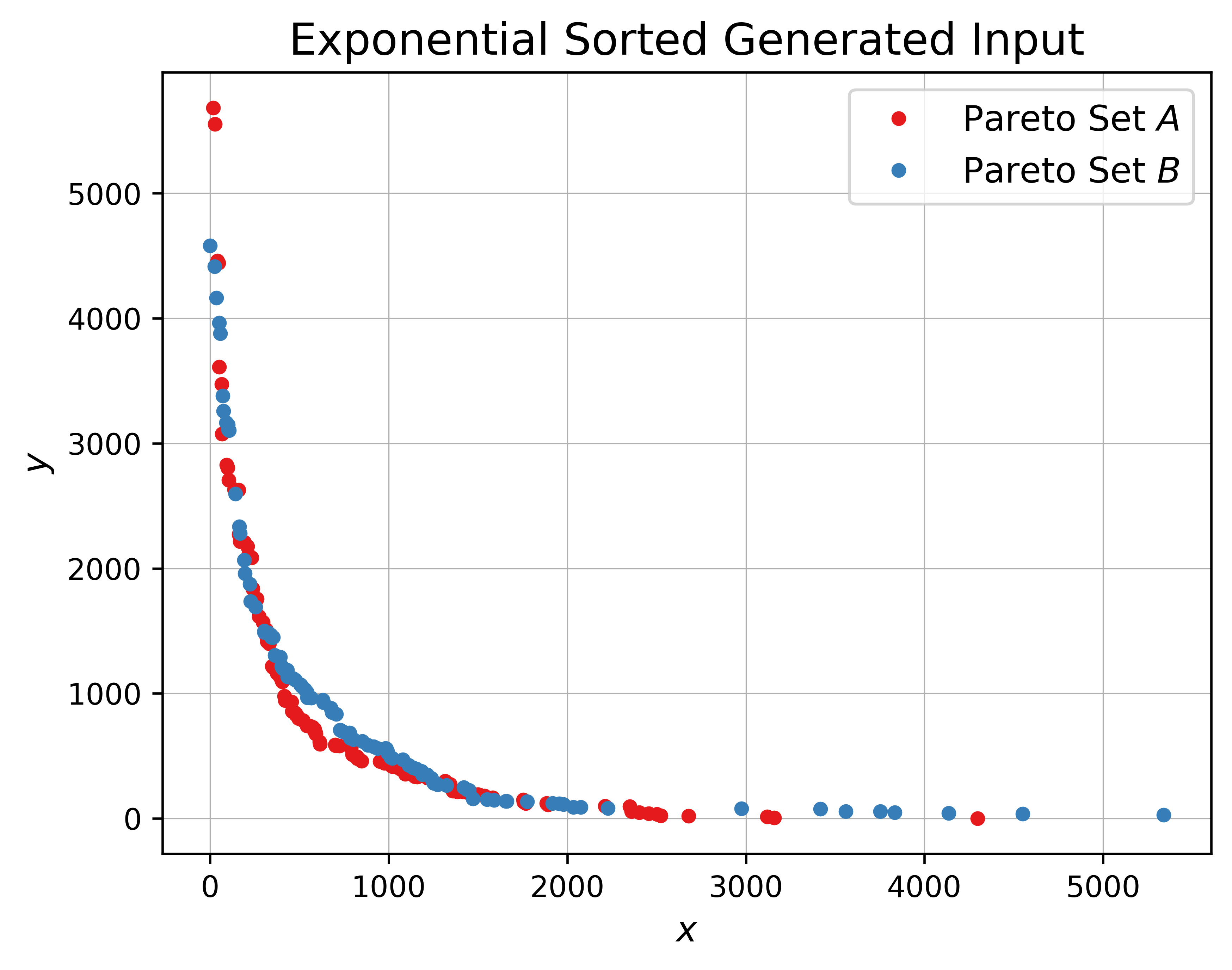}\hfill
    \includegraphics[width=0.33\textwidth]{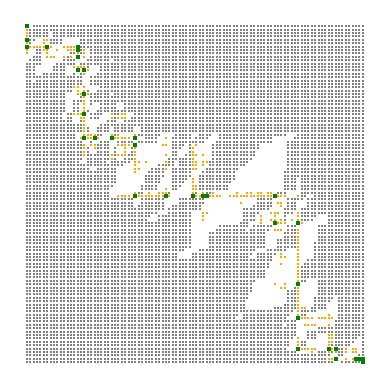}\hfill
    \includegraphics[width=0.33\textwidth]{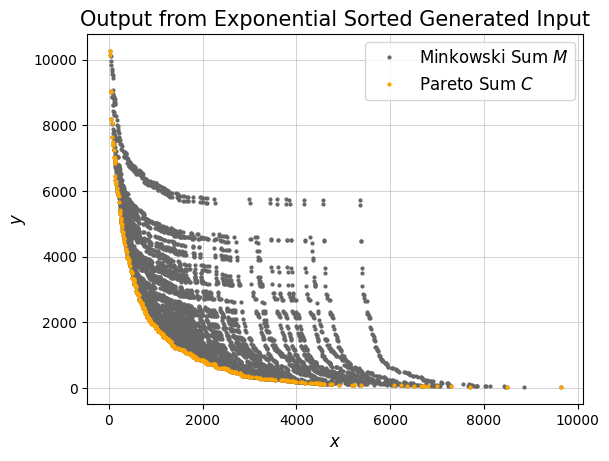}\\ 
    \caption{Input Pareto sets generated by the sorted generator using the uniform, shifted uniform, Gaussian, and exponential distribution for the sequences of $x$- and $y$-coordinates.}
    \label{fig:sortedGen}
\end{figure}

For the proposed output-sensitive algorithms we want to investigate the resulting size of the Pareto sum for different input Pareto sets. The left image in Figure \ref{fig:outputSizeGenerators} depicts the output sizes for the naive, incremental and sorted generators. The output sizes $k=|C|$ are relatively small with $k$ being roughly less than $4n$. For the three generators the resulting output sizes are fairly similar. As discussed earlier, the sorted generator is most efficient and thus will be used for the experiments.

\begin{figure}
    \includegraphics[width=0.495\textwidth]{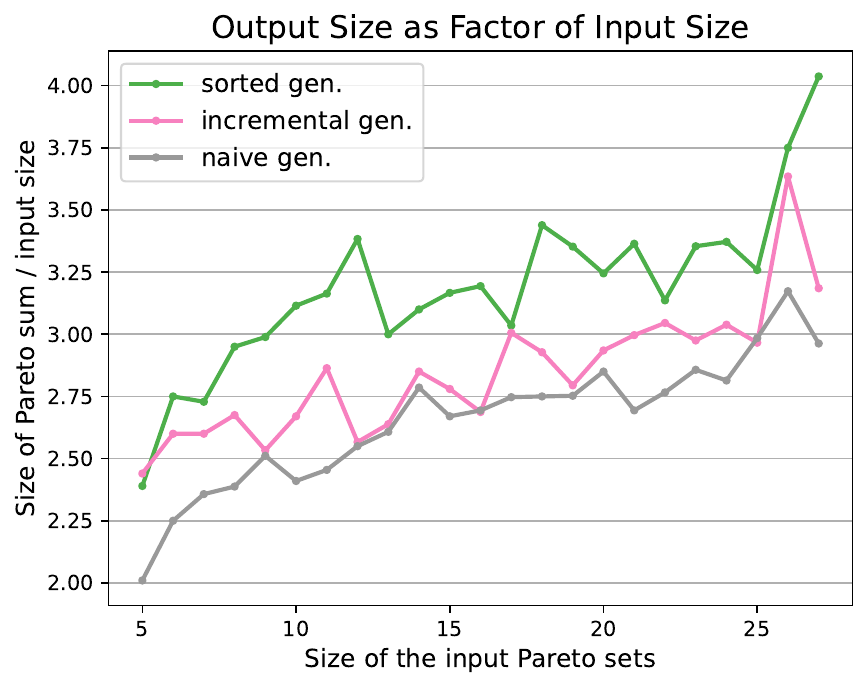}\hfill
    \includegraphics[width=0.495\textwidth]{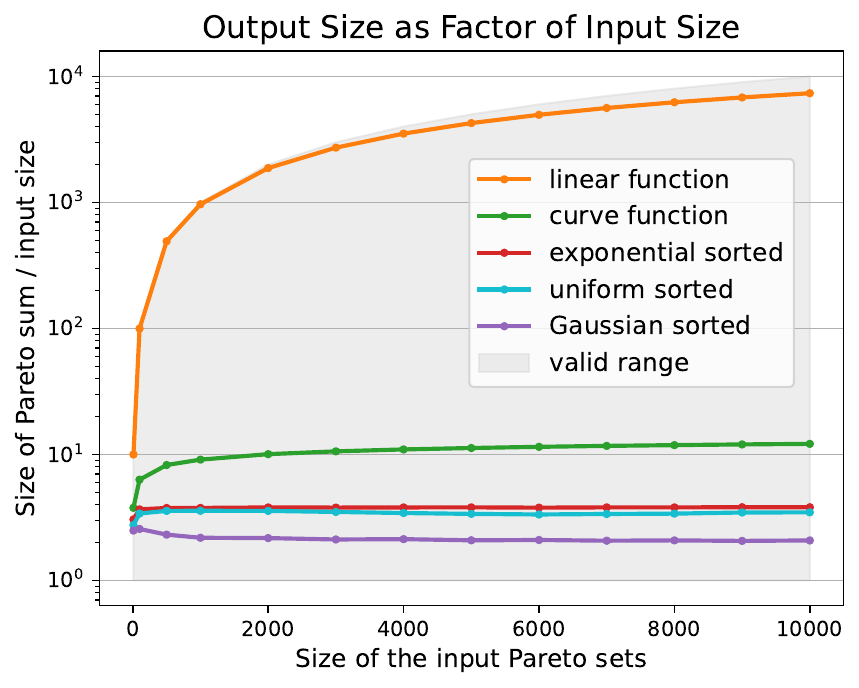}
    \caption{The output sizes $k=|C|$ of the Pareto sum are investigated for the different input generators. On the left, the output sizes as a factor of $n$ are given for the naive, incremental and sorted generators. The right image shows the output sizes as a factor of $n$ for the sorted generators following the uniform, Gaussian and exponential distribution and the two function generators. The shaded area marks the valid range between $n$ and $n^2$ for the output size.}
    \label{fig:outputSizeGenerators}
\end{figure}

\subsubsection{Function Generator}
Since we want to evaluate the influence of the output size on the performance of the algorithms, we propose the function generator which produces larger output sizes than the aforementioned generators. The idea is to sample a random $x$-value in the range $(0,n]$ and the corresponding $y$-value is determined by a function. 
We provide two variants, the first one uses a curve-shaped function where $y=0.3/x$ and the second follows a linear function, in this case $y=-x+n$. Example Pareto sets stemming from the curve and linear function generator are depicted in Figure \ref{fig:functionGen}.

Similarly to the previous generators, the right image in Figure \ref{fig:outputSizeGenerators} shows the resulting output size of Pareto sums when using the function generators compared to the sorted generator following the uniform, Gaussian and exponential distributions. The gray shaded area encloses the valid range $n \leq |C| \leq n^2$ for the output size. As shown in Figure \ref{fig:functionGen} the input sets produced by the linear function generator result in every point in the Minkowski sum $M$ being Pareto-optimal, hence $k=|C|= n^2$. The curve function provides output sizes between the small ones from the sorted generator and the large ones from the linear function generator which is advantageous for a thorough analysis of the proposed algorithms.

\begin{figure}
    \includegraphics[width=0.33\textwidth]{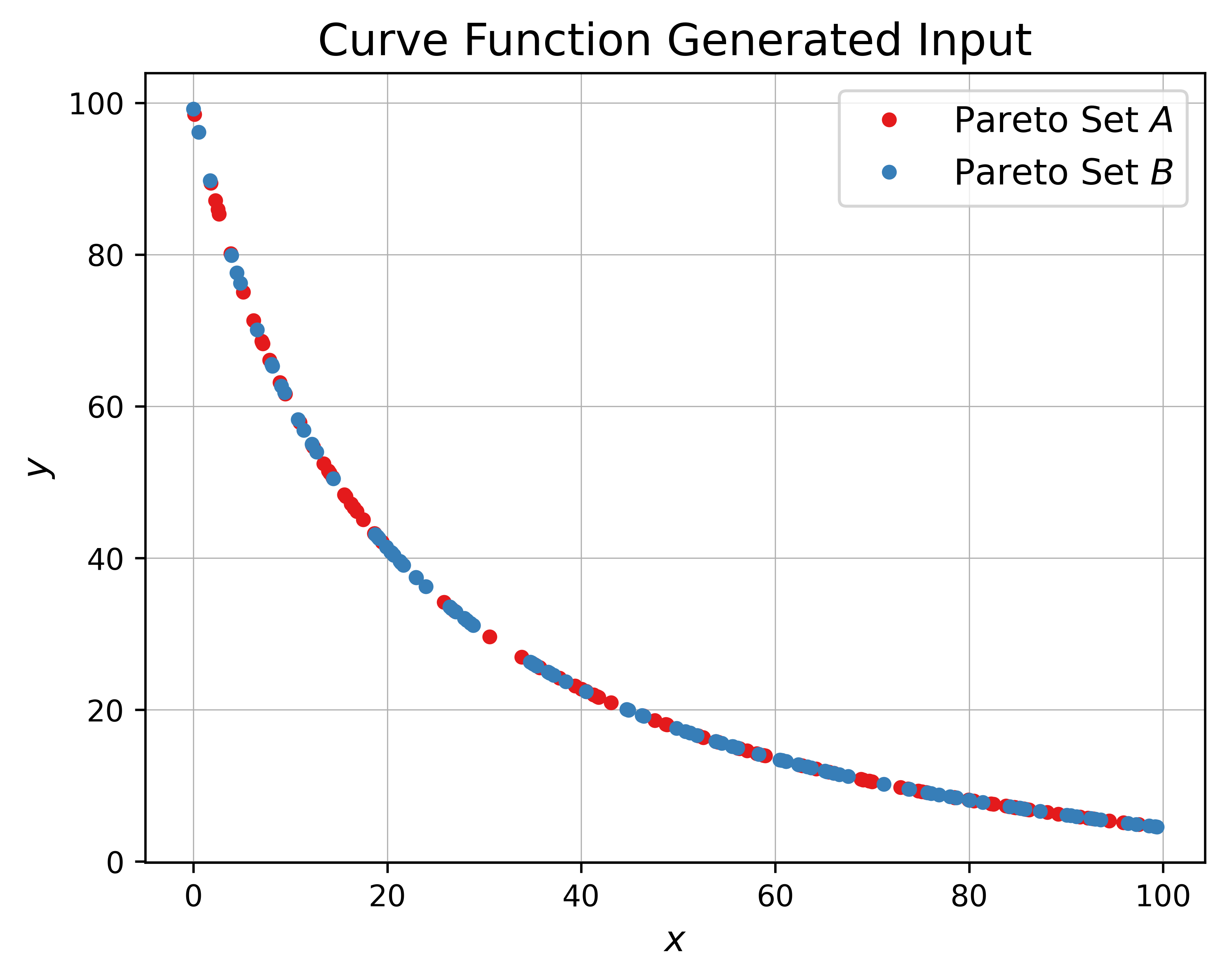}\hfill
    \includegraphics[width=0.33\textwidth]{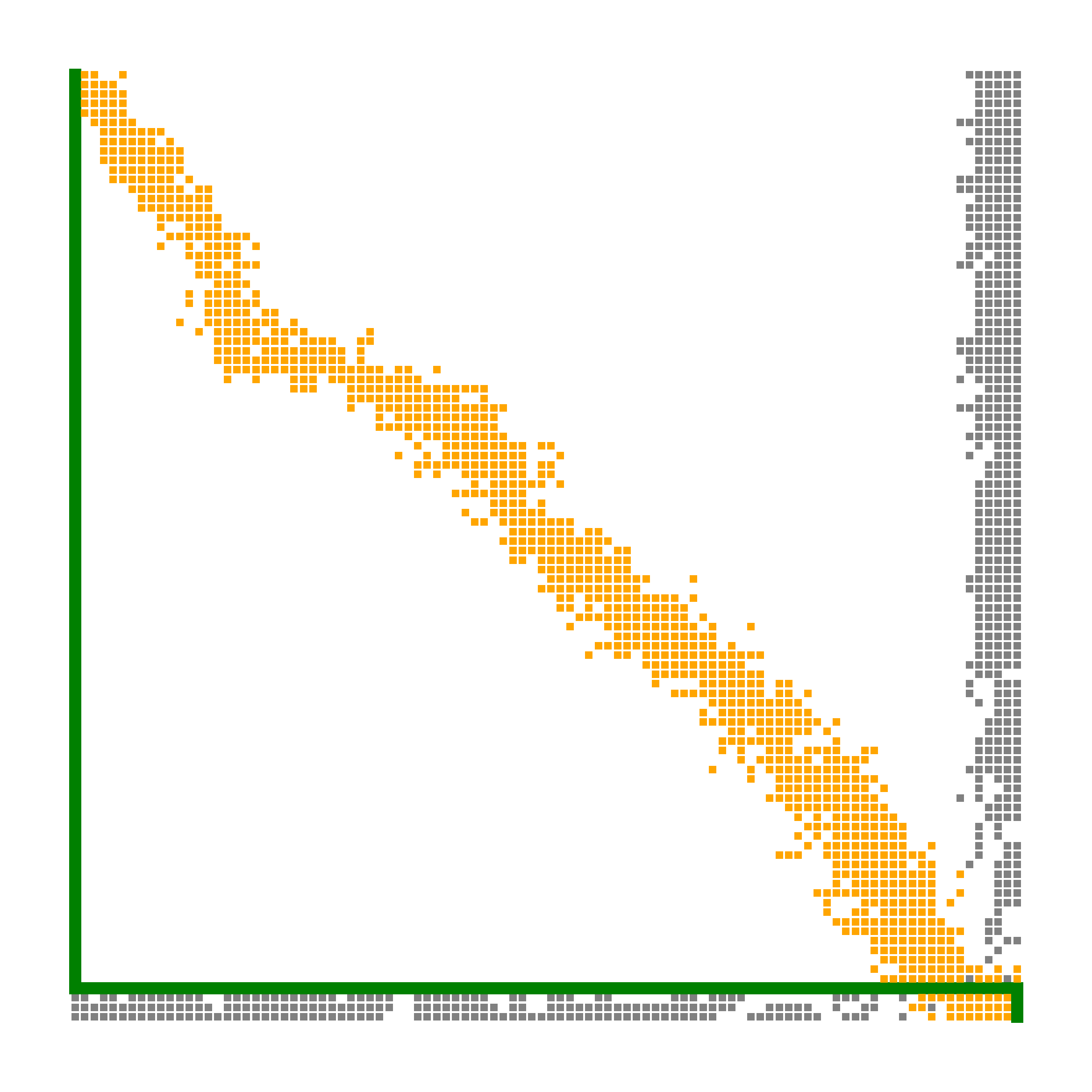}\hfill
    \includegraphics[width=0.33\textwidth]{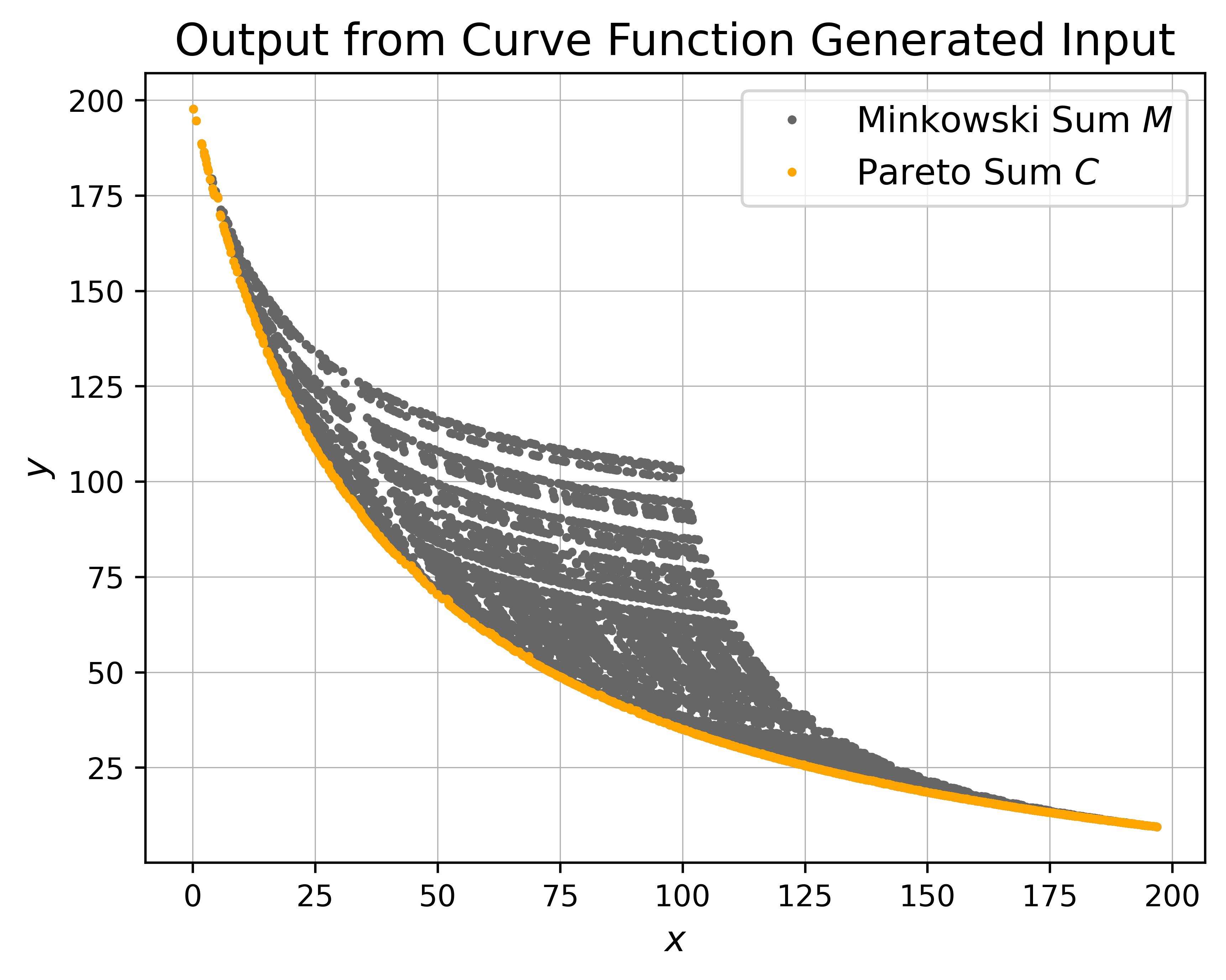}\\
    \includegraphics[width=0.33\textwidth]{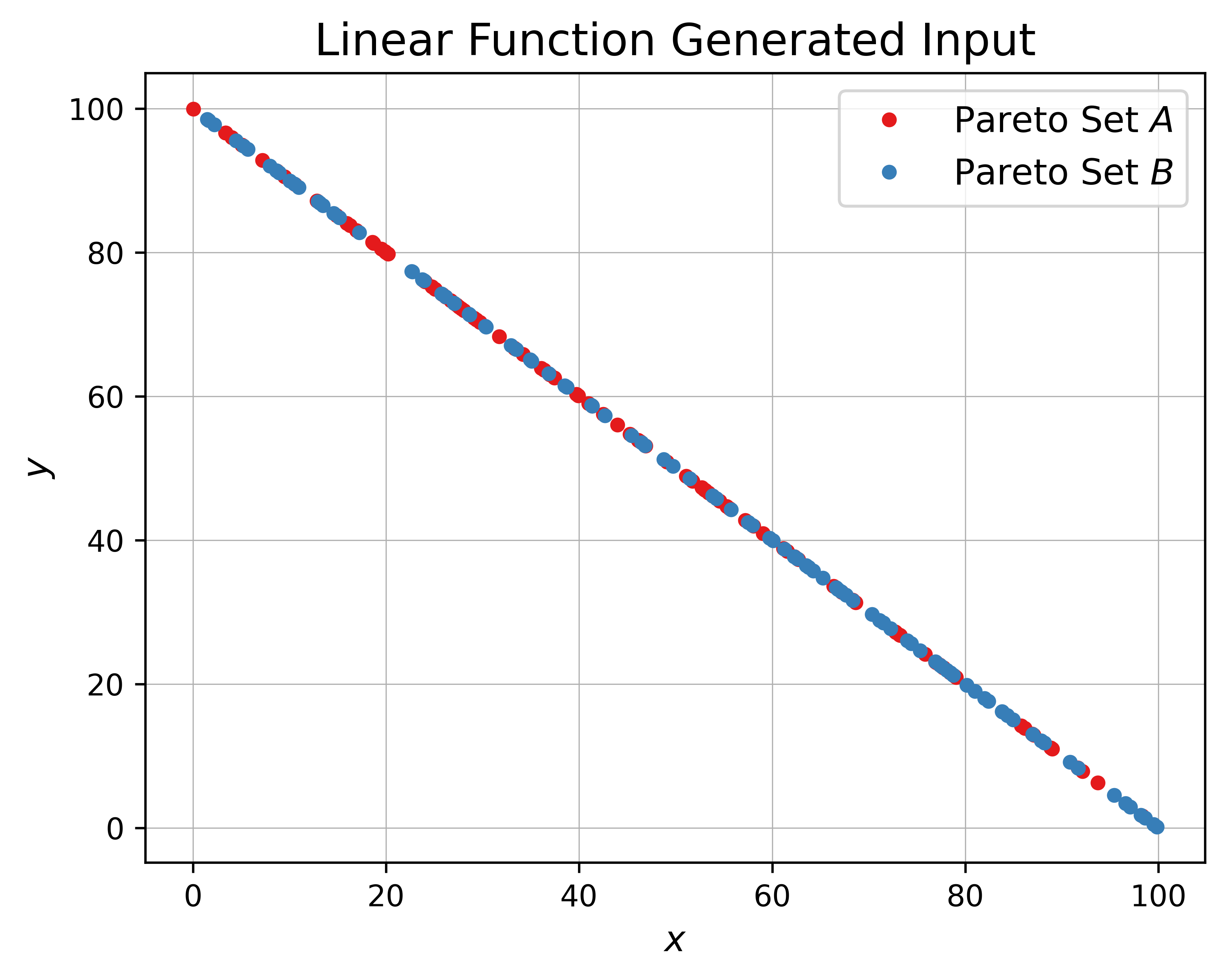}\hfill
    \includegraphics[width=0.33\textwidth]{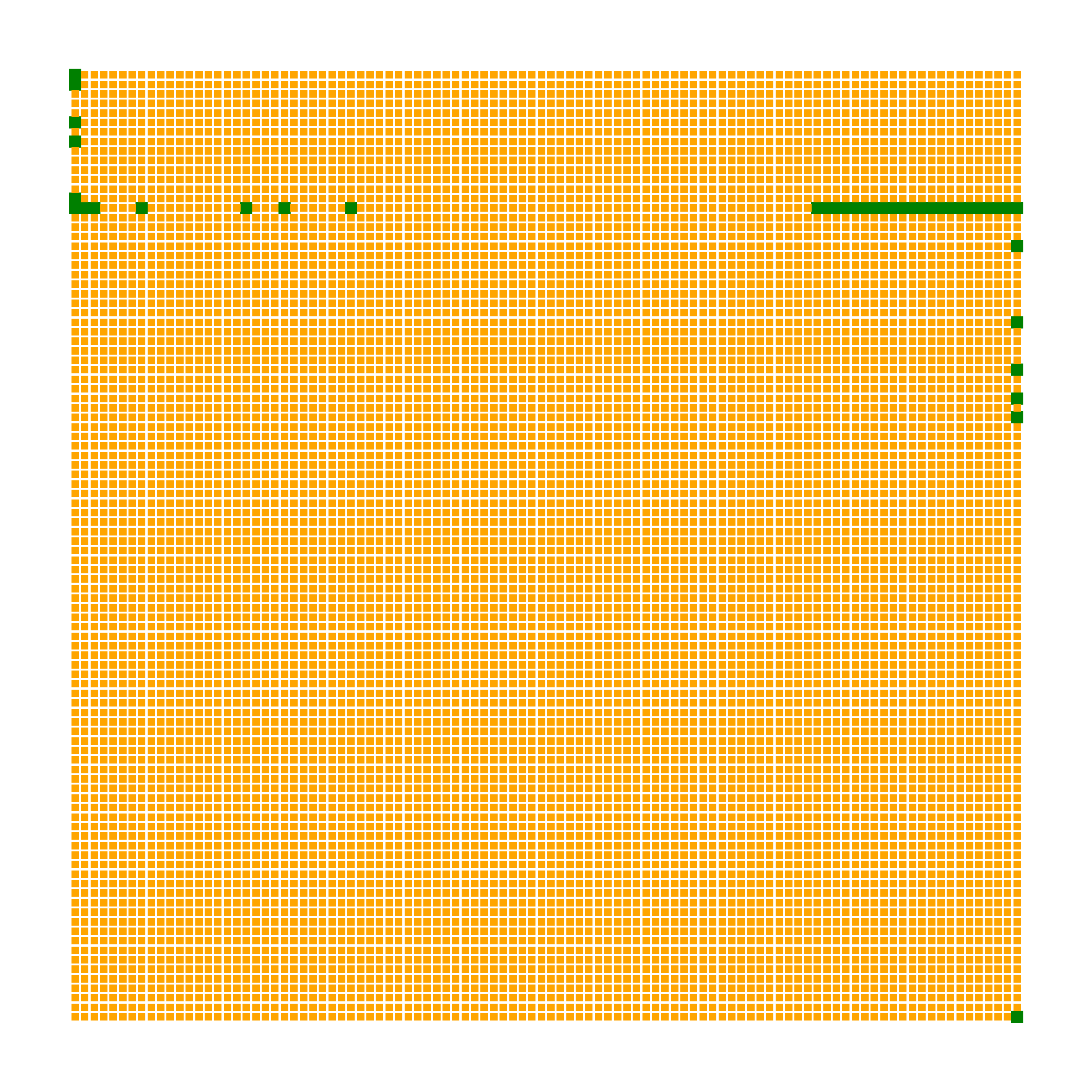}\hfill
    \includegraphics[width=0.33\textwidth]{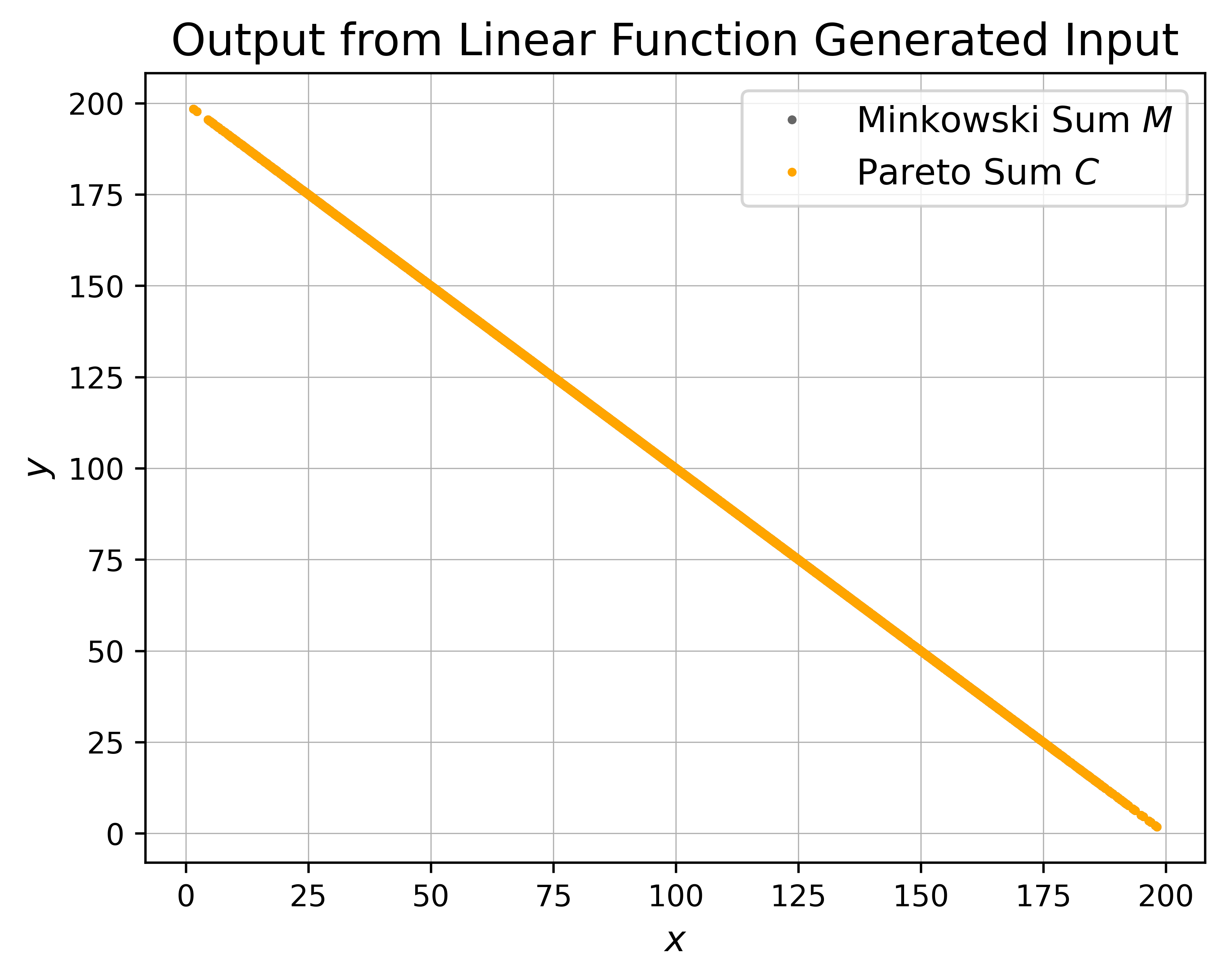}\\
    \caption{Example input Pareto sets based on the function generator. The images on top show the input sets, Minkowski matrix and Pareto sum for the curve function generator and the bottom images for the linear function generator.}
    \label{fig:functionGen}
\end{figure}

\subsection{Running Time Assessment on Generated Data}
Figure \ref{fig:RuntimeSmall} depicts the average running times on the five input variants stemming from the above generators, namely the sorted generator using the uniform, Gaussian and exponential distribution and the function generator with the curve and linear functions. As derived in our theoretical analysis, the Brute Force algorithm is by far the slowest on all five distributions such that the execution already exceeded the time limit for input sizes above $n=3,000$. The engineered Binary Search approach prioritizing columns where Minkowski points are located has an improved running time which is up to one order of magnitude better than for the BS algorithm. 

On uniform and Gaussian distributed input sets the NonDomDC algorithm \cite{klamroth2022efficient} exhibits slower running times for the doubling variant than the sequential version (SND). As for the other input generators, the doubling approach (DND) was faster. For both algorithms, the intermediate output size influences the performances of the algorithms. During the merging of intermediate Pareto sets, either in a MergeSort-like fashion (DND) or sequentially processing the columns one by one (SND), we observe a blow-up with the structure of these algorithms. Figure \ref{fig:intermediateND} examines said intermediate output size with respect to the actual output size $k=|C|$ for Gaussian distributed input sets. The DND, approach requires $65\%$  more space than the actual size $k$ whilst considering a significant amount of points within the matrix which are not relevant for the final Pareto sum. For SND, the effect is smaller but still worth mentioning since it maintains intermediate Pareto sets which exceed the output size $k$ by up to $23\%$. Due to that observation, the SND algorithm is one order of magnitude faster than the DND approach for the Gaussian distributed input sets. As for the linear function generated inputs, the DND algorithm is significantly faster than the sequential variant by two orders of magnitude. This might be due to rebuilding the Pareto tree after $520$ insertions for the SND approach, as recommended by the authors in \cite{lang2022space}. For linear function generated inputs sets, virtually all points in the Minkowski matrix are in the final Pareto sum as depicted in Figure \ref{fig:outputSizeGenerators}. Together with maintaining only one Pareto tree during the merging of columns in the SND algorithm, the rebuilding of the entire tree is performed more often than for input generators which result in smaller values of $k$. For DND the trees representing the intermediate Pareto sets are stored for each level in the MergeSort-like procedure. Hence, the rebuilding of the trees does not take place and does not influence the running time of the algorithm. 

\begin{figure}[!ht]
    \includegraphics[width=0.495\textwidth]{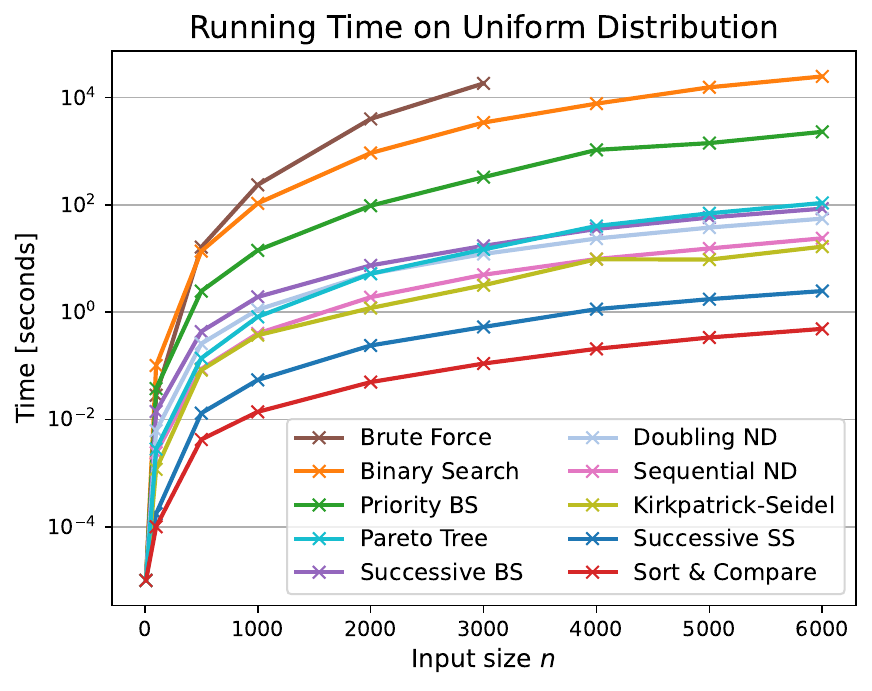}\hfill
    \includegraphics[width=0.495\textwidth]{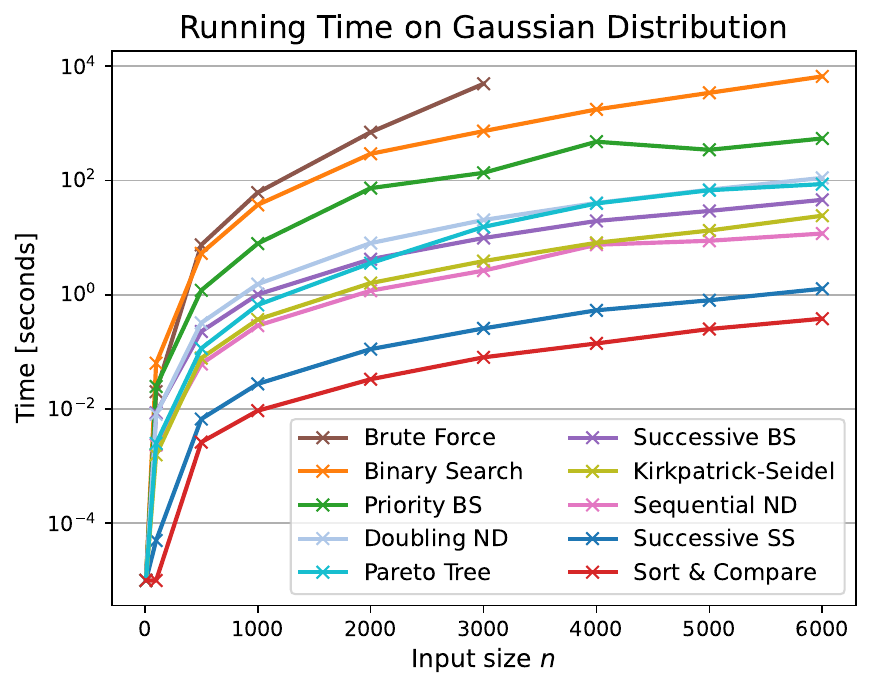}\\
    \includegraphics[width=0.495\textwidth]{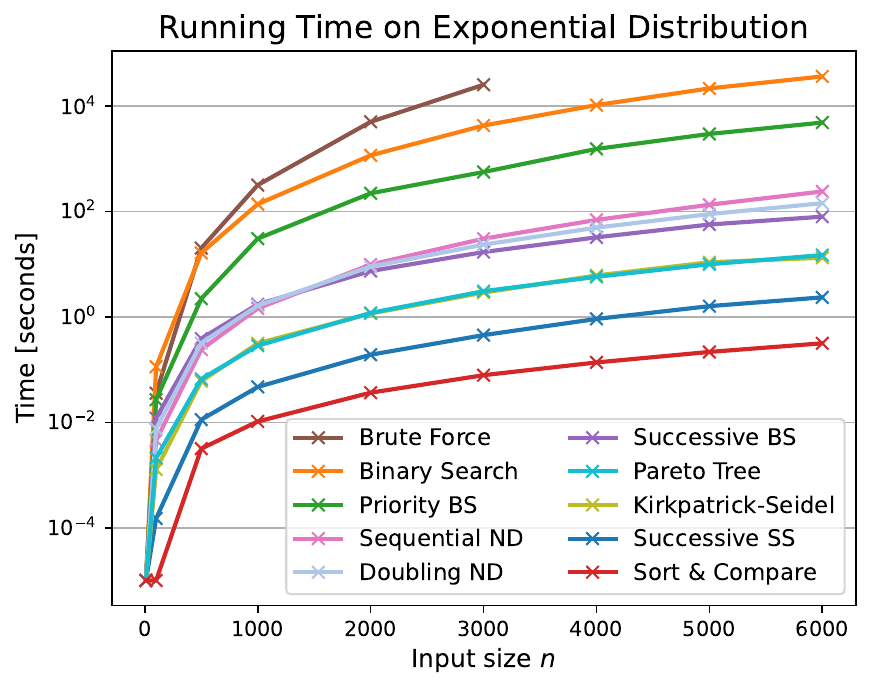}\hfill 
    \includegraphics[width=0.495\textwidth]{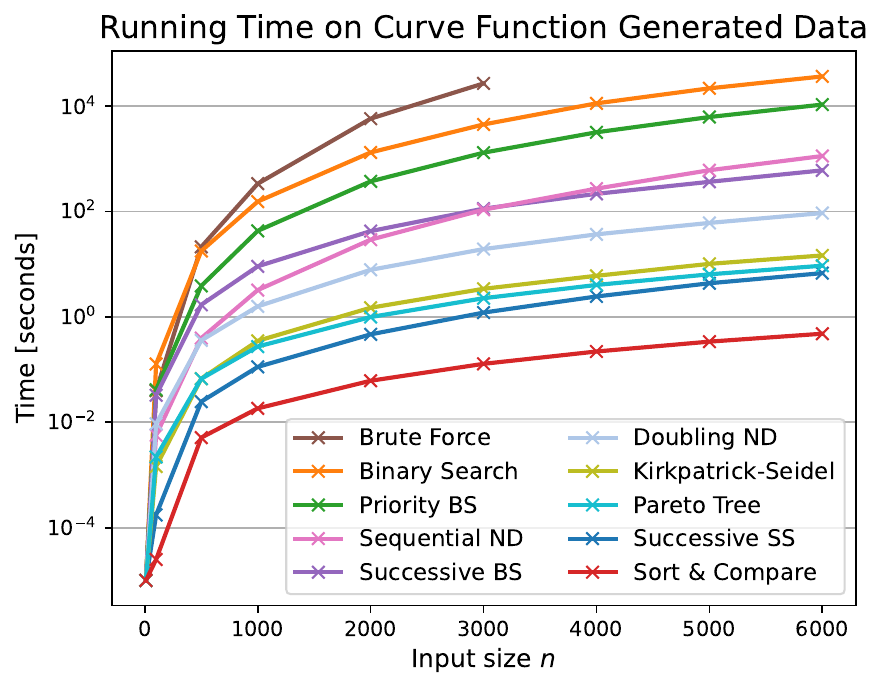} \\
    \includegraphics[width=0.495\textwidth]{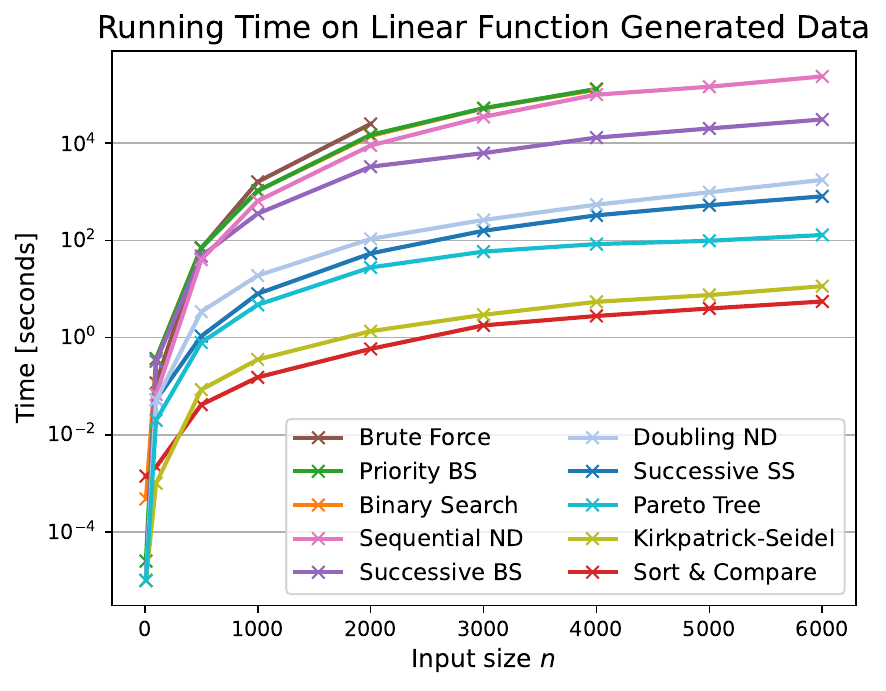} \hfill
  \caption{Average running times of the baseline algorithms BF, BS, PBS, SC, the known approaches NonDomDC and Kirkpatrick-Seidel, the successive algorithms SBS, SSS and the algorithm based on Pareto trees on five input generators.}
        \label{fig:RuntimeSmall}
\end{figure}

\begin{figure}[!ht]
    \centering
    \includegraphics[width=0.6\textwidth]{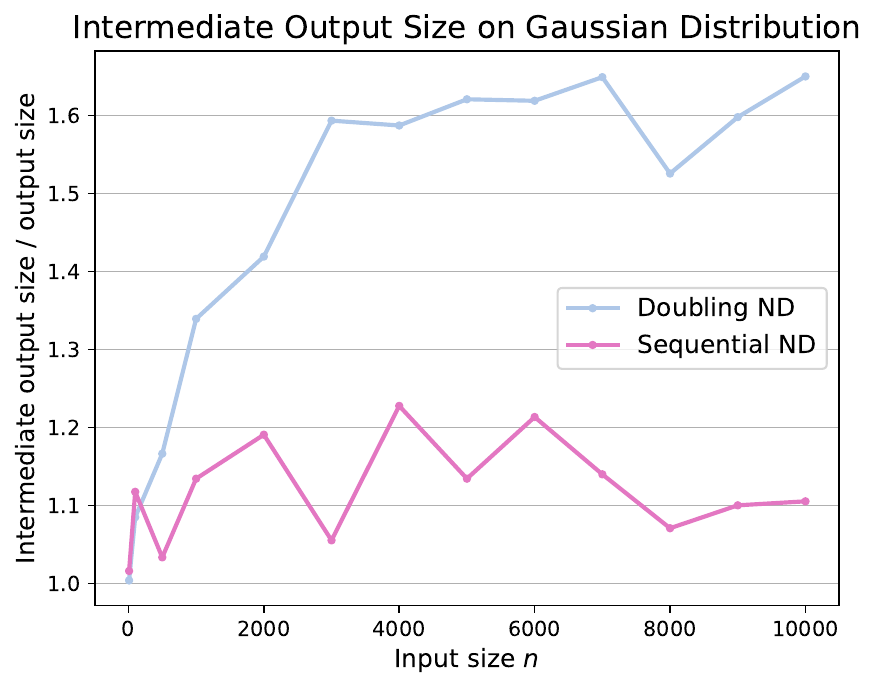}
    \caption{Intermediate output sizes of the sequential and doubling NonDomDC algorithm as a factor of the actual output size $k=|C|$ for inputs generated using the sorted Gaussian generator.}
    \label{fig:intermediateND}
\end{figure}

The running times of the simple Pareto tree algorithm perform comparably to the SND and DND algorithm, which also incorporate the Pareto trees, on uniform and Gaussian distributed input sets. However, it is remarkable that for the sorted exponential generator as well as the curve and linear function generators, the Pareto tree algorithm significantly outperforms the more elaborate SND and DND algorithms. Moreover, the Pareto tree algorithm requires less space than the ND algorithms since it only stores the final Pareto sum in the tree with non-Pareto points never being added to the tree as is done for both the doubling and sequential variant of ND. For these distributions, there is no advantage in using the NonDomDC algorithm opposed to directly leveraging only the data structure, the Pareto trees. For the curve input sets, the Pareto tree algorithm is the fourth best algorithm. As for the linear function generator, it ranks third outperforming SSS and only KS and SC having faster running times. 

The Kirkpatrick-Seidel algorithm \cite{kirkpatrick1985output} has a similar running time as the ND algorithms on uniform and Gaussian distributed input Pareto sets. However, the KS algorithm requires significantly more space than the output-sensitive algorithms, especially considering that the output sizes on these distributions are less than $4n$ whereas the KS algorithm stores the entire matrix of size $n^2$. Hence, storing all points of the Minkowski matrix is especially wasteful in these scenarios. For input sets generated by the exponential and the curve function generator, the KS algorithm performs better than ND and is similar to the Pareto tree algorithm, whilst the Pareto tree algorithm, again, consumes less space. Regarding linear function generated inputs, the quadratic space consumption of the KS algorithm and the output-sensitive algorithms are essentially the same. On those inputs, the KS approach ranks second in terms of running time only being outperformed by the Sort \& Compare algorithm.

Considering the running times of the successive algorithms, the Successive Binary Search algorithm (SBS) is faster than the SND approach on all generators except for the uniform and Gaussian. However, as for all our proposed algorithms, the SBS algorithm uses less space than the ND algorithms. Yet the SBS algorithm outperforms its baseline counterparts BS and PBS by at least one order of magnitude and up to two orders of magnitude for the uniform, Gaussian and exponential distributed inputs. In general, the running time of the successive algorithms heavily depends on how often we call the range minimum oracle in order to determine the next Pareto point. To this end, the oracle is called for each point in the Pareto sum, thus $k$ times. Consequently, the running time of these approaches is better on input Pareto sets which result in a smaller Pareto sum, hence in a smaller value of $k$ and less calls to the range minimum oracle. This is the case for the uniform, Gaussian and exponential distributions as shown in Figure \ref{fig:outputSizeGenerators}. The Successive Sweep Search algorithm (SSS) ranks second for these three distributions for exactly the same reason. For the curve function generator, this still holds true since the output size $k$ is still significantly smaller than $n^2$. Only for the linear function generator we observe that the increased number of calls to the range minimum oracle leads to higher running times of the SSS algorithm. However, the SSS approach is still faster than both ND algorithms. 

\begin{figure}[!ht]
    \includegraphics[width=0.495\textwidth]{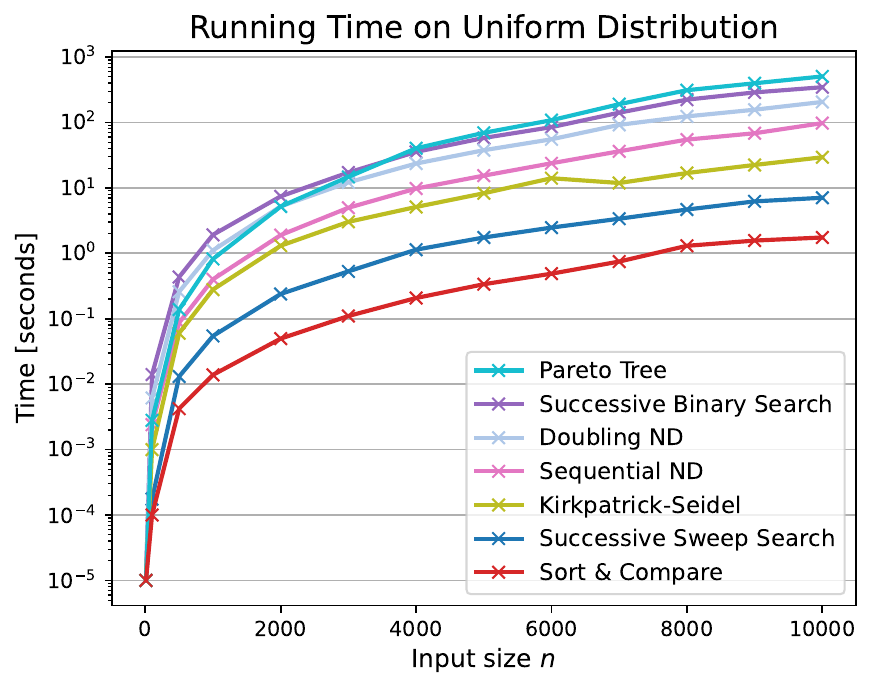}\hfill
    \includegraphics[width=0.495\textwidth]{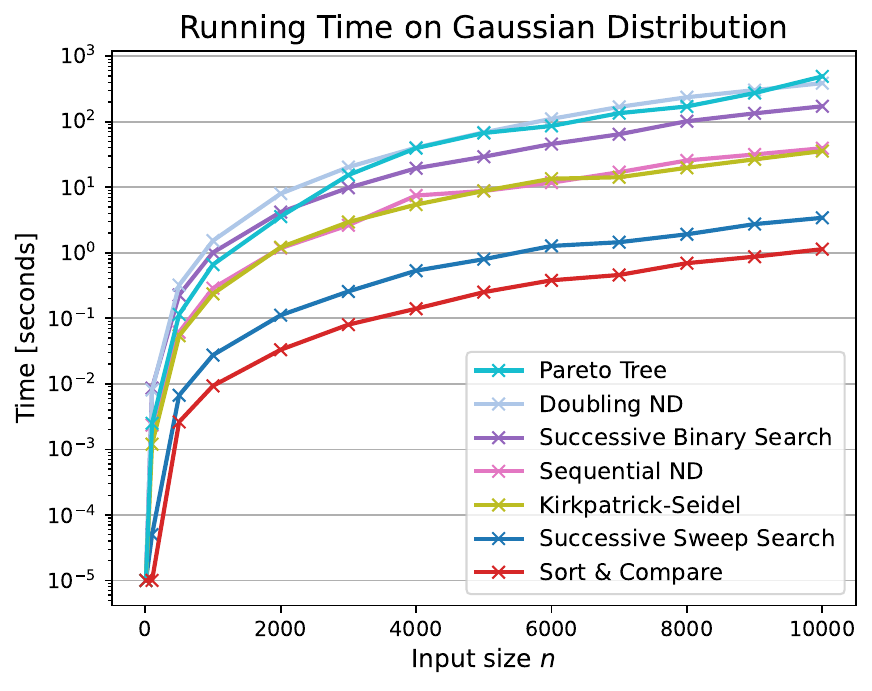}\\
    \includegraphics[width=0.495\textwidth]{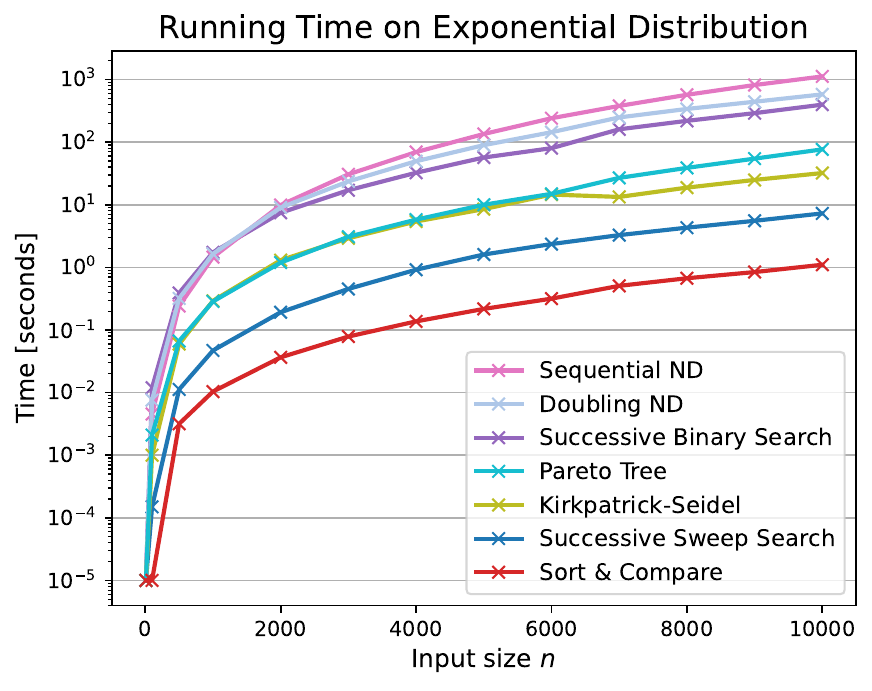}\hfill
    \includegraphics[width=0.495\textwidth]{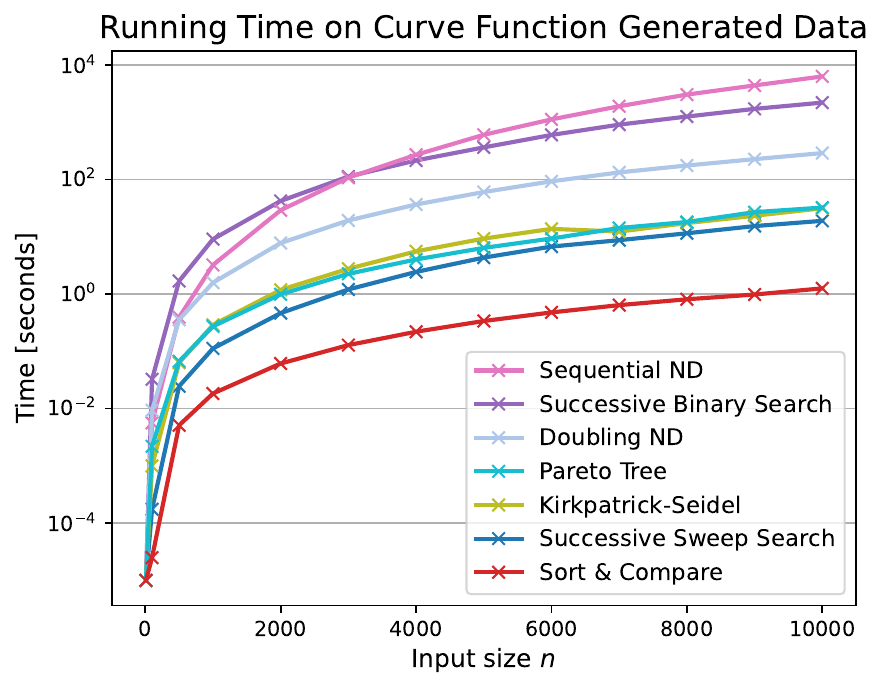}\\
    \includegraphics[width=0.495\textwidth]{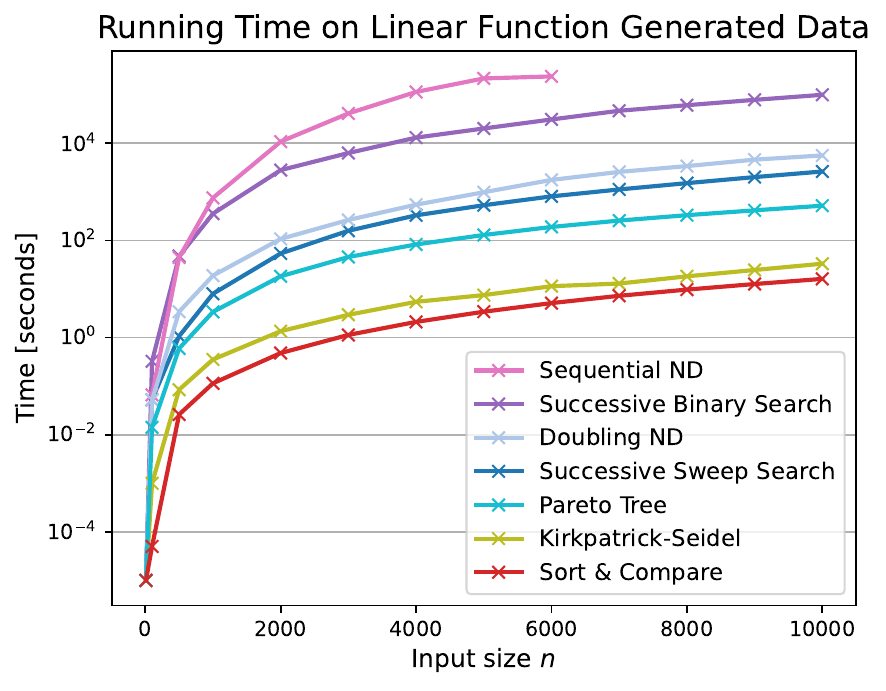}    
    \caption{Average running times of a selection of algorithms for a wider range of input sizes up to $n=10,000$ using the five input generators. }
    \label{fig:RuntimeMiddle}
\end{figure}

The Sort \& Compare algorithm consistently yields the fastest running times among all five input variations. It outperforms the best previously known algorithms, namely ND and KS, by up to four orders of magnitude. Since the SC algorithm always considers all points in the matrix regardless of their occurrence in the Pareto sum $C$, the running time of the algorithm is not influenced by the actual output size. Therefore, we conclude that the SC algorithm provides the most efficient running time of all proposed algorithms whilst ensuring an output-sensitive space consumption.

For larger input sizes depicted in Figure \ref{fig:RuntimeMiddle}, we observe that the Pareto tree algorithm performs significantly better than both variants of ND on exponential, curve and linear input generators. On input sizes above $n=6,000$ of the linear function generator, the SND algorithm exceeded the time limit and thus is not considered for higher values of $n$. 

\begin{figure}[!ht]
    \includegraphics[width=0.495\textwidth]{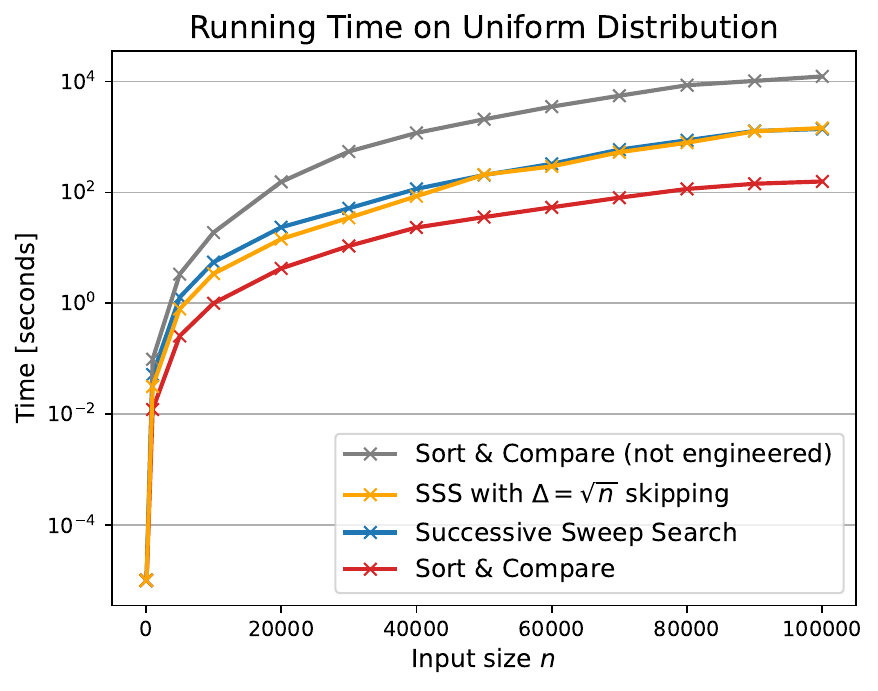}\hfill
    \includegraphics[width=0.495\textwidth]{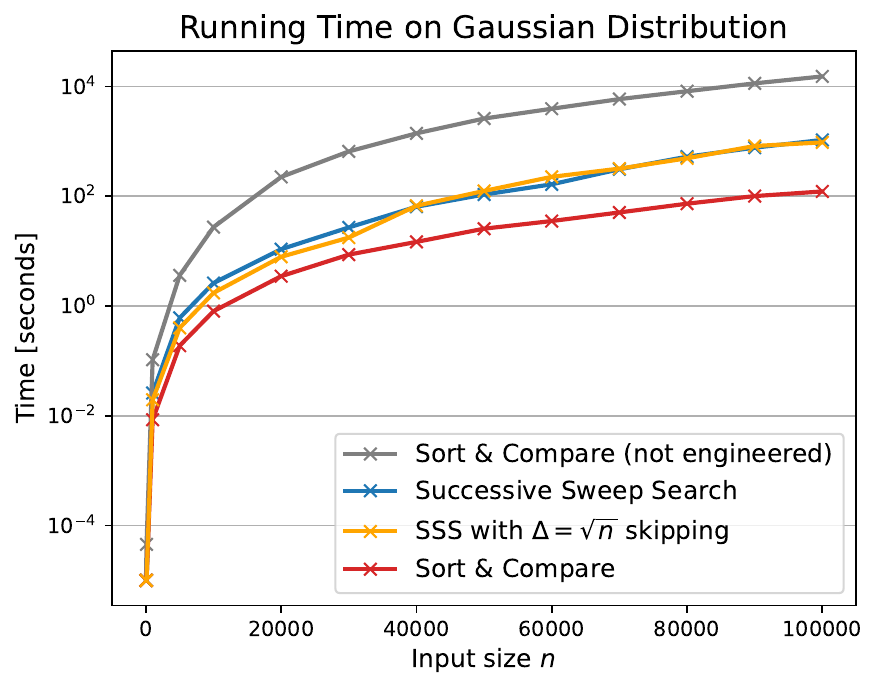}\\
    \includegraphics[width=0.495\textwidth]{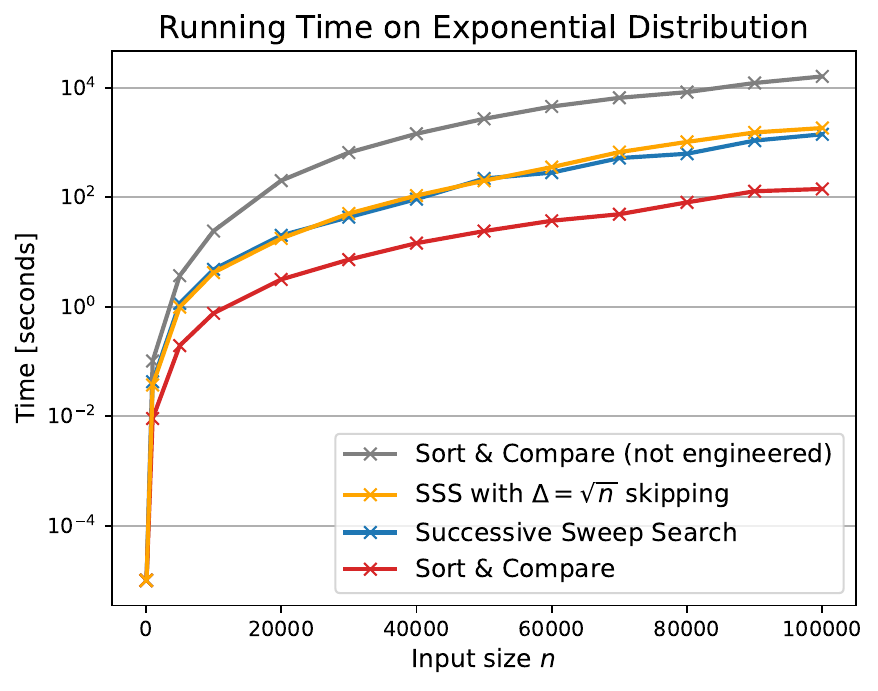}\hfill 
    \includegraphics[width=0.495\textwidth]{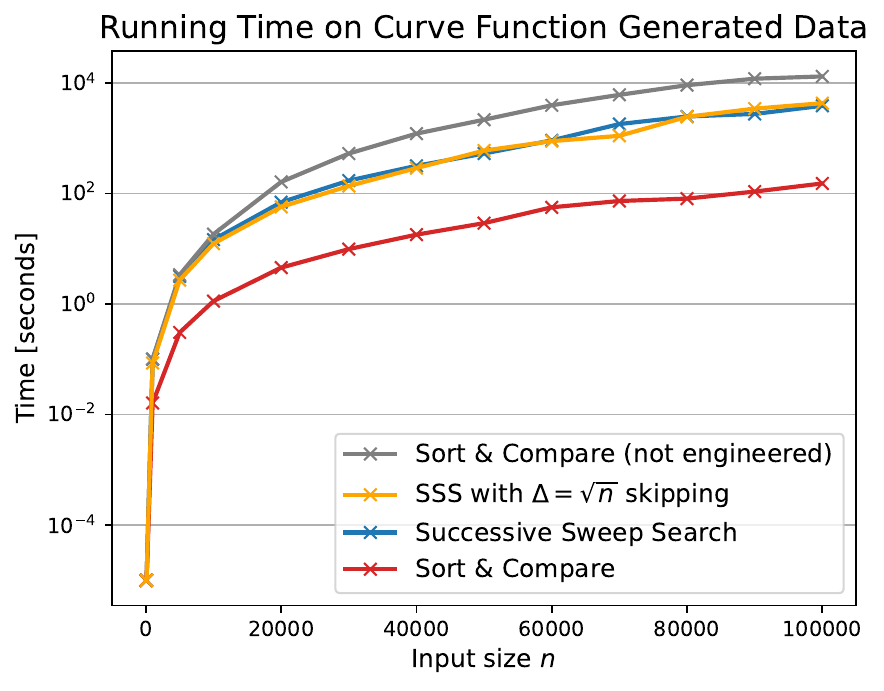}\\
    \includegraphics[width=0.495\textwidth]{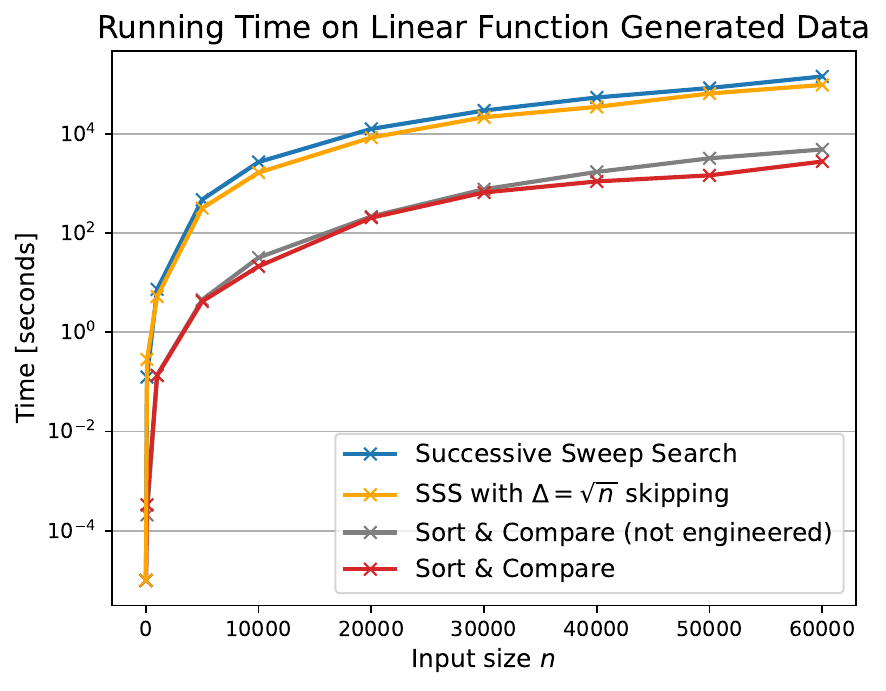} \hfill
    \includegraphics[width=0.495\textwidth]{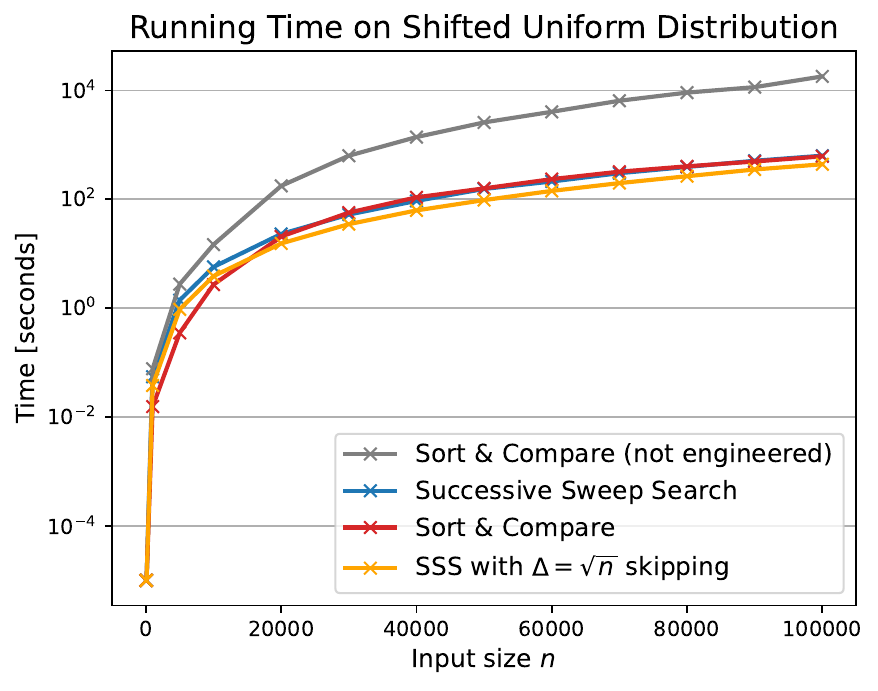}
    \caption{Average running times for the best algorithms Successive Sweep Search and Sort \& Compare for input sizes up to $n=100,000$ on six variations of generated input sets.}
    \label{fig:RuntimeBest}
\end{figure}

The running time of the best two algorithms SSS and SC are evaluated on instances of up to $n=100,000$ in Figure \ref{fig:RuntimeBest}. The Sort \& Compare algorithm exhibits the fastest running times on all input distributions. On the Gaussian distribution, SC yields its overall lowest running times. On average it requires $122$ seconds to compute the Pareto sum of two input Pareto sets each of size $n=100,000$. The SSS algorithm also shows its fastest running time on the Gaussian distribution. Computing the Pareto sum takes on average $1,065$ seconds for the SSS algorithm. On average, the SC algorithm is better than the SSS algorithm by one order of magnitude for all five generators. Both algorithms, the SC and SSS, have their slowest running times on the linear function generated inputs. Solving the largest instance reported for said generator and $n=60,000$ takes $2,500$ seconds for the SC and $142,000$ seconds for the SSS algorithm. In contrast, solving instances of the same input size on the Gaussian distributed input sets results in $35$ seconds for the SC and $165$ seconds for the SSS algorithm. The increase in running time is by a factor of $71$ for the SC and $860$ for the SSS algorithm. This observation is due to the effect described above that the running time of the Successive Sweep Search algorithm heavily depends on the number of Pareto points in the matrix and thus on the size $k$ of the final Pareto sum.

Further, we observe a significant speed-up by the engineering of the Sort \& Compare algorithm. We added the constraint, that only points that are not already dominated by the last point added to the Pareto sum are inserted into the min-heap for any given column. With this modification we obtain a speed-up in average running time of two orders of magnitude for the uniform, Gaussian, exponential and curve function generated input sets. For the linear function generators the engineering yields only minor speed-ups since virtually all points in the Minkowski sum $M$ are part of the Pareto sum $C$ and hence need to be added to the min-heap. Lastly, for the shifted uniform distribution, the engineering yields speed-ups up to a factor of 30 for the Sort \& Compare algorithm. Consequently, the engineered SC is always faster than the previous version of SC on all six input distributions Comparing the SSS to the SC algorithm, the previous variant of SC consistently requires more time than the SSS algorithm. 

The effects of the engineered version of SSS using the $\Delta$-skipping depend on the input distribution. Here, we modified the range-minimum oracle such that it checks if it might skip over $\Delta=\sqrt{n}$ points in one column and still remain within the range instead of performing upward linear search. The same is done for skipping over $\Delta$ indices in one row of the Minkowski sum. For the uniform, Gaussian, exponential and curve function generator, the running times for the original SSS and the SSS with the $\Delta$-skipping are fairly similar. As depicted in Figure \ref{fig:sortedGen}, the points in the matrix do not follow a specific pattern for these input distributions. Hence skipping over larger areas in the matrix is not feasible and the algorithm requires upward linear search in order to find the candidates for the range-minimum as for the original SSS version. From Figure \ref{fig:functionGen} we recall that the linear function generator yields the largest output sizes with Pareto sums containing effectively all points in the Minkowski sum. Consequently, the output size $k=n^2$ for these input sets. In general, the running time of the SSS algorithm depends on the number of calls to the range-minimum oracle which is exactly $k$. However, during one call to the oracle, we can efficiently improve the running time by skipping over $\Delta$ points in rows or columns during the left-to-right sweep of the matrix. Therefore, the engineered SSS algorithm provides a speed-up of up to 1.6 in comparison to the regular SSS. Lastly, we consider the Minkowski sum for inputs following the shifted uniform distribution from Figure \ref{fig:functionGen}. For these input sets, we can skip over large index ranges in the rows and columns which do not need to be considered for the range-minimum during one oracle call. The $\Delta$-skipping yields speed-ups of up to 1.6 compared to the regular SSS version. Further, on the shifted uniform distribution, the $\Delta$-skipping SSS is even faster than the Sort \& Compare algorithm for input sizes above $n=20,000$. 

\subsection{Real Data Sets}
As a real-world application of Pareto sum computation, we consider bi-criteria route planning in road networks. Here, given an input graph $G(V,E)$ and costs $c_1,c_2:E \rightarrow \mathbb{R}^+$, the goal is to either compute all Pareto-optimal paths with respect to $c_1,c_2$ between two nodes $s,t \in V$, or the path optimal with respect to one cost while not exceeding a budget on the other (also known as the constrained shortest path problem). To accelerate query answering,  a bi-criteria contraction hierarchy (BCH) data structure can be used. In the preprocessing phase of a BCH, the input graph is augmented with additional edges, also called shortcuts. The shortcut insertion is guided by a node permutation $\pi :V \rightarrow \{1,\dots,n\}$. For nodes $u,w \in V$, a shortcut $\{u,w\}$ is inserted if and only if there exists a simple path from $u$ to $w$ on which no node has a higher $\pi$ value than $\max(\pi(u),\pi(w))$. The shortcut  represents all simple paths $p$ between $u$ and $w$ with that property. For each Pareto-optimal $p$, the respective cost tuple $(c_1(p),c_2(p))$ should be assigned to the shortcut. To compute these Pareto sets for all shortcuts in an efficient manner, a bottom-up approach is used.  Let $u$ be the inner node on a path $p$ from $u$ to $w$ with maximum $\pi$-value. If the Pareto sets $A$ and $B$ of the shortcuts $\{u,v\}$ and $\{v,w\}$ are known, respectively, the Pareto set of $\{u,w\}$ is the Pareto sum $C$ of $A$ and $B$. If there are multiple paths $p$, the final Pareto set of $\{u,w\}$ is formed by the non-dominated elements of the union of all these Pareto sums. The non-dominated union of two Pareto sets  can be computed in linear time by merging the presorted sets to obtain the sorted union  and then applying the simple non-dominance check as described in the SC approach. In the final BCH, queries can be answered with a bi-directional Pareto-Dijkstra run that relaxes shortcut edges instead of many original edges whenever possible. This significantly reduces the search space and allows to answer queries orders of magnitude faster \cite{storandt2012route,funke2015personalized}.

In our experiments, we use test graphs extracted from OpenStreetMap with Euclidean distance and positive height difference as edge costs (in compliance with \cite{storandt2012route}). We comparatively evaluate four algorithms (KS, SC, DND and SSS) for Pareto sum computation on these data sets. For SSS, we used $\Delta$-skipping with $\Delta=\sqrt |A|$. Note that Pareto sets $A$ and $B$ do not necessarily have the same size here, but all proposed Pareto sum computation algorithms can be easily adapted. 
\begin{table}
   \caption{Experimental results for BCH computation on three road networks of different size. The table shows the input graph sizes, the number of edges in the augmented graph (original + shortcuts), the number of non-trivial Pareto sum computations, as well as the running times for conducting these computations with four different algorithms. 
    The final row shows the preprocessing time spent on operations other than  Pareto sum computation. }
    \label{tab:road}
\setlength\tabcolsep{2.2em}
    \begin{tabular}{|l|rrr|}
    \hline
                & ROAD1 & ROAD2 & ROAD3\\
    \hline
    \#Nodes     & 349479 & 1246440 & 3835238\\
    \#Edges      & 720363 & 2612260 & 8037228\\
    \hline
    \#BCH-edges & 1325259 & 4981957 & 15703653\\
 \#PS computations & 899390 & 4424857 & 48844050\\
 \hline 
    Kirkpatrick-Seidel  & 32.09 s& 1234.15  s& $>$24 h\\
    Sort \& Compare & 13.45 s & 587.77  s& 47374.67 s \\
    Doubling ND & 19.03 s &  454.42   s & 37447.61 s\\
    Successive Sweep Search & 4.83 s & 129.62  s & 9668.33  s\\
    \hline
    Additional preprocessing & 7.23  s& 54.66  s& 1237.58 s\\
    \hline
    \end{tabular}  

\end{table}
\begin{figure}
    \centering
    \includegraphics[width=0.49\textwidth]{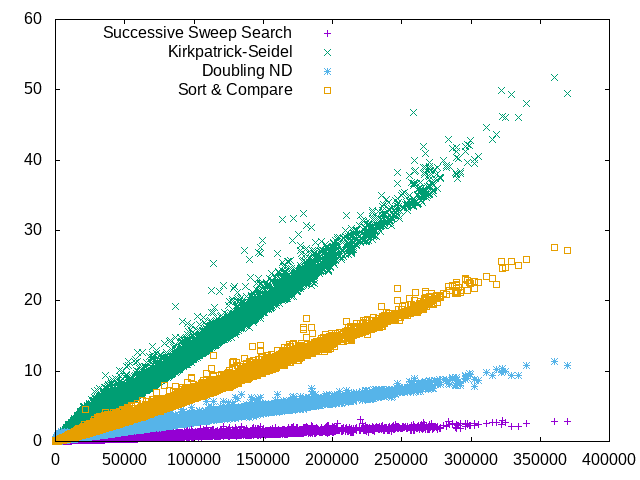}\hfill
    \includegraphics[width=0.49\textwidth]{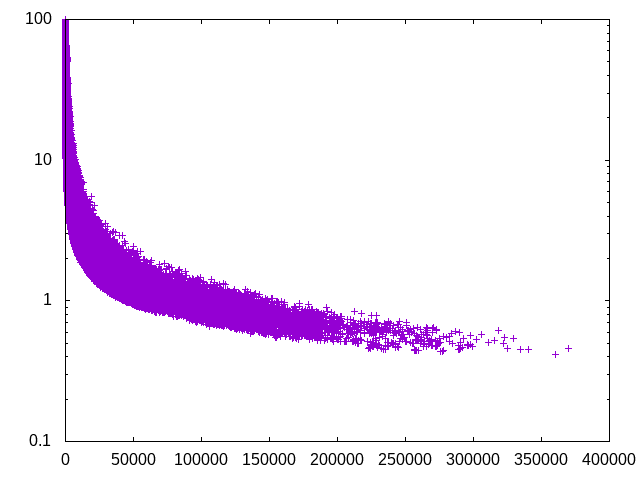}
    \caption{Detailed results for the ROAD2 instance. Left: Running times in seconds per Pareto sum computation in dependency of the size of the Minkowski sum. Right: Pareto sum size as percentage (logscale) of the size of the Minkowski sum.}
    \label{fig:road2}
\end{figure}
Table \ref{tab:road} shows the characteristics of the three road network instances we considered in our experiments and the outcomes. The number of Pareto sum computations  in the preprocessing phase of the BCH reported in the table excludes trivial inputs where either $A$ or $B$ has size 1. We observe that the time spent on non-trivial Pareto sum computations dominates the overall preprocessing time, especially on larger networks.
There are significant differences in running time between the algorithms we tested, though. On all instances, SSS is the fastest approach. It is roughly an order of magnitude faster than the KS algorithm which fully computes and stores $M$. With KS, we could not compute a BCH data structure within a day on our largest instance with about  4 million nodes. Interestingly, in contrast to the experiments on generated data, DND and SSS outperform SC.
Figure \ref{fig:road2} shows  the running times for all individual Pareto sum computations  as well as the size of the respective results for the ROAD2 instance. We observe that the larger the Minkowski sum $M$, the smaller the relative output size. This explains why SSS consistently outperforms the other approaches, especially on larger inputs. Figure \ref{fig:road3} shows results for the two best algorithms, DND and SSS, on  ROAD3. Here,   $|M|$ was up to $3 \cdot 10^6$ and the percentage of elements in the Pareto sum $C$ even approached 0.1.
 This is very beneficial for the SSS algorithm as the smaller the output size the fewer range minimum oracle calls are needed.

\begin{figure}
    \centering
    \includegraphics[width=0.49\textwidth]{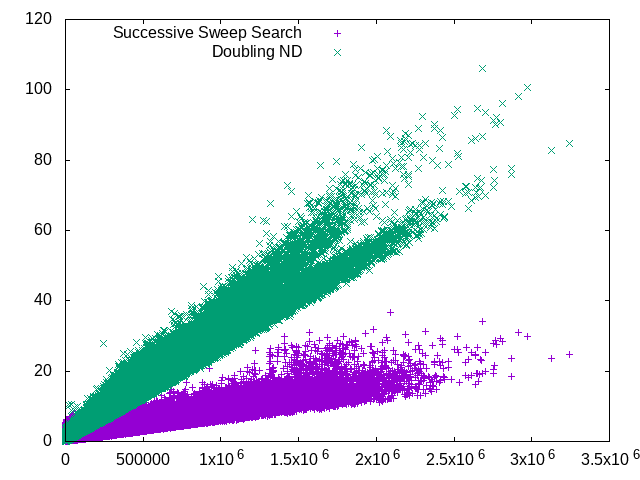}\hfill
    \includegraphics[width=0.49\textwidth]{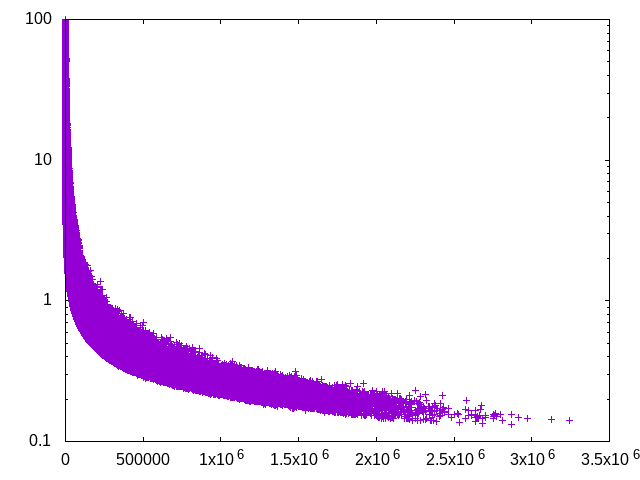}
    \caption{Detailed results for the ROAD3 instance. Left: Running times in seconds per Pareto sum computation in dependency of the size of the Minkowski sum. Right: Pareto sum size as percentage (logscale) of the size of the Minkowski sum.}
    \label{fig:road3}
\end{figure}

Furthermore, we used DND and SSS in query answering to combine   Pareto sets in the bi-directional Pareto-Dijkstra  run whenever the forward and the backward search meet. On the ROAD3 instance, a  speed-up of up to 5  over DND was achieved when using SSS.

\section{Conclusions and Future Work}
We introduced scalable algorithms for Pareto sum computation which avoid the computation of the whole Minkowski sum. Our engineered Sort \& Compare algorithm and our Successive Sweep Search algorithm were shown to perform best  while guaranteeing an output-sensitive space consumption.  Which algorithm is superior depends on the input point distribution as well as the output size. 
Sort \& Compare as well as Successive Sweep may be amenable to parallelization using a divide-and-conquer approach, leading to efficient algorithms with span $\mathcal{O}(n\log n)$. However, one has to be careful to preserve the output sensitivity of Successive Sweep.

On the theoretical side, it would be interesting to obtain (conditional) lower bounds for instances with sublinear output size. For instances with linear output size, we have shown that a subquadratic time algorithm is unlikely to exist. The Sort \& Compare algorithm almost matches the running time lower bound. It might be possible, though, to shave off the log factor.  
Perhaps more interesting is the question whether subquadratic algorithms exist for average-case inputs. In particular, smoothed analysis \cite{SpiTen01} might be interesting where we look at slightly perturbed worst-case inputs.

Another direction for future work is the consideration of higher-dimensional input points. While some algorithms are easily generalizable, novel range minimum oracles need to be designed for the successive algorithms to work. Furthermore, generalizing the lower bound to higher dimensions would be interesting as well.

\backmatter

\bmhead{Acknowledgments}
The conference version of this article appeared in ESA 2023 Track B: \url{https://drops.dagstuhl.de/entities/document/10.4230/LIPIcs.ESA.2023.60}

\bibliography{references}

\end{document}